\documentclass{CSML}

\def\dOi{13(1:11)2017}
\lmcsheading%
{\dOi}
{1--38}
{}
{}
{Apr.~28, 2016}
{Mar.~17, 2017}
{}

\usepackage{hyperref}
\usepackage[all]{xy}
\usepackage{amsfonts,amssymb}
\usepackage{oldlfont}
\usepackage{tikz}
\usepackage{multirow}
\usetikzlibrary{shapes}
\usetikzlibrary{arrows}
\usepackage{breakurl} 
\theoremstyle{plain}

\newtheorem{theorem}[thm]{Theorem}
\newtheorem{lemma}[thm]{Lemma}
\newtheorem{proposition}[thm]{Proposition}
\newtheorem{corollary}[thm]{Corollary}
\theoremstyle{definition}
\newtheorem{definition}[thm]{Definition}
\newtheorem{example}[thm]{Example}
\newtheorem{remark}[thm]{Remark}
\newenvironment{proof*}{\proof}{}
\newcommand{\BEGIN}{\sym{B}}
\newcommand{\END}{\sym{E}}
\newcommand{\REWRITE}{\sym{W}}
\newcommand{\GORIGHT}{\sym{R}}
\newcommand{\GUESS}{\sym{G}}
\newcommand{\ROTATE}{\sym{O}}
\newcommand{\CUT}{\sym{C}}
\newcommand{\MOVELEFT}{\sym{L}}
\newcommand{\WAIT}{\sym{S}}
\newcommand{\FINISH}{\sym{F}}
\newcommand{\FINISHTWO}{\sym{f}}

\newcommand{\Types}{{\mathbb{T}}}

\newcommand{\typset}[1]{{\mathcal{#1}}}
\newcommand{\typeA}{{\typset{A}}}
\newcommand{\typeAB}{{\typset{A\!B}}}
\newcommand{\typeB}{{\typset{B}}}
\newcommand{\typeC}{{\typset{C}}}
\newcommand{\typeD}{{\typset{D}}}
\newcommand{\typeE}{{\typset{E}}}
\newcommand{\typeM}{{\typset{M}}}
\newcommand{\typeCopyA}{{\overline{\typeA}}}
\newcommand{\typeConst}{{\typset{K}}}
\newcommand{\typeT}{{\typset{T}}}
\newcommand{\sym}[1]{{\ensuremath{{\mathsf{#1}}}}}
\newcommand{\Redex}[1]{{\setlength{\fboxsep}{0pt}\colorbox{black!10!white}{$\vphantom{\overline{X}\underline{X}X_1^1}$#1}}}
\usepackage{xcolor} 
\usepackage{colortbl}
\usepackage{stmaryrd,amssymb} 
\usepackage{microtype}%if unwanted, comment out or use option "draft"
\usepackage{verbatim}
\usepackage{multirow}
\newcommand\TTTT{{T\kern-0.2em\raisebox{-0.3em}T\kern-0.2emT\kern-0.2em\raisebox{-0.3em}2}}
\newcommand{\aprove}{\sym{AProVE}}
\newcommand{\nat}{\mbox{\sf I\hspace{-0.6mm}N}}
\newcommand{\desda}{\; \Longleftrightarrow \;}

\newcommand{\tr}{\mbox{\sf tr}}

\newcommand{\spl}{\sym{split}}
\newcommand{\li}{\langle}
\newcommand{\ri}{\rangle}

\newcommand{\trans}[2]{{\llparenthesis #1 \rrparenthesis}_{_{\!\!\,\tiny #2\!}}}
\newcommand{\transA}[1]{\trans{#1}{A}}
\newcommand{\toweak}{\mathrel{{\to}\hspace*{-0.5pt}{\texttt{\raisebox{-0.5pt}{=}}}}}
\newcommand{\rotatesrs}{{\mathrm{rot}}(\Sigma_A)}
\newcommand{\torot}{\to_{\rotatesrs}}
\newcommand{\cycto}{\mathrel{{\circ}\!{\rightarrow}}}
\newcommand{\rto}[1]{\cycto_{#1}}
\newcommand{\tosplit}{\to_{\sym{split}(R)}}
\newcommand{\torotate}[1][R]{\to_{\sym{rotate}(#1)}}
\newcommand{\toshift}{\to_{\sym{shift}(R)}}
\newcommand{\shift}{\curvearrowright}

\newcommand{\prefixrewrite}{\hookrightarrow} %\mathrel{{\rule{.35pt}{2.55pt}}\hspace*{-.909pt}{\rightarrow}}}
\newcommand{\suffixrewrite}{\mathrel{\text{\rotatebox[origin=c]{180}{\reflectbox{$\hookrightarrow$}}}}}
\newcommand{\map}{\Phi}
\newcommand{\mapShift}{\map_{S}}
\newcommand{\mapSplit}{\map}
\newcommand{\mapRotate}{\map_{R}}
\newcommand{\src}{{\ensuremath{\mathsf{source}}}}
\newcommand{\tgt}{{\ensuremath{\mathsf{target}}}}
\newcommand{\dec}{{\ensuremath{\mathsf{Dec}}}}
\newcommand{\?}{\hspace{0.5pt}}

\begin{document}
%%%%%%%%%%%%%%%%%%%%%%%%%%%%%%%%%%%%%%%%%%%%%%%%%%%%%%%%%%%%%%%%%%%%%%%%%%%%%%
%% title, authors, affiliations
\title[Termination of Cycle Rewriting]{Termination of Cycle Rewriting by Transformation and Matrix Interpretation}
%%%%%%%%%%%%%%%%%
\author[D.~Sabel]{David Sabel\rsuper a}
\address{{\lsuper a}Goethe-University Frankfurt,
Department of Computer Science and Mathematics,
Computer Science Institute,
60629 Frankfurt am Main, Germany}
\email{sabel@ki.informatik.uni-frankfurt.de}
\thanks{{\lsuper a}The first author is supported by the Deutsche Forschungsgemeinschaft (DFG)  under grant SA 2908/3-1.}
%%%%%%%%%%%%%%%%%
\author[H.~Zantema]{Hans Zantema\rsuper b}
\address{{\lsuper b}TU Eindhoven, Department of Computer Science\\
P.O.\ Box 513, 5600 MB Eindhoven, The Netherlands\linebreak
Radboud University Nijmegen, Institute for Computing and Information Sciences\\
P.O.\ Box 9010, 6500 GL Nijmegen, The Netherlands}
\email{h.zantema@tue.nl}
%%%%%%%%%%%%%%%%%%%%%%%%%%%%%%%%%%%%%%%%%%%%%%%%%%%%%%%%%%%%%%%%%%%%%%%%%%%%%%
%% mandatory lists of keywords and classifications:
\keywords{rewriting systems, string rewriting, termination, relative termination}%mandatory: Please provide 1-5 keywords
%%%%%%%%%%%%%%%%%
\subjclass{F.4.2 Grammars and other rewriting systems}
%%%%%%%%%%%%%%%%%
\begin{abstract}
We present techniques to prove termination of cycle rewriting, that is, string rewriting on cycles, which are strings in which the start and end are connected. 
Our main technique is to transform cycle rewriting into string rewriting and then apply state of the art techniques to prove termination of the string rewrite system. 
We present three such transformations, and prove for all of them that they are sound and complete. In this
way not only termination of string rewriting of the transformed system implies termination of the original
cycle rewrite system, a similar conclusion can be drawn for non-termination. 

Apart from this transformational approach, we present a uniform framework of matrix interpretations,
covering most of the earlier approaches to automatically proving termination of cycle rewriting.

All our techniques serve both for proving termination and relative termination. 

We present several experiments showing the power of our techniques.
\end{abstract}
\maketitle

\section{Introduction}
Cycles can be seen as strings of which the left end is connected to the right end, 
by which the string has no left end or right end any more. In Fig.~\ref{fig:example-intro}
a pictorial representation of two such cycles is shown. 

String rewriting can not only be applied on strings, but also on cycles. 
Applying string rewriting on cycles, {i.e.}~replacing a substring of a cycle by another substring, 
is briefly called cycle rewriting. 
For instance, applying the string rewrite rule $a\?a\?a \to a\?b\?a\?b\?a$ 
to the cycle in Fig.~\ref{fig:example-intro} (a) results in the cycle shown Fig.~\ref{fig:example-intro} (b).

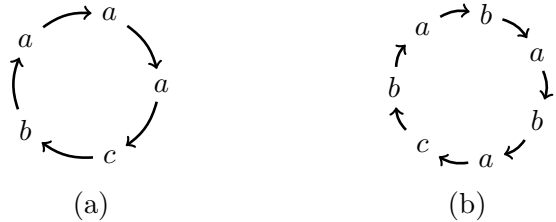
\begin{figure}[htp]
\begin{tikzpicture}
\node at (0,0){ 
\begin{tikzpicture}
\node (a1) at (1*360/5:1cm) {$a$};
\node (b1) at (2*360/5:1cm) {$a$};
\node (b2) at (3*360/5:1cm) {$b$};
\node (c1) at (4*360/5:1cm) {$c$}; 
\node (a2) at (5*360/5:1cm) {$a$};
\draw[->,line width=1pt,bend left=20] (b1) to node {} (a1);
\draw[->,line width=1pt,bend left=20] (b2) to node {} (b1);
\draw[->,line width=1pt,bend left=20] (c1) to node {} (b2);
\draw[->,line width=1pt,bend left=20] (a2) to node {} (c1);
\draw[->,line width=1pt,bend left=20] (a1) to node {} (a2);
\end{tikzpicture}
};
\node at (5,0){
\begin{tikzpicture}
\node (a1) at (0.5*360/7:1cm) {$a$};
\node (bn1) at (1.5*360/7:1cm) {$b$};
\node (b1) at (2.5*360/7:1cm) {$a$};
\node (b2) at (3.5*360/7:1cm) {$b$};
\node (c1) at (4.5*360/7:1cm) {$c$};
\node (a2) at (5.5*360/7:1cm) {$a$};
\node (bn2) at (6.5*360/7:1cm) {$b$};
\draw[->,line width=1pt,bend left=20] (bn1) to node {} (a1);
\draw[->,line width=1pt,bend left=20] (b1) to node {} (bn1);
\draw[->,line width=1pt,bend left=20] (b2) to node {} (b1);
\draw[->,line width=1pt,bend left=20] (c1) to node {} (b2);
\draw[->,line width=1pt,bend left=20] (a2) to node {} (c1);
\draw[->,line width=1pt,bend left=20] (bn2) to node {} (a2);
\draw[->,line width=1pt,bend left=20] (a1) to node {} (bn2);
\end{tikzpicture}
};
\node at (0,-1.6) {(a)};
\node at (5,-1.6) {(b)};
\end{tikzpicture}
\caption{Illustration of Cycles and Cycle Rewriting\label{fig:example-intro}}
\end{figure}

Rewriting behavior is strongly influenced by allowing cycles, for instance, in string rewriting the
single rule $a\?b \to b\?a$ is terminating, but in cycle rewriting it is not, since the string $a\?b$
represents the same cycle as $b\?a$. 

In many areas cycle rewriting is more natural than string rewriting.
For instance, the problem of 5 dining philosophers can be expressed as a cycle
${F}\?{T}\?{F}\?{T}\?{F}\?{T}\?{F}\?{T}\?{F}\?{T}$
where $F$ denotes a fork, and $T$ denotes a thinking philosopher. 
Writing $L$ for a philosopher who has picked up her left fork, but not her right fork, and $E$ for an eating philosopher, 
a classical (deadlocking) modeling of the dining philosophers problem (for arbitrary many philosophers)
can be expressed by the cycle rewrite system consisting of the rules
${T}\?{F} \to {L}$, ${F}\?{L} \to {E}$, ${E} \to {F}\?{T}\?{F}$.
As a cycle rewrite system this is clearly not terminating.

Also from the viewpoint of graph transformation, cycle rewriting is very natural. 
For instance, in \cite{BKZ14} it was shown that if all rules of a graph transformation system
are string rewrite rules, termination of the transformation system coincides with termination
of the cycle rewrite system, and not with termination of the string rewrite system. 
Developing techniques to prove cycle termination also helps to understand and develop
new techniques for proving termination of graph transformation systems (see {e.g.}~\cite{BKNZ15}).

So both string rewriting and cycle rewriting provide natural semantics for string rewrite 
systems, also called semi-Thue systems.  Historically, string rewriting got a lot of attention as being a 
particular case of term rewriting, while cycle rewriting hardly got any attention until recently.
In 2015 automated proving of cycle termination became a new category in the Termination Competition 
\cite{GMRTW:15,termcomp}. In 2015 two tools participated, and in 2016 three tools participated in the category.

In \cite{ZKB14} a first investigation of termination of cycle rewriting was made. Some techniques were presented to
prove cycle termination, implemented in a tool {\tt torpacyc}. Further a transformation $\phi$ was given
such that for every string rewriting system (SRS) $R$, string termination of $R$ holds if and only if cycle
termination of $\phi(R)$ holds. As a consequence, cycle termination is undecidable. 

However, for making
use of the strong power of current tools for proving termination of string rewriting in order to prove
cycle termination, a transformation the other way around is needed: transformations $\psi$ such that for
every SRS $R$, cycle termination of $R$ holds if and only if string termination of $\psi(R)$ holds.  
The `if' direction in this `if and only if' is called `sound', the `only if' is called complete. This
implies a way 
to prove cycle termination of an SRS $R$: apply a tool for proving termination of string rewriting
to $\psi(R)$ for a sound transformation $\psi$. Conversely, a way to prove cycle non-termination of $R$ is to
prove non-termination of $\psi(R)$ for a complete transformation $\psi$.
The main topic of this paper is to investigate such transformations, and to exploit them to prove termination 
of cycle rewriting, or non-termination, or relative (non-)termination. Here relative termination deals
with two rewrite systems: relative termination means that every infinite reduction of the union of them
contains only finitely many steps of one of them. In detail we give fully worked-out 
proofs of soundness and completeness  for three
approaches to transform cycle termination into string termination (called $\sym{split}$, $\sym{rotate}$, and $\sym{shift}$),
and also for relative termination. 

Using transformations to exploit the power of tools for termination of term rewriting to prove a modified
property was used before in \cite{GZ03,GM04}. However, there the typical observation was that the complete
transformations were complicated, and for non-trivial examples, termination of $\psi(R)$ could not be
proved by the tools, while for much simpler sound  (but incomplete)
transformations $\psi$, termination of $\psi(R)$ could
often be proved by the tools. In our current setting this is different: 
one of our introduced transformations, the
transformation $\spl$, for which we prove that it is sound and complete, we show that for several
systems $R$ for which all approaches from \cite{ZKB14} fail, cycle termination of $R$ can be concluded 
from an automatic termination proof of $\spl(R)$ generated by \aprove~\cite{FBEFFOPSKSST:14,aprove} 
or \TTTT~\cite{KSZM09,ttt2}.

It can be shown that if strings of size $n$ exist admitting cycle reductions in which for every 
rule the number of applications of that rule is more than linear in $n$, then all techniques from
\cite{ZKB14} fail to prove cycle termination. Nevertheless, in quite simple examples this may occur while 
cycle termination holds. As an example consider the following. 

A number of people are in a circle, and each of them carries a number, represented in binary notation with a bounded number of bits. 
Each of them may increase his/her number by one, as long as it fits in the bounded number of bits.
Apart from that, every person may increase the number of bits of the number of its right neighbor by two. In order to avoid trivial non-termination, the latter is only allowed if the leading bit of the number is 0, and the new leading bit is put to 1, and the other to 0, by which effectively one extra bit is added. 
We will prove that this process will always terminate by giving an SRS in which all of the above steps can be described by a number of cycle rewrite steps, and prove cycle termination. 
In order to do so we write $P$ for person, and 0 and 1 for the bits of the binary number. 
For carry handling we introduce an extra symbol $c$ of which the meaning is a 1 with a carry.
Assume for every person its number is stored left from it. So if the number ends in 0, by adding one this last bit 0 is replaced by 1, expressed by the
rule $0\?P \to 1\?P$. 
In case the number ends in 1, a carry should be created, since $c$ represents a  1 with a carry this is expressed by the rule $1\?P \to c\?P$.
Next the carry should be processed. In case it is preceded by 0, this 0 should be replaced by 1, while the $c$ is replaced by 0; this is expressed by the rule $0\?c \to 1\?0$.
In case it is preceded by 1, a new carry should be created while again the old carry is replaced by 0; this is expressed by the rule $1\?c \to c\?0$.  In this way adding one to any number in binary notation can be expressed by a number of rewrite steps, as long as no overflow occurs.
Finally, we have to add a rule to extend the bit size of the number of the right neighbor: the leading bit should be 0, while it is replaced by $1\?0\?0$: adding two extra bits of which the leading one is 1 and the other is 0. 
This is expressed by the rule $P\?0 \to P\?1\?0\?0$.
Summarizing: we have to prove cycle termination of the SRS consisting of the five rules 
$$0\?P \to 1\?P,\; 1\?P \to c\?P,\; 0\?c \to 1\?0, \; 1\?c \to c\?0, \; P\?0 \to P\?1\?0\?0.$$
This is fundamentally impossible by the techniques presented in \cite{ZKB14}: by one of the techniques the last rule can be removed, but starting in $0^n\?P$ a reduction can be made in which all of the remaining four rules are applied an exponential number of times, by which the techniques from \cite{ZKB14} fail.

In this paper we give two ways to automatically prove that cycle termination holds for the above example
$R$: \TTTT~succeeds in proving termination of $\spl(R)$, and the other is a variant of matrix
interpretations for which we show that it proves cycle termination. The latter is another main topic of
this paper: we give a self-contained uniform presentation of trace-decreasing matrix interpretations for 
cycle termination that covers the tropical and arctic variant from \cite{ZKB14} and the natural variant from 
\cite{sabel-zantema:15}, now also serving for relative termination. In this way we cover all current known
techniques for proving cycle termination except for match bounds (\cite{ZKB14}).

The paper is organized as follows.
Section~\ref{secprel} recalls the basics of cycle rewriting and introduces relative cycle termination.
The main section Section~\ref{sectransf} first adapts type introduction \cite{zantema:94} for string rewriting and relative termination (Theorem~\ref{theo:typeintro-rel}),
and then the three transformations $\sym{split}$, $\sym{rotate}$, and $\sym{shift}$ are presented and soundness and completeness of all of them is proved 
(Theorems~\ref{thm:split-sound-and-complete}, \ref{theo:shift-sound-and-complete}, and \ref{theo:rotate-sound-and-complete}).
Also adapted transformations for relative cycle transformation are presented and their soundness and completeness is shown
(Theorems~\ref{thm:split-sound-and-complete-relative}, \ref{theo:shift-sound-and-complete-relative}, and \ref{theo:rotate-sound-and-complete-relative}).
In Section~\ref{secmatr} trace-decreasing matrix interpretations are presented,
and we show how they can be used to prove cycle termination (Theorem~\ref{thmmatr}) and relative cycle termination (Theorem~\ref{thmmatrrel})
for the three instances of natural, tropical and arctic matrix interpretations.
In Section~\ref{secexp} experiments on implementations of our techniques are reported. We conclude in Section~\ref{secconcl}.

Compared to our paper \cite{sabel-zantema:15},  this paper contains full proofs, extends all the termination techniques 
to relative termination, and presents a framework for matrix interpretations which covers  and
extends the previous approaches presented in \cite{ZKB14,sabel-zantema:15}.

\section{Preliminaries}
\label{secprel}
In this section we briefly recall the required notions for string and cycle rewriting.

A {\em signature} $\Sigma$ is a finite alphabet of symbols. With $\Sigma^*$ we denote the set of strings over $\Sigma$. 
With $\varepsilon$ we denote the empty string and for $u,v \in \Sigma^*$, we write $u\?v$ for the concatenation of the strings $u$ and $v$.
With $|u|$ we denote the length of string $u \in \Sigma^*$ and for $a \in \Sigma$ and $n \in \nat$, 
$a^n$ denotes $n$ replications of symbol $a$, {i.e.}~$a^0 = \varepsilon$ and $a^i = a\? a^{i-1}$ for $i >0$.

Given a binary relation $\to$, we write $\to^i$ for $i$ steps, $\to^{\leq i}$ for at most $i$ steps, 
$\to^{< i}$ for at most $i-1$ steps, $\to^*$ for the reflexive-transitive closure of $\to$, and $\to^+$ 
for the transitive closure of $\to$. For binary relations $\to_1$ and $\to_2$, we write $\to_1 . \to_2$ for the
composition of $\to_1$ and $\to_2$, {i.e.}~$a \to_1 . \to_2 c$ iff there exists a $b$ {s.t.}~$a \to_1 b \text{ and } b \to_2 c$.

\subsection{String Rewriting}
A {\em string rewrite system} (SRS) is a finite set $R$ of rules $\ell \to r$ where $\ell,r \in \Sigma^*$ and $\ell \not= \varepsilon$. 
The {\em rewrite relation} ${\to_R} \subseteq {(\Sigma^* \times \Sigma^*)}$ is defined as follows: 
if $w = u\?\ell\?v \in \Sigma^*$ and $(\ell \to r) \in R$, then $w \to_R u\?r\?v$. 
The {\em prefix-rewrite relation} $\prefixrewrite_R$ is defined as: 
if $w = \ell\?u \in \Sigma^*$ and $(\ell \to r) \in R$, then $w \prefixrewrite_R r\?u$.
The {\em suffix-rewrite relation} $\suffixrewrite_R$ is defined as:
if $w = u\?\ell \in \Sigma^*$ and $(\ell \to r) \in R$, then $w \suffixrewrite_R u\? r$.

A (finite or infinite) sequence of rewrite steps $w_1 \to_R w_2 \to_R \cdots$ is called a {\em rewrite sequence} 
(sometimes also a {\em reduction} or a {\em derivation}). 
For an SRS $R$, the rewrite relation $\to_R$ is called {\em non-terminating} if there exists a string $w \in \Sigma^*$ and an 
infinite rewrite sequence $w \to_R w_1 \to_R \cdots$. Otherwise, $\to_R$ is {\em terminating}.
If $\to_R$ is terminating (non-terminating, resp.) we also say $R$ is \emph{string terminating} (string non-terminating, resp.).

\subsection{Cycle Rewriting}
We recall the notion of cycle rewriting from \cite{ZKB14}. A string can be viewed as a cycle, 
{i.e.}~the last symbol of the string is connected to the first symbol. To represent cycles by strings, we define the
equivalence relation $\sim$ as follows:
$$
u \sim v \text{ iff  } u = w_1\?w_2 \text { and }  v = w_2\?w_1 \text{ for some strings }w_1,w_2 \in \Sigma^*
$$
With $[u]$ we denote the equivalence class of string $u$ {w.r.t.}~$\sim$.

The {\em cycle rewrite relation} ${\cycto_R} \subseteq {({\Sigma/{\sim}} \times {\Sigma/{\sim}})}$ of an 
SRS $R$ is defined as 
$$
[u] \cycto_R [v]  \text{ iff } \exists w\in\Sigma^*: \ell\?w \in [u], (\ell \to r) \in R,~\text{and}~r\?w\in[v] 
$$

The cycle rewrite relation $\cycto_R$ is called {\em non-terminating} iff there exists a string 
$w \in \Sigma^*$ and an infinite sequence 
$[w] \cycto_R [w_1] \cycto_R [w_2] \cycto_R \cdots$. 
Otherwise, $\cycto_R$ is called {\em terminating}. 
If $\cycto_R$ is terminating (non-terminating, resp.) we also say that $R$ is \emph{cycle terminating}
(cycle non-terminating, resp.).

We recall some known facts about cycle rewriting.
\begin{proposition}[see \cite{ZKB14}]\label{prop:properties-rta14}
Let $\Sigma$ be a signature, $R$ be an SRS, and $u,v \in \Sigma^*$. 
\begin{enumerate}
\item If $u \to_R v$ then $[u] \cycto_R [v]$. 
\item If $\cycto_R$ is terminating, then $\to_R$ is terminating. 
\item Termination of $\to_R$ does not necessarily imply termination of $\cycto_R$.
\item Termination of $\cycto_R$ is undecidable.
\item For every SRS $R$ there exists a transformed SRS $\phi(R)$ {s.t.}~the following three properties are equivalent:
 \begin{itemize}
 \item $\to_R$ is terminating.
 \item $\to_{\phi(R)}$ is terminating.
 \item $\cycto_{\phi(R)}$ is terminating.
 \end{itemize}
\end{enumerate}
\end{proposition}

\noindent For an SRS $R$, the last property implies that termination of $\to_R$ can be proved by proving termination of the translated cycle rewrite relation $\cycto_{\phi(R)}$. 
In \cite{ZKB14} it was used to show that termination of cycle rewriting is undecidable and for further results on derivational complexity for cycle rewriting.

\subsection{Relative Termination}
We will also consider relative termination of cycle and string rewrite systems.
Hence, in this section we recall the definition of relative termination of string rewrite systems (see {e.g.}~\cite{geser1990})
and introduce relative cycle termination. 
\begin{definition}
Let ${\to_1 \subseteq \to_2} \subseteq {(O \times O)}$
be binary relations on a set $O$.
We say $\to_1$ is {\em terminating relative} to $\to_2$
iff every infinite sequence $o_1 \to_{2} o_2 \to_{2}  \cdots$ contains only
finitely many $\to_1$-steps.  

Let $S \subseteq R$ be string rewrite systems over an  alphabet $\Sigma$.
If the string rewrite relation $\to_S$ is terminating relative to 
the string rewrite relation $\to_R$ then we say
$S$ is \emph{string terminating relative to} $R$.
If the cycle rewrite relation $\cycto_S$ is terminating relative to 
the cycle rewrite relation $\cycto_R$ then we say
$S$ is \emph{cycle terminating relative to} $R$.

For $S \subseteq R$,
we sometimes call $R \setminus S$ the weak rules and $S$ the strict rules.
We often write rules in $R \setminus S$ as $\ell \toweak r$ and rules in $S$ 
as $\ell \to r$.
\end{definition}
Note that for all SRSs $R$ and $S$ the identities
${\to_{R \cup S}} = {\to_{R} \cup \to_{S}}$ and
${\cycto_{R \cup S}} = {\cycto_{R} \cup \cycto_{S}}$
hold, which we will sometimes use.

Since every string rewrite step, is also a cycle rewrite step (on the equivalence class {w.r.t.}~$\sim$),
any infinite string rewrite sequence giving evidence for string non-termination, also gives evidence
for cycle non-termination. Thus we have:
\begin{corollary}
If $S$ is cycle terminating relative to $R$, then
$S$ is also string terminating relative to $R$.
\end{corollary}
However, as for usual termination, relative cycle termination 
is different from relative string termination:
\begin{example}
Let $S = \{a\?b \to c\?a\}$ and $R= S \cup \{c \toweak b\}$.
Then $S$ is string terminating relative to $R$ (which is easy to prove), 
while $S$ is not cycle terminating relative to $R$, since
the infinite (looping) cycle rewrite sequence
$[a\?b] \cycto_S [c\?a] \cycto_R [a\?b] \ldots$ contains
infinitely many $\cycto_S$ steps.
\end{example}

The following proposition provides several characterizations for relative termination. We formulate it in a general form (for all binary relations). Even though its proof is quite standard,
we include it for the sake of completeness.
\begin{proposition}\label{prop:rel-prop}
Let $\to_1 \subseteq \to_2$ be
binary relations on a set $O$.
The following five propositions are equivalent:
\begin{enumerate}
 \item\label{fact1} The relation $\to_2^*.\to_1.\to_2^*$ is terminating.
 \item\label{fact2a} The relation $\to_1.\to_2^*$ is terminating.
 \item\label{fact2} The relation $\to_2^*.\to_1$ is 
 terminating.
 \item\label{fact4} $\to_1$ is terminating relative to $\to_2$.
\end{enumerate}
\end{proposition}
\begin{proof}
We show the claim by a chain of implications:
\begin{description}
\item[\rm$\eqref{fact1}\implies\eqref{fact2a}$] 
Clearly ${\to_1.\to_2^*} \subseteq {\to_2^*.\to_1.\to_2^*}$,
and thus non-termination of $\to_1.\to_2^*$ implies non-termination of $\to_2^*.\to_1.\to_2^*$.
\item[\rm$\eqref{fact2a}\implies\eqref{fact2}$]
Assume $\to_2^*.\to_1$ is non-terminating.
Then there exists an infinite sequence s.t.
$o_{i,1} \to_2^* o_{i,2} \to_1 o_{i+1,1}$ for $i=1,2,\ldots$.
This implies $o_{i,2} \to_1 o_{i+1,1} \to_2^* o_{i+1,2}$
for all $i=1,2,\ldots$
and thus $\to_1.\to_2^*$ is non-terminating.
\item[$\eqref{fact2} \implies \eqref{fact4}$:]
We show $\neg \eqref{fact4} \implies \neg \eqref{fact2}$
Thus, we assume that 
there exists a sequence $o_1 \to_{2} o_2 \to_{2} \cdots$ s.t.
for every $n$ there exists a number $m_n \geq n$ {s.t.}~$o_{m_n} \to_1 o_{m_n+1}$.
Assume that each number $m_n$ is minimal {w.r.t.}~$n$.
Then the given sequence can be written as 
$o'_{j} \to_2^{k_j} o'_{j+k_j}
\to_1 o'_{j+k_j+1} = o'_{{j+1}}$
{s.t.}~$o'_0 = o_1$ and $k_j \geq 0$ for $j=0,1,\ldots$.
Since $o'_{j} \to_2^*.\to_{1} o'_{j+1}$, this shows that $\to_2^*.\to_1$ is non-terminating.

\item[\rm$\eqref{fact4} \implies \eqref{fact1}$:]
If $\to_2^*.\to_1.\to_2^*$ is non-terminating, then the infinite rewrite sequence consisting
of infinitely many $\to_2^*.\to_1.\to_2^*$-steps contains infinitely many $\to_1$-steps.
\qedhere
\end{description}
\end{proof}

\noindent For SRSs $S \subseteq R$ , the previous proposition can be instantiated
with $\to_1~:=~\to_S$ and $\to_2~:=~\to_R$ to derive characterizations 
of relative string termination, and with $\to_1~:=~\cycto_S$ and 
$\to_2~:=~\cycto_R$ to derive characterizations of relative cycle termination.

\section{Transforming Cycle Termination into String Termination}
\label{sectransf}
The criteria given in Proposition~\ref{prop:properties-rta14} and the involved transformation $\phi$, 
which transforms string rewriting into cycle rewriting,
provide a method to prove string termination by proving cycle termination.
However, it does not provide a method to prove termination of the cycle rewrite relation $\cycto_{R}$ by proving 
termination of the string rewrite relations $\to_R$ or $\to_{\phi(R)}$. Hence,  in this section we develop transformations 
$\psi$ {s.t.}~termination of $\to_{\psi(R)}$ implies termination of $\cycto_R$. 
We call such a transformation $\psi$ {\em sound}. However, there are ``useless'' sound transformations, for instance, transformations where $\psi(R)$ is always non-terminating. 
So at least one wants to find sound transformations which permit to prove termination of non-trivial cycle rewrite relations.
However, a better transformation should fulfill the stronger property that  $\to_{\psi(R)}$ is terminating if and only if $\cycto_R$ is terminating. 
If termination of $\cycto_R$ implies termination of $\to_{\psi(R)}$, then we say $\psi$ is {\em complete}. 
For instance, for a complete transformation, non-termination proofs of $\to_{\psi(R)}$ also imply non-termination of $\cycto_R$.
Hence, our goal is to find sound and complete transformations $\psi$.

Besides such transformations, we will consider transformations 
$\psi_{\mathit{rel}}(S,R)= (S',R')$ which are sound and complete for 
relative termination, i.e.~transformations  which transform $(S,R)$ with $S \subseteq R$ , such that
cycle termination of $S$ relative to $R$ holds, if and only if $S'$ is string terminating relative to $R'$.

We will introduce and discuss three transformations $\sym{split}$, $\sym{rotate}$, and  $\sym{shift}$
where the most important one is the transformation $\sym{split}$, since it has the following properties:
The transformation is sound and complete, and as our experimental results show, it behaves well in practice when proving termination of cycle rewriting.
The other two transformations $\sym{rotate}$ and $\sym{shift}$ are also sound and complete, but rather complex and -- as our experimental results show --
they do not behave as well as the transformation $\sym{split}$ in practice. 
We include all three transformations in this paper to document some different approaches to transform cycle rewriting into string rewriting. 
We also consider variations of the three transformations for relative termination.
We show that all three variations are sound and complete transformations for relative termination. 

Since our completeness proofs, use type introduction \cite{zantema:94}, we recall this technique in Section~\ref{sec:types} focused on
typed string rewriting only, and prove a (novel) theorem that type introduction is correct for relative string termination.
In the remaining sections~\ref{sec:split}, \ref{subsec:shift}, and \ref{subsec:rotate} we successively introduce and treat the three transformations.
\subsection{Type Introduction}\label{sec:types}
The technique of type introduction was presented in \cite{zantema:94}, for termination of term 
rewriting. Here we are only interested in string rewriting (being the special case of
term rewriting having only unary symbols), but for our purpose need to extend this result 
to relative termination.

An signature $\Sigma$ is {\em typed} if there is a set $\Types$ of types (also called sorts),
and every $a \in \Sigma$ has a {\em source type} $\tau_1 \in \Types$ and  a {\em target type} 
$\tau_2 \in \Types$, notation $a : \tau_1 \to \tau_2$, $\tau_1 = \src(a)$, $\tau_2 = \tgt(a)$.

For a non-empty string $w \in \Sigma^+$ its target $\tgt(w)$ is defined to be the target of 
its first element; its source $\src(w)$ is defined to be the source of its last element.
A non-empty string $w \in \Sigma^+$ is called {\em well-typed} if either it is in
$\Sigma$, or it is of the shape $au$ for $a \in \Sigma$ and $u \in \Sigma^+$ is
well-typed, and $\src(a) = \tgt(u)$.

An SRS is called {\em well-typed} if for every rule $\ell \to r$ we either have
\begin{itemize}
\item $r = \varepsilon$ and $\src(\ell) = \tgt(\ell)$, or
\item $r \neq \varepsilon$ and $\src(\ell) = \src(r)$ and $\tgt(\ell) = \tgt(r)$.
\end{itemize}
The following lemma is straightforward.
\begin{lem}
If $R$ is a well-typed SRS over $\Sigma$ and $w \to_R w'$ for $w \in \Sigma^+$ being
well-typed, then $w'$ is well-typed too.
\end{lem}

So in an infinite reduction with respect to a  well-typed SRS, all strings are well-typed
if and only if the initial string is well-typed.

For a well-typed  SRS $R$, we say that $R$ is {\em string terminating in the
typed setting} if there does not exist an infinite $\to_R$-reduction consisting of well-typed strings.

For well-typed  SRSs $S \subseteq R$ we say that $S$ is {\em string terminating relative} to $R$ {\em in the
typed setting} if every infinite $\to_R$-reduction consisting of well-typed strings contains
only finitely many $\to_S$-steps.
 
The main theorem, to be exploited several times in this paper, states that this notion of
relative termination in the typed setting is equivalent to the notion of
relative termination in the general setting without typing requirements.

In order to prove this theorem we need a notion of decomposition of (possibly untyped)
strings and a lemma stating some key properties of this decomposition. We denote a string
consisting of $n$ strings $u_1,\ldots,u_n$ by $[u_1,\ldots,u_n]$. The decomposition
$\dec(u)$ of a string $u \in \Sigma^+$ is such a string of strings and is defined as follows:
\begin{itemize}
\item $\dec(a) = [a]$,
\item if $u \in \Sigma^+$ and $\dec(u) = [u_1,\ldots,u_n]$, then 
\begin{itemize}
\item $\dec(au) = [au_1,\ldots,u_n]$ if $\src(a) = \tgt(u_1)$, and
\item $\dec(au) = [a,u_1,\ldots,u_n]$ if $\src(a) \neq \tgt(u_1)$.
\end{itemize}
\end{itemize}
By construction for any $u \in \Sigma^+$ with $\dec(u) = [u_1,\ldots,u_n]$ we have the 
following properties:
\begin{itemize}
\item $u_i$ is well-typed for $i = 1,\ldots,n$;
\item $u = u_1 \ldots u_n$;
\item if $v$ is a well-typed substring of $u$, then it also a substring of $u_i$ for some 
$i = 1,\ldots,n$;
\end{itemize}
As we consider well-typed SRSs only, with non-empty left hand sides, every rewrite step
applied on such $u$ applies to
one of the corresponding $u_i$. In case of a collapsing rule, that is, a rule with
empty right hand side, it may be the case that a type clash is removed, decreasing the length
$|\dec(u)|$ of $\dec(u)$. For instance, for $a : \tau_1 \to \tau_1$, $b : \tau_2 \to \tau_2$,
we have $baabab \to_R baaab$ for $R = \{b \to \varepsilon\}$, while $\dec(baabab) = [b,aa,b,a,b]$ has length 5 and
$\dec(baaab) = [b,aaa,b]$ has length 3. In all other cases the rewriting of $u$ takes place in one 
of the elements of $\dec(u)$, while all other elements remain unchanged. These
observations are summarized in the following lemma.

\begin{lem}
\label{lemtyp}
Let $R$ be a well-typed SRS  with non-empty left hand sides and $u \to_R v$. Then
\begin{itemize}
\item $|\dec(u)| \geq |\dec(v)|$, and 
\item if $|\dec(u)| = |\dec(v)|$, then there exists $i \in \{1,\ldots,n\}$ such that
$u_i \to_R v_i$, and $u_j = v_j$ for $j \neq i$, where
$\dec(u) = [u_1,\ldots,u_n]$ and $\dec(v) = [v_1,\ldots,v_n]$.
\end{itemize}
\end{lem}

\noindent Now we are prepared for the main theorem. The way it is used is as follows: for proving (relative)
termination, try to find a typing such that the SRS is well-typed. Then according to the theorem
the infinite reduction for which a contradiction has to be derived, may be assumed to be
well-typed.

\begin{thm}\label{theo:typeintro-rel}
Let $S \subseteq R$ be well-typed SRSs with non-empty left hand sides. Then  $S$ is 
string terminating relative to $R$ if and
only if $S$ is string terminating relative to $R$ in the typed setting.
\end{thm}
\begin{proof}
The `only if'-part is trivial. For the `if'-part assume we have an infinite $R$-reduction
$u_1 \to_R u_2 \to_R u_3 \to_R \cdots$, not
well-typed; we have to prove it contains only finitely many $S$-steps.
According to the first claim of Lemma \ref{lemtyp} there exist $n,N$ such $|\dec(u_i)| =
|\dec(u_{i+1})| = n$ for all $i \geq N$. Write $\dec(u_i) = [u_{i1},\ldots,u_{in}]$ for $i \geq
N$. According to the second part of Lemma \ref{lemtyp} for every $j = 1,\ldots,n$ we have
either $u_{ij} = u_{i+1,j}$ or  $u_{ij} \to_R u_{i+1,j}$ for $i \geq N$. If  $u_{ij} \to_R u_{i+1,j}$ occurs infinitely often
this yields a well-type infinite $R$-reduction, containing only finitely many $S$-steps since
$S$ is terminating relative to $R$ in the typed setting. If $u_{ij} \to_R u_{i+1,j}$ occurs
finitely often, it also contains only finitely  many $S$-steps. As the number of $S$-steps in 
the finite part $u_1 \to_R^* u_N$ is finite too, we conclude that the total number of 
$S$-steps in the original reduction is finite.
\end{proof}

By instantiating the previous theorem with $S=R$ we obtain correctness of type introduction for string termination:
\begin{corollary}\label{cor:typeintro-term}
Let $R$ be a well-typed SRS. Then $R$ is string terminating if and only if 
$R$ is string terminating in the typed setting.
\end{corollary}

\subsection{The Transformation Split}\label{sec:split}
The idea of the transformation $\sym{split}$ is to perform a single cycle rewrite step  $[u] \cycto_R [v]$ which
uses rule $(\ell \to r) \in R$, by either applying a string rewrite step $u \to_R v$ or 
by splitting the rule $(\ell \to r)$ into two rules $(\ell_p \to r_p)$ and $(\ell_s \to r_s)$, 
where $\ell = \ell_p\?\ell_s$ and $r = r_p\?r_s$. Then a cycle rewrite step can be simulated by 
a prefix and a subsequent suffix rewrite step: first apply rule $\ell_s \to r_s$ to a prefix of $u$  and 
then apply rule $\ell_p \to r_p$ to a suffix of the obtained string. 

\begin{example}\label{example-simple}
Let $R = \{a\?b\?c \to b\?b\?b\?b\}$ and  $[b\?c\?d\?d\?a] \cycto_R [b\?b\?d\?d\?b\?b]$. The rule $a\?b\?c \to b\?b\?b\?b$ can be split into the rules $a \to b\?b$ and $b\?c \to b\?b$ {s.t.}~$b\?c\?d\?d\?a \prefixrewrite_{\{b\?c \to b\?b\}} b\?b\?d\?d\?a \suffixrewrite_{\{a \to b\?b\}} b\?b\?d\?d\?b\?b$.
\end{example}
We describe the idea of the transformation $\sym{split}$ more formally. 
It uses the following observation of cycle rewriting:
if $[u] \cycto_R [v]$, then $u \sim \ell\?w$, $(\ell \to r) \in R$, and $v \sim r w$. 
From $u \sim \ell\?w$ follows that $u = u_1\?u_2$  and $\ell\?w = u_2\?u_1$ for some $u_1,u_2$. We consider the cases for $u_2$: 

\begin{enumerate}
 \item\label{enum-splitexpl1} If $u_i = \varepsilon$ (for $i=1$ or $i=2$), then $u = \ell\?w$ and  
$u \prefixrewrite_R r\?w$ 
by a prefix string rewrite step.
 \item\label{enum-splitexpl2} If $\ell$ is a prefix of $u_2$, {i.e.}~$\ell\?u_2' = u_2$, then $w = u_2'\?u_1$, $u = u_1\?\ell\?u_2' \to_R u_1\?r\?u_2'$, and $u_1\?r\?u_2' \sim r\?w$.
 \item\label{enum-splitexpl3} If $u_2$ is a proper prefix of $\ell$, then there exist $\ell_p, \ell_s$ with $\ell = \ell_p\?\ell_s$ {s.t.}~$u_2 = \ell_p$ and $\ell_s$ is a 
 non-empty prefix of $u_1$, {i.e.}~$u_1 = \ell_s\?w$ and $u = u_1\?u_2 = \ell_s\?w\?\ell_p \prefixrewrite_{\{\ell_s\to r_s\}}r_s\?w\?\ell_p \suffixrewrite_{\{\ell_p \to r_p\}} r_s\?w\?r_p \sim r\?w$ if $r_p\?r_s = r$.
\end{enumerate}

\noindent The three cases show that a cycle rewrite step $[u] \cycto_{\{\ell \to r\}} [v]$ can either be performed by applying a string rewrite step $u \to_{\{\ell \to r\}} v'$ where $v' \sim v$ (cases \ref{enum-splitexpl1} and \ref{enum-splitexpl2}) or in case \ref{enum-splitexpl3} by
 splitting $\ell \to r$ into two rules $\ell_p \to r_p$ and $\ell_s \to r_s$ such that $u \prefixrewrite_{\{\ell_s \to r_s\}} u'$ replaces a prefix of $u$ by $r_s$ and
$u' \suffixrewrite_{\{\ell_p \to r_p\}} v'$ replaces a suffix  of $u'$ by $r_p$ {s.t.}~$v' \sim v$.

For splitting a rule $(\ell \to r)$ into rules $\ell_p \to r_p$ and $\ell_s \to r_s$, we may choose any decomposition of $r$ for $r_p$ and $r_s$ ({s.t.}~$r=r_p\,r_s$).
In the following, we will work with $r_p = r$ and $r_s = \varepsilon$.

The above cases for cycle rewriting show that a {\em sound} transformation of the cycle rewrite relation  $\cycto_R$ 
into a string rewrite relation is the SRS which consists of all rules of $R$ and all pairs of rules $\ell_s \to \varepsilon$ and $\ell_p \to r$ for all $(\ell \to r) \in R$ and all $\ell_p,\ell_s$ with $|\ell_p|> 0$, $|\ell_s|> 0$, and $\ell = \ell_p\?\ell_s$.  
However, this transformation does not ensure that the rules evolved by splitting are used as prefix and suffix rewrite steps only.
Indeed, the transformation in this form is useless for nearly all cases, since whenever the right-hand side $r$ of a rule $(\ell \to r) \in R$ contains a symbol $a\in\Sigma$ which is the first or the last symbol in $\ell$, then the transformed SRS is non-terminating. 
For instance, for $R = \{a\?a \to a\?b\?a\}$ the cycle rewrite relation $\cycto_R$ is terminating, while the rule $a \to a\?b\?a$ (which would be generated by splitting the left-hand side of the rule) leads to non-termination of the string rewrite relation. Note that this also holds if we choose any other decomposition of the right-hand side.
Hence, in our transformation we introduce additional symbols to ensure:

\begin{itemize}
 \item $\ell_s \to \varepsilon$ can only be applied to a prefix of the string.
 \item $\ell_p \to r$ can only be applied to a suffix of the string.
 \item If $\ell_s \to \varepsilon$ is applied to a prefix, then also $\ell_p \to r$ must be applied, 
in a synchronized manner ({i.e.}~no other rule $\ell'_B \to \varepsilon$ or $\ell'_A \to r'$ can be applied in between).
\end{itemize}

\noindent In detail, we will prepend the fresh symbol $\sym{B}$ to the beginning of the string, 
and append the fresh symbol $\sym{E}$ to the end of the string. 
These symbols guarantee, 
that prefix rewrite steps $\ell\?u \prefixrewrite_{(\ell \to r)} r\?u$ can be expressed 
with usual string rewrite rules by replacing the left hand side $\ell$ with $\sym{B}\ell$ and 
analogous for suffix rewrite steps $u\?\ell \suffixrewrite_{(\ell \to r)} u\?r$ 
by replacing the left hand side $\ell$ with $\ell\?\sym{E}$.

Let $(\ell_i \to r_i)$ be the $i^{\text{th}}$ rule of the SRS which is split into two rules $\ell_s \to \varepsilon$ and $\ell_p \to r_i$,
where $\ell_p\?\ell_s = \ell_i$. After applying the rule 
$\ell_s \to \varepsilon$ to a prefix of the string, the symbol $\sym{B}$ will be replaced
by the two fresh symbols $\sym{W}$ (for ``wait'')
and $\sym{R}_{i,j}$ where $i$ represents the $i^{\text{th}}$ rule 
 and $j$ means that $\ell_i$ has been split after $j$ symbols ({i.e.}~$|\ell_p| = j$). 
The fresh symbol $\sym{L}$ is used to signal that the suffix has been rewritten by rule $\ell_p \to r$.
Finally, we use a copy of the alphabet, to ensure completeness of the transformation:
for an alphabet $\Sigma$, we denote by $\overline{\Sigma}$ a fresh copy of $\Sigma$, {i.e.}~$\overline{\Sigma} = \{\overline{a} \mid a \in \Sigma\}$. For a word $w \in \Sigma^*$ with $\overline{w}\in\overline{\Sigma}^*$, we denote the word $w$ where every symbol $a$ is replaced by $\overline{a}$. Analogously, for a word $w \in \overline{\Sigma}^*$ with $\underline{w}\in\Sigma$, we denote $w$ where every symbol $\overline{a}$  is replaced by the symbol $a$.

\begin{definition}[The transformation $\sym{split}$] 
Let $R = \{\ell_1 \to r_1, \ldots, \ell_n \to r_n\}$ be an SRS over alphabet $\Sigma$. Let $\overline{\Sigma}$ be a fresh copy of $\Sigma$ and let
 $\sym{B}, \sym{E}, \sym{W}, \sym{R}_{i,j}, \sym{L}$  be fresh symbols (fresh for $\Sigma \cup \overline{\Sigma}$). The SRS $\sym{split}(R)$ over alphabet $\Sigma \cup \overline{\Sigma} \cup \{\sym{B},\sym{E},\sym{L},\sym{W}\} \cup \bigcup_{i=1}^n\{ \sym{R}_{i,j} \mid 1 \leq j < |\ell_i|\}$ consists of the following string rewrite rules: 
\begin{flushleft}
\begin{minipage}{.5\textwidth}
\begin{align}
  \ell_i                  &\to r_i &&\hspace*{-1.1ex}\text{for all $(\ell_i \to r_i) \in R$} \label{split-1}\tag{splitA} 
\\
  \overline{a} \? \sym{L} &\to \sym{L} \? a&&\hspace*{-1.1ex}\text{for all $a \in \Sigma$} \label{split-2}\tag{splitB}
\\
   \sym{W} \? \sym{L} &\to \sym{B} \label{split-3}\tag{splitC}
   \\
\sym{R}_{i,j} \? a &\to \overline{a} \? \sym{R}_{i,j}&&\hspace*{-1.1ex}\text{for all $a \in \Sigma$} \label{split-6}\tag{splitD}
\end{align}
\end{minipage}\qquad
\begin{minipage}{.44\textwidth}
\begin{align}
 \intertext{for all $(\ell_i \to r_i) \in R$ and all $1 \leq j < |\ell_i|$ and
    $\ell_p \? \ell_s= \ell_i$ with $|\ell_p| = j$:}
      \sym{B} \? \ell_s &\to \sym{W} \? \sym{R}_{i,j} \label{split-4}\tag{splitE}
\\
\sym{R}_{i,j} \? \ell_p \? \sym{E} &\to \sym{L} \? r_i \? \sym{E} \label{split-5}\tag{splitF}
\end{align}
\end{minipage}
\end{flushleft}
\end{definition}
We describe the intended use of the rules and the extra symbols. The symbols $\sym{B}$ and $\sym{E}$ mark the start and the end of the string,
i.e.~for a cycle $[u]$ the SRS $\sym{split}(R)$ rewrites $\sym{B}\?u\?\sym{E}$.

Let $[u] \cycto_R [w]$. The rule \eqref{split-1} covers the case that also $u \to_R w$ holds.
Now assume that for $w' \sim w$ we have $u \prefixrewrite_{\{\ell_s \to \varepsilon\}} v \suffixrewrite_{\{\ell_p \to r\}} w'$  (where $(\ell_p\ell_s \to r)\in R$).
Rule \eqref{split-4} performs the prefix rewrite step and replaces $\sym{B}$ by $\sym{W}$ to ensure that no other such a rule can be applied. 
Additionally, the symbol $\sym{R}_{i,j}$ corresponding to the rule and its splitting is added to
ensure that only the right suffix rewrite step is applicable. Rule \eqref{split-6} moves the symbol $\sym{R}_{i,j}$ to right
and rule \eqref{split-5} performs the suffix rewrite step. Rules \eqref{split-2} and \eqref{split-3} are used to finish the simulation
of the cycle rewrite step by using the symbol $\sym{L}$ to restore the original alphabet and to finally replace $\sym{W}\?\sym{L}$ by
$\sym{B}$.

\begin{example}
For $R_1 = \{a \? a \to a \? b \? a\}$ the transformed string rewrite system $\sym{split}(R_1)$ is:
$$\begin{array}{l@{~}c@{~}l@{~}l@{\qquad}l@{~}c@{~}l@{~}l@{\qquad}l@{~}c@{~}l@{~}l}
a\?a &\to& a\?b\?a & \eqref{split-1}&
\overline{a}\?\sym{L} &\to& \sym{L}\?a &\eqref{split-2}&
\overline{b}\?\sym{L} &\to& \sym{L}\?b &\eqref{split-2}\\
\sym{W}\?\sym{L} &\to& \sym{B}&\eqref{split-3}&
\sym{B}\?a\?&\to&\sym{W}\?\sym{R}_{1,1}&\eqref{split-4}&
\sym{R}_{1,1}\?a\?\sym{E} &\to& \sym{L}\?a\?b\?a\?\sym{E}&\eqref{split-5}\\
\sym{R}_{1,1}\?a &\to& \overline{a}\?\sym{R}_{1,1}&\eqref{split-6}&
\sym{R}_{1,1}\?b &\to& \overline{b}\?\sym{R}_{1,1}&\eqref{split-6}\\
\end{array}$$
For instance, the cycle rewrite step $[a\?b\?a] \cycto_{R_1} [b\?a\?b\?a]$
is simulated in the transformed system by the following sequence of string rewrite steps (where the redex is highlighted by a gray background):
$$
\Redex{$\sym{B}\?a$}\?b\?a\?\sym{E}
\xrightarrow{\!\tiny\ref{split-4}\!\!}
\sym{W}\?\Redex{$\sym{R}_{1,1}\?b\?$}a\?\sym{E}
\xrightarrow{\!\tiny\ref{split-6}\!\!}
\sym{W}\?\overline{b}\?\Redex{$\sym{R}_{1,1}\?a\?\sym{E}$}
\xrightarrow{\!\tiny\ref{split-5}\!\!}
\sym{W}\?\Redex{$\overline{b}\?\sym{L}\?$}a\?b\?a\?\sym{E}
\xrightarrow{\!\tiny\ref{split-2}\!\!}
\Redex{$\sym{W}\?\sym{L}\?$}b\?a\?b\?a\?\sym{E}
\xrightarrow{\!\tiny\ref{split-3}\!\!}
\sym{B}\?b\?a\?b\?a\?\sym{E}.$$

As a further example, for the system $R_2 = \{a\?b\?c \to c\?b\?a\?c\?b\?a,~a\?a \to a\}$, the transformed string rewrite system $\sym{split}(R_2)$ is:
$$
\begin{array}{@{}l@{~}c@{~}l@{~}l@{\quad}l@{~}c@{~}l@{~}l@{\quad}l@{~}c@{~}l@{~}l@{}}
a\?b\?c &\to& c\?b\?a\?c\?b\?a &\eqref{split-1} &
a\?a &\to& a\?&\eqref{split-1}&
\sym{W}\?\sym{L} &\to& \sym{B}     &\eqref{split-3}\\
\overline{a}\?\sym{L}&\to& \sym{L}\?a &\eqref{split-2}&
\overline{b}\?\sym{L} &\to& \sym{L}\?b &\eqref{split-2}&
\overline{c}\?\sym{L} &\to& \sym{L}\?c &\eqref{split-2}\\
\sym{B}\?b\?c &\to& \sym{W}\?\sym{R}_{1,1}&\eqref{split-4}&
\sym{B}\?c &\to& \sym{W}\?\sym{R}_{1,2}&\eqref{split-4}&
\sym{B}\?a &\to& \sym{W}\?\sym{R}_{2,1}&\eqref{split-4}\\
\sym{R}_{1,1}\?a\?\sym{E} &\to& \sym{L}\?c\?b\?a\?c\?b\?a\?\sym{E} &\eqref{split-5}&
\sym{R}_{1,2}\?a\?b\?\sym{E} &\to& \sym{L}\?c\?b\?a\?c\?b\?a\?\sym{E} &\eqref{split-5}&
\sym{R}_{2,1}\?a\?\sym{E} &\to& \sym{L}\?a\?\sym{E}&\eqref{split-5}\\
\sym{R}_{1,1}\?a &\to& \overline{a}\?\sym{R}_{1,1}&\eqref{split-6}&
\sym{R}_{1,2}\?a &\to& \overline{a}\?\sym{R}_{1,2} &\eqref{split-6}&
\sym{R}_{2,1}\?a &\to& \overline{a}\?\sym{R}_{2,1} &\eqref{split-6}\\
\sym{R}_{1,1}\?b &\to& \overline{b}\?\sym{R}_{1,1} &\eqref{split-6}&
\sym{R}_{1,2}\?b &\to& \overline{b}\?\sym{R}_{1,2} &\eqref{split-6}&
\sym{R}_{2,1}\?b &\to& \overline{b}\?\sym{R}_{2,1} &\eqref{split-6}\\
\sym{R}_{1,1}\?c &\to& \overline{c}\?\sym{R}_{1,1} &\eqref{split-6}&
\sym{R}_{1,2}\?c &\to& \overline{c}\?\sym{R}_{1,2} &\eqref{split-6}&
\sym{R}_{2,1}\?c &\to& \overline{c}\?\sym{R}_{2,1} &\eqref{split-6}\\
\end{array}$$
Termination of $\sym{split}(R_1)$ and $\sym{split}(R_2)$ can be proved by \aprove~and \TTTT.
\end{example}

\begin{proposition}[Soundness of $\sym{split}$]\label{prop:split-sound}
 If $\tosplit$ is terminating then $\cycto_R$ is terminating.
\end{proposition}
\begin{proof}
By construction of $\sym{split}(R)$, it holds that if $[u] \cycto_R [v]$, then $\sym{B}\?u'\?\sym{E} 
\tosplit^+ \sym{B}\?v'\?\sym{E}$ with $u \sim u'$ and $v \sim v'$. 
Thus, for every infinite sequence 
$[w_1] \cycto_R [w_2] \cycto_R \cdots$, 
 there exists an infinite sequence  
$\sym{B}~w_1'~\sym{E} \tosplit \sym{B}~w_2'~\sym{E} \tosplit \cdots$ with $w_i \sim w_i'$ for all $i$.
\end{proof}

\subsubsection{Completeness of Split}
We use type introduction for string rewriting (see Section~\ref{sec:types}) and use the set of types
 $\Types := \{\typeA, \typeCopyA, \typeConst, \typeT\}$ and type the symbols used by $\sym{split}(R)$ as follows:
 $$
\begin{array}{c@{\qquad\qquad}c@{\qquad\qquad}c}
\begin{array}{l@{~}ll}
 \sym{L}   & : \typeA \to \typeCopyA
\\
 \sym{B}   & : \typeA \to \typeT
 \\&
\end{array}
&
\begin{array}{l@{~}ll}
 \sym{W}   & : \typeCopyA \to \typeT
 \\
 \sym{E}   & : \typeConst \to \typeA
\\&
\end{array}
&\begin{array}{l@{~}ll}
 a          & : \typeA \to \typeA & \text{for all $a \in \Sigma$}
\\
 \overline{a}          & : \typeCopyA \to \typeCopyA & \text{for all $\overline{a} \in \overline{\Sigma}$}
\\
 \sym{R}_{i,j} & : \typeA \to \typeCopyA &  \text{for all $\sym{R}_{i,j}$}
\end{array}
\end{array}
$$
First one can verify that $\sym{split}(R)$ is a typed SRS, i.e.~the left hand sides 
and right hand sides are well-typed with the same type:
\eqref{split-1} rewrites strings of type $\typeA \to \typeA$, 
\eqref{split-2} and \eqref{split-6} rewrite strings of type $\typeA \to \typeCopyA$,
\eqref{split-3}  and \eqref{split-4} rewrite strings of type $\typeA  \to \typeT$,
and \eqref{split-5} rewrites strings of type $\typeConst \to \typeCopyA$.
Thus $\sym{split}(R)$ is a typed SRS.

\begin{lemma}\label{lemma:sort-s2-sufficient}
If a typed string $w$ of type $\tau_1 \to \tau_2$ with $\tau_1,\tau_2\in\Types$ admits an infinite reduction
{w.r.t.}~$\sym{split}(R)$, then there exists a typed string $w'$ of type $\typeConst \to \typeT$, which admits an infinite reduction {w.r.t.}~$\sym{split}(R)$.
\end{lemma}
\begin{proof}
Assume that $w$ is of type $\tau_1 \to \tau_2 \not= \typeConst \to \typeT$.
Note that $\tau_1 \not= \typeT$ and $\tau_2 \not=\typeConst$, since no well-typed, non-empty strings of these types exist.

We prepend and append symbols to $w$, resulting in a string $uwv$ s.t. $uwv$ is well-typed with
type $\typeConst \to \typeT$. If $\tau_2 = \typeA$ then choose $u = \sym{B}$, if $\tau_2 = \typeCopyA$ then 
choose $u = \sym{W}$, otherwise choose $u = \varepsilon$. If $\tau_1 = \typeA$, then choose $v=\sym{E}$,
if $\tau_1=\typeCopyA$ then choose $v = \sym{L}\sym{E}$.

Clearly, if $w \to_{\sym{split}(R)} w_1 \to_{\sym{split(R)}} \cdots$  is 
an infinite reduction {w.r.t.}~$\sym{split}(R)$, then also
$uwv \to_{\sym{split}(R)} uw_1v \to_{\sym{split(R)}} \cdots$  is an infinite reduction 
{w.r.t.}~$\sym{split}(R)$.
\end{proof}

Inspecting the typing of the symbols shows:
\begin{lemma}\label{lem:cases-sort-S2}
Any well-typed string of type $\typeConst \to \typeT$ is of one of the following forms:
\begin{itemize}
\item $\sym{B}\?u\?\sym{E}$ where $u \in \Sigma^*$, 
\item $\sym{W}\?w\?\sym{L}\?u\?\sym{E}$ where $w \in \overline{\Sigma}^*$ and $u\in\Sigma^*$, or
\item $\sym{W}\?w\?\sym{R}_{i,j}\?u\?\sym{E}$ where $w \in \overline{\Sigma}^*$ and $u\in\Sigma^*$.
\end{itemize}
\end{lemma}

We define a mapping from well-typed strings of type $\typeConst \to \typeT$ into untyped strings over $\Sigma$ as follows:
\begin{definition}
For a string $w : \typeConst \to \typeT$, the string $\mapSplit(w) \in \Sigma^*$ is defined according to the cases of Lemma~\ref{lem:cases-sort-S2}:
$$\begin{array}{llll}
 \mapSplit(\sym{B}\?u\?\sym{E})            &:=& u
\\
 \mapSplit(\sym{W}\?w\?\sym{L}\?u\?\sym{E}) &:=& \underline{w}\?u
\\
 \mapSplit(\sym{W}\?w\?\sym{R}_{i,j}\?u\?\sym{E}) &:= &\ell_s\?\underline{w}\?u &  \text{
if $(\sym{B}\?\ell_s \to \sym{W}\?\sym{R}_{i,j})\in\sym{split}(R)$.}
\end{array}$$
\end{definition}

\begin{lemma}\label{lemma:mapping}
Let $w$ be a well-typed string of type $\typeConst \to \typeT$ and $w \tosplit w'$. Then $[\mapSplit(w)] \cycto_R^* [\mapSplit(w')]$.
\end{lemma}
\begin{proof}
We inspect the cases of Lemma~\ref{lem:cases-sort-S2} for $w$:
\begin{itemize}
 \item If $w = \sym{B}\?u\?\sym{E}$ where $u \in \Sigma^*$, 
 then the step $w \tosplit w'$ can use a rule of type \eqref{split-1} or \eqref{split-4}.
 If rule $\eqref{split-1}$ is applied, then $\mapSplit(w) \to_R \mapSplit(w')$ and thus $[\mapSplit(w)] \cycto_R [\mapSplit(w')]$.
 If rule $\eqref{split-4}$ is applied, then $w = \sym{B}\?\ell_2\?u' \tosplit \sym{W}\?\sym{R}_{i,j}\?u' = w'$
and $\mapSplit(w)  = \ell_2\?u' = \mapSplit(w')$ and thus $[\mapSplit(w)] = [\mapSplit(w')]$.
\item If $w = \sym{W}\?v\?\sym{L}\?u\?\sym{E}$ where $v \in \overline{\Sigma}^*$ and $u\in\Sigma^*$,
then the step $w \tosplit w'$ can use rules of type \eqref{split-1}, \eqref{split-2}, or  \eqref{split-3}.
 If rule $\eqref{split-1}$ is used, then $\mapSplit(w) \to_R \mapSplit(w')$ and thus $[\mapSplit(w)] \cycto_R [\mapSplit(w')]$.
 If rule $\eqref{split-2}$ or $\eqref{split-3}$ is used, then $\mapSplit(w) = \mapSplit(w')$ and thus
 $[\mapSplit(w)] = [\mapSplit(w')]$.
\item If $w = \sym{W}\?v\?\sym{R}_{i,j}\?u\?\sym{E}$ where $v \in \overline{\Sigma}^*$ and $u\in\Sigma^*$,
then the step $w \tosplit w'$ can use a rule of type \eqref{split-1}, \eqref{split-5}, or  \eqref{split-6}.
If rule $\eqref{split-1}$ is used, then $\mapSplit(w) \to_R \mapSplit(w')$ and thus $[\mapSplit(w)] \cycto_R [\mapSplit(w')]$.
If rule $\eqref{split-6}$ is used, then $\mapSplit(w) = \mapSplit(w')$ and thus
$[\mapSplit(w)] = [\mapSplit(w')]$.
If rule $\eqref{split-5}$ is used, then 
$w = \sym{W}\?v\?\sym{R}_{i,j}\?\ell_p\?\sym{E}$ and $w'= \sym{W}\?v\?\sym{L}\?r_i\?\sym{E}$
and $\mapSplit(w) = \ell_s\?\underline{v}\?\ell_p \sim \ell_p\ell_s\underline{v} \to_R r_i\?\underline{v} \sim \underline{v}\?r_i = \mapSplit(w')$ and thus $[\mapSplit(w)] \cycto_R [\mapSplit(w')]$.
\qedhere
\end{itemize}
\end{proof}

\begin{theorem}[Soundness and completeness of $\sym{split}$]\label{thm:split-sound-and-complete}
The transformation $\sym{split}$ is sound and complete, {i.e.}~$\tosplit$ is terminating if, and only if $\cycto_R$ is terminating.
\end{theorem}
\begin{proof}
Soundness is proved in Proposition~\ref{prop:split-sound}. It remains to show completeness. 
W.l.o.g.~we assume that $\to_R$ is terminating, since otherwise $\cycto_R$ is obviously non-terminating.
Type introduction (Corollary~\ref{cor:typeintro-term}) and Lemma~\ref{lemma:sort-s2-sufficient} show that it is sufficient to construct a non-terminating cycle rewrite sequence for any well-typed string $w$ of type 
$\typeConst \to \typeT$ where $w$ has an infinite $\tosplit$-reduction. 
Let $w$ be a well-typed string of type $\typeConst \to \typeT$ {s.t.}~$w$ admits an infinite reduction $w \tosplit w_1 \tosplit w_2 \tosplit \cdots$.
Lemma~\ref{lemma:mapping} shows that the cycle rewrite sequence $[\mapSplit(w)] \cycto_R^* [\mapSplit(w_1)] \cycto_R^* [\mapSplit(w_2)] \cycto_R^* \cdots$ exists. It remains to show that the constructed sequence is infinite.
One can observe that the infinite $\tosplit$-sequence starting with $w$ must have infinitely many applications of rule $\eqref{split-5}$, since
 every sequence of $\xrightarrow{\eqref{split-1}\vee\eqref{split-2}\vee\eqref{split-3}\vee\eqref{split-6}\vee\eqref{split-4}}$-steps is terminating (since we assumed that $\to_R$ is terminating).
Since $\mapSplit(\cdot)$ translates an $\eqref{split-5}$-step into exactly one $\cycto_R$-step,
the sequence $[\mapSplit(w)] \cycto_R^* [\mapSplit(w_1)] \cycto_R^* [\mapSplit(w_2)] \cycto_R^* \ldots$ consists of infinitely many $\cycto_R$-steps.
\end{proof}
 
\begin{remark}
Note that simulating a cycle rewrite step $w \cycto_R w'$ requires $O(|w|)$ $\tosplit$-steps.
Thus the derivation height of $\sym{split}(R)$ (i.e.~the length of a maximal string rewrite sequence for $\sym{B}w\sym{E}$)
is asymptotically the square of the derivation height of $R$ (i.e.~the length of a maximal cycle rewrite sequence  for $w$).
Note also that the same arguments and properties apply to the two other transformations, $\sym{shift}$ and  $\sym{rotate}$,
which will be presented later.
\end{remark}

\subsubsection{Relative Termination}
We discuss how the transformation $\sym{split}$ can be adapted to show relative termination of 
cycle rewriting.
With $\sym{\ref{split-1}}(\cdot)$ (and $\sym{\ref{split-5}}(\cdot)$, resp.) we 
denote the rules which are generated by the transformation
$\sym{split}(\cdot)$ according to rule \eqref{split-1} (and \eqref{split-5}, resp.). 
The idea of using the transformation $\sym{split}$ for relative termination, 
is to transform $S \subseteq R$  into $\sym{split}(R)$ where,
only the rules corresponding to $\sym{\ref{split-1}}(S)$ and $\sym{\ref{split-5}}(S)$ become strict rules,
while all other rules are weak rules in the transformed system.
\begin{definition}
Let $S \subseteq R$ be  SRSs over  alphabet $\Sigma$.
The transformation $\sym{split}_{\mathit{rel}}$ is defined as follows:
$$
\begin{array}{lcl}
\sym{split}_{\mathit{rel}}(S,R)  &:=& (\sym{\ref{split-1}}(S) \cup \sym{\ref{split-5}}(S),\sym{split}(R))
\end{array}
$$
\end{definition}

\begin{example}
Let 
$S = \{a\?b\?c \to b\?a\?c\}$
and ${R} = {S \cup {\{b\?c \toweak c\?b\}}}$.
Then 
$\sym{split}_{\mathit{rel}}(S,R)=(S',R')$ with
\[
\begin{array}{l@{}c@{}l@{}}
R'&~=~\{
        &b\?c \toweak c\?b,
        \sym{B}\?c \toweak \sym{W}\sym{R}_{1,1},
        \sym{B}\?b\?c \toweak\sym{W}\sym{R}_{2,1},
        \sym{B}\?c\? \toweak \sym{W}\sym{R}_{2,2},
        \sym{R}_{1,1}\?b\?\sym{E}\toweak \sym{L}\?c\?b\sym{E},
\\
      &&\sym{R}_{1,1}\?b \toweak \overline{b}\?\sym{R}_{1,1},
        \sym{R}_{2,1}\?b \toweak \overline{b}\?\sym{R}_{2,1},
        \sym{R}_{2,2}\?b \toweak \overline{b}\?\sym{R}_{2,2},
        \sym{R}_{1,1}\?c \toweak \overline{c}\?\sym{R}_{1,1},
        \sym{R}_{2,1}\?c \toweak \overline{c}\?\sym{R}_{2,1},
\\
       &&\sym{R}_{2,2}\?c \toweak \overline{c}\?\sym{R}_{2,2},
        \sym{R}_{1,1}\?a \toweak \overline{a}\?\sym{R}_{1,1},
        \sym{R}_{2,1}\?a \toweak \overline{a}\?\sym{R}_{2,1},
        \sym{R}_{2,2}\?a \toweak \overline{a}\?\sym{R}_{2,2},
        \overline{b}\?\sym{L} \toweak \sym{L}\?b,\\
       &&
        \overline{c}\?\sym{L} \toweak \sym{L}\?c,
        \overline{a}\?\sym{L} \toweak \sym{L}\?a,
        \sym{W}\?\sym{L} \toweak \sym{B}\} \cup S'
\\         
S'&~=~\{&\sym{R}_{2,1}\?a\?\sym{E} \to \sym{L}\?b\?a\?c\?\sym{E},
         \sym{R}_{2,2}\?a\?b\?\sym{E} \to \sym{L}\?b\?a\?c\?\sym{E},
         a\?b\?c \to b\?a\?c\}
\end{array}
\]
The termination provers \aprove\ and \TTTT\ automatically show that $S'$ is string terminating relative to $R'$.
Since -- as we show below -- $\sym{split}_{\mathit{rel}}$ is sound for relative termination,
 we can conclude that $S$ is cycle terminating relative to $R$. 
 
As another example, consider
$S=\{a\?a \to a\?b\?a\}$ and
$R = S \cup \{a\?b \toweak b\?a\}$, and let $\sym{split}_{\mathit{rel}}(S,R)= (S',R')$.
Both provers show that $S'$ is not terminating relative to $R'$.
Since the transformation is also complete for
relative termination, we can conclude that $S$ is not cycle terminating relative to $R$, which can also be verified by the 
following counter-example:  the infinite
cycle rewrite sequence $[a\?a] \cycto_S [a\?b\?a] \cycto_R [a\?a\?b] \cycto_S [a\?b\?a\?b] \cycto_R [a\?a\?b\?b] \cdots$ has infinitely many
$\cycto_S$-steps. 
\end{example}

For $S \subseteq R$ and $\sym{split}_{\mathit{rel}}(S,R) = (S',R')$,
we will show that $S$ is cycle terminating relative to $R$ if, and only if 
$S'$ is string terminating relative to $R'$.

First note that for every SRS $R$ and every cycle rewrite step, $v_1 \cycto_R v_2$ implies 
$\sym{B} v_1' \sym{E} \to_{\sym{split}(R)}^+ \sym{B} v_2' \sym{E}$
where $v_j \sim v_j'$ for $j=1,2$ and the 
$\to_{\sym{split}(R)}$-sequence contains exactly one application
of rule $\sym{\ref{split-1}}(R)$ or $\sym{\ref{split-5}}(R)$
(see Proposition~\ref{prop:split-sound} and inspect the construction of $\sym{split}$).

\begin{proposition}\label{prop-split-rel-sound}
The transformation $\sym{split}_{\mathit{rel}}$ is sound for relative termination.
\end{proposition}
\begin{proof}
Let $\Sigma_A$ be an alphabet, $S \subseteq R$ be SRSs over $\Sigma_A$, and
$\sym{split}_{\mathit{rel}}(S,R) = (S',R')$.
Assume that $S$ is not cycle terminating relative to $R$. Then 
there exists an infinite reduction $[w_i] \cycto_R^* [u_i] \cycto_S [w_{i+1}]$
for $i=1,2,\ldots$.
Then $\sym{B} w_i' \sym{E} \to_{\sym{split}(R)}^* \sym{B}\?u_i'\?\sym{E} \to_{\sym{split}(S)}^+ \sym{B}\?w_{i+1}'\?\sym{E}$ with 
$w_i \sim w_i', w_{i+1} \sim w_{i+1}'$ and $u_i \sim u_i'$. 
The sequence $u_i' \to_{\sym{split}(S)}^+ w_{i+1}$ uses exactly once
the rule $\sym{\ref{split-1}}(S)$ or the rule $\sym{\ref{split-5}}(S)$.
Thus we have 
$\sym{B}\?w_i'\?\sym{E} \to_{R'}^* \sym{B}\?u_i'\?\sym{E} \to_{R'}^*.\to_{S'}.\to_{R'}^* \sym{B}\?w_{i+1}'\?\sym{E}$ for all $i=1,2,\ldots$.
Hence $S'$ is not string terminating relative to $R'$.
\end{proof}
For proving completeness, we again use the typed variant of the string
rewrite system.\newpage

\begin{theorem}\label{thm:split-sound-and-complete-relative}
The transformation $\sym{split}_{\mathit{rel}}$ is sound and complete for relative termination.
\end{theorem}
\begin{proof}
Soundness was proved in Proposition~\ref{prop-split-rel-sound}.
For completeness, consider $S \subseteq R$ with $\sym{split}_{\mathit{rel}}(S,R) = (S',R')$ and
assume that a string $w_1$ of type 
$\typeConst \to \typeT$
has an infinite reduction of the form
$w_i \to_{S'} u_i \to_{R'}^* w_{i+1}$
for all $i =1,2,\ldots$ 
(which is sufficient due to Theorem~\ref{theo:typeintro-rel} and Proposition~\ref{prop:rel-prop}).

We first consider the first reductions of the form
$w_i \to_{S'} u_i$.
If $w_i \to_{S'} u_i$ by rule $\sym{\ref{split-1}}(S)$,
then clearly $[\mapSplit(w_i)] \cycto_S [\mapSplit(u_i)]$.
If $w_i \to_{S'} u_i$ by rule $\sym{\ref{split-5}}(S)$,
then due to typing 
$w_i = \sym{W}\?v\?\sym{R}_{i,k}\?\ell_p\?\sym{E}$ 
and $u_i = \sym{W}\?v\?\sym{L}\?r_j\?\sym{E}$.
Applying $\mapSplit(\cdot)$ to $w_i$ and $u_i$ shows
that $\mapSplit(w_i) = \ell_s\?\underline{v}\?\ell_p 
\sim \ell_p\ell_s\underline{v} = \ell_j\?\underline{v}$ and
$\mapShift(u_i)= \underline{v}\?r_j$ which implies
$[\mapSplit(w_i)] \cycto_S [\mapSplit(u_i)]$.

Now, for the reduction sequences $u_i \to_{R'}^* w_{i+1}$,
Lemma~\ref{lemma:mapping} shows that
$[\mapSplit(u_i)] \cycto_R^* [\mapSplit(w_{i+1})]$ holds.
Thus, $[\mapSplit(w_i)] \cycto_S . \cycto_R^* [\mapSplit(w_{i+1})]$
for all $i=1,2,\ldots$. Hence, $S$ is not cycle terminating relative to $R$.
\end{proof}

\subsection{The Transformation Shift}\label{subsec:shift}
We first present the general ideas of the transformation
$\sym{shift}$ before giving its definition.
We write $\shift$ for the relation which moves the first element of a string to the end, 
{i.e.}~$a\?u \shift u\?a$ for every $a \in \Sigma$ and $u \in \Sigma^*$. 
Clearly, $u \sim v$ if and only if $u \shift^{< |u|} v$.

For a string rewrite system $R$, we define ${\mathrm{len}}(R)$ as the size of the largest left-hand side of the rules in $R$, 
{i.e.}~${\mathrm{len}}(R) = \max_{(\ell \to r) \in R} |\ell|$.

The approach of the transformation $\sym{shift}$ 
is to shift at most ${\mathrm{len}}(R)-1$ symbols from the left end to the right end and then 
to apply a string rewrite step (i.e. this the relation $\shift^{< {\mathrm{len}}(R)} . \to_R$). 

\begin{example}\label{example-simple-cont-shift}
As in Example~\ref{example-simple}, let $R = \{a\?b\?c \to b\?b\?b\?b\}$ and  $[b\?c\?d\?d\?a] \cycto_R [b\?b\?d\?d\?b\?b]$. 

The approach of transformation $\sym{shift}$ is to simulate the cycle rewrite step, by first shifting symbols from the left end to the right end of the string until $a\?b\?c$ becomes a substring, and then applying a string rewrite step, i.e.
$b\?c\?d\?d\?a \shift^{< {\mathrm{len}}(R)} . \to _R d\?d\?b\?b\?b\?b$, 
since $b\?c\?d\?d\?a \shift c\?d\?d\?a\?b \shift  d\?d\?a\?b\?c \to_R d\?d\?b\?b\?b\?b$.
\end{example}

The following lemma obviously holds:
\begin{lemma}\label{lem:Rightarrow-simulates-sim-shift}
Let $R$ be an SRS over  an alphabet $\Sigma$ and $u,v\in\Sigma^*$: 
If $[u]\cycto_R [v]$ then there exist  $u',v'\in\Sigma^*$ {s.t.}~$u \shift^* u' \to_R  v' \shift^* v$.

Moreover, we  can restrict the number of $\shift$-steps before a rewrite step:
For $u,v,w \in\Sigma^*$,
if $u \shift^k v \to_R w$ for some $ k \geq {\mathrm{len}}(R)$, then $u \shift^m v' \to_R v'' \shift^* w$ for some $v',v''\in\Sigma^*$ and $m < {\mathrm{len}}(R)$.
\end{lemma}

Using the previous lemmas, we are able to prove the following proposition:

\begin{proposition}\label{prop:N-shifts-are-sufficient-shift}
Let $R$ be an SRS. If $\cycto_R$ is non-terminating, then
$\shift^{< {\mathrm{len}}(R)} . \to_R$ admits an infinite reduction.
\end{proposition}

\begin{proof}
Let $[w_1] \cycto_R [w_2] \cycto_R \cdots$ 
be an infinite cycle rewrite sequence.
By Lemma~\ref{lem:Rightarrow-simulates-sim-shift} we have
$w_1 \shift^* w_1' \to_R w_1'' \shift^* w_2 \shift^* w_2' \to_R w_2'' \shift^* \cdots$
and by the second part of the lemma, we can always move more than ${\mathrm{len}}(R)$ $\shift$-steps to the right,
and thus we derive an infinite sequence $w_1 \shift^{< {\mathrm{len}}(R)} u_1 \to_R u_1' \shift^{< {\mathrm{len}}(R)} u_2 \to_R u_2' \cdots$
and thus $w_1$ admits an infinite reduction of the required form.
\end{proof}

For an SRS $R$, the SRS $\sym{shift}(R)$ encodes the relation $\shift^{<{\mathrm{len}}(R)}.\to_R$
where extra symbols are used to separate the steps, and copies of the alphabet underlying $R$ are used to ensure completeness of the transformation.

For the remainder of the section, we fix an SRS $R$ over alphabet $\Sigma_A = \{a_1,\ldots,a_n\}$.
Let us write $\Sigma_B, \Sigma_C$ for fresh copies of the alphabet $\Sigma_A$.
We use the following notation to switch between the alphabets: for $X,Y \in \{A,B,C\}$ and  $w \in \Sigma_X$ we write $\trans{w}{Y}$ to denote the copy of $w$ in the alphabet $Y$ where every symbol is translated from alphabet $X$ to alphabet $Y$.

\begin{definition}[The transformation $\sym{shift}$]
Let $R$ be an SRS over alphabet $\Sigma_A$ and let $N = \max(0,{\mathrm{len}}(R)-1)$.
The SRS $\sym{shift}(R)$ 
over the alphabet 
 $\Sigma_A \cup \Sigma_B \cup \Sigma_C \cup
 \{\sym{B},\sym{E},\sym{W},\sym{V},\sym{M},\sym{L},\sym{R},\sym{D}\}$
 (where
 $\sym{B},\sym{E},\sym{W},\sym{V},\sym{M},\sym{L},\sym{R},\sym{D}$ are
 fresh for $\Sigma_A \cup \Sigma_B \cup \Sigma_C$) consists of the
 following rules:\vspace{-3 pt}
\begin{flushleft}
\begin{minipage}{.53\textwidth}
\begin{align}
   \sym{B} &\to \sym{W}\?\sym{M}^N\?\sym{V} \label{shift-1}\tag{shiftA}
\\
  \sym{M} &\to \varepsilon \label{shift-2}\tag{shiftB}
 \\
  \sym{M}\?\sym{V}\?a &\to \sym{V}\trans{a}{B} &&\hspace*{-2.8ex}\text{for all $a \in \Sigma_A$} \label{shift-3}\tag{shiftC}
\\
  b\?a &\to a\?b &&\hspace*{-2.8ex}\text{for all $a\in\Sigma_A$, $b\in\Sigma_B$\!\!\!\!} \label{shift-4}\tag{shiftD}
\\
  b\?\sym{E} &\to \trans{b}{A}\?\sym{E} &&\hspace*{-2.8ex}\text{for all $b \in \Sigma_B$} \label{shift-5}\tag{shiftE}
 \end{align}
\end{minipage}%
\begin{minipage}{.47\textwidth}
\begin{align}
  \sym{W}\?\sym{V} &\to \sym{R}\?\sym{L} ~~\label{shift-6}\tag{shiftF}
\\
  \sym{L}\?a &\to \trans{a}{C}\?\sym{L} &&\hspace*{-1.1ex}\text{for all $a \in \Sigma_A$}\label{shift-7}\tag{shiftG}
\\
  \sym{L}\?\ell &\to \sym{D}\?r &&\hspace*{-1.1ex}\text{for all $(\ell \to r)\!\in\!R$\!\!}\label{shift-8}\tag{shiftH}
\\
  c\?\sym{D} &\to \sym{D}\trans{c}{A} &&\hspace*{-1.1ex}\text{for all $c\in\Sigma_C$}\label{shift-9}\tag{shiftI}
\\
 \sym{R}\?\sym{D} &\to \sym{B} ~~\label{shift-10}\tag{shiftJ}
\end{align}
\end{minipage}
\end{flushleft}
\end{definition}\medskip

\noindent The rules $\eqref{shift-1}$ - $\eqref{shift-5}$ encode the relation $\shift^{< {\mathrm{len}}(R)}$, {i.e.}~for $u\?v \in \Sigma_A^*$ with $|u| < {\mathrm{len}}(R)$, the string $\sym{B}\?u\?v\?\sym{E}$ is rewritten into 
$\sym{W}\?\sym{V}\?v\?u\?\sym{E}$ by these five rules. The sequence of symbols $\sym{M}$ generated by rule \eqref{shift-1} denotes the potential of moving at most $\sym{len}(R)-1$ symbols. The rules \eqref{shift-2} and \eqref{shift-3} either remove
one from the potential or start the moving of one symbol. The rule $\eqref{shift-4}$ performs the movement of a single symbol
until it reaches the end of the string and rule $\eqref{shift-5}$ finishes the movement. 

The remaining rewrite rules perform a single string rewrite step, {i.e.}~for a rule $(\ell \to r)\in R$ the string 
$\sym{W}\?\sym{V}\? w_1 \? \ell \? w_2\?\sym{E}$ is rewritten to $\sym{B}\?w_1\?r\? w_2\sym{E}$ by rules 
$\eqref{shift-6}$ - $\eqref{shift-10}$. 

\begin{example}
For $R_1 = \{a \? a \to a \? b \? a\}$, the transformed string rewrite system $\sym{shift}(R_1)$ is:
$$\begin{array}{l@{~}c@{~}l@{~}l@{\qquad}l@{~}c@{~}l@{~}l@{\qquad}l@{~}c@{~}l@{~}l}
 \sym{B} 
   &\to
   &\sym{W}\?\sym{M}\?\sym{V} 
   &\eqref{shift-1}
&\sym{M} 
   &\to
   &\varepsilon
   &\eqref{shift-2}
&\sym{M}\?\sym{V}\?a 
   &\to 
   & \sym{V}\trans{a}{B} 
   &\eqref{shift-3}
\\
 \sym{M}\?\sym{V}\?b 
   &\to
   &\sym{V}\trans{b}{B} 
   &\eqref{shift-3}
&\trans{a}{B}\?a 
   &\to
   &a\trans{a}{B} 
   &\eqref{shift-4}
&\trans{a}{B}\?b 
   &\to
   &b\trans{a}{B} 
   &\eqref{shift-4}
\\
 \trans{b}{B}\?a 
   &\to
   &a\trans{b}{B} 
   &\eqref{shift-4}
&\trans{b}{B}\?b 
   &\to
   &b\trans{b}{B} 
   &\eqref{shift-4}
&\trans{a}{B}\?\sym{E} 
   &\to 
   &a\?\sym{E}
   &\eqref{shift-5}  
\\
 \trans{b}{B}\?\sym{E} 
   &\to 
   &b\?\sym{E}
   &\eqref{shift-5}  
&\sym{W}\?\sym{V} 
   &\to
   &\sym{R}\?\sym{L}
   &\eqref{shift-6}
&\sym{L}\?a 
   &\to
   &\trans{a}{C}\?\sym{L}
   &\eqref{shift-7}
\\   
 \sym{L}\?b 
   &\to
   &\trans{b}{C}\?\sym{L}
   &\eqref{shift-7}
&\sym{L}\?a\?a 
   &\to 
   &\sym{D}\?a\?b\?a 
   &\eqref{shift-8}
&\trans{a}{C}\?\sym{D} 
   &\to
   &\sym{D}\?a
   &\eqref{shift-9}
\\
\trans{b}{C}\?\sym{D} 
   &\to
   &\sym{D}\?b
   &\eqref{shift-9}
&\sym{R}\?\sym{D}
   &\to
   &\sym{B}
   &\eqref{shift-10}
\end{array}
$$ 
The cycle rewrite step $[a\?b\?a] \cycto_{R_1} [b\?a\?b\?a]$
is simulated in the transformed system as follows:
$$\begin{array}{@{}l@{}l@{}}
\Redex{$\sym{B}$}\?a\?b\?a\?\sym{E}

&\xrightarrow{\!\tiny\ref{shift-1}\!\!}
\sym{W}\?\Redex{$\sym{M}\?\sym{V}\?a$}\?b\?a\?\sym{E}

\xrightarrow{\!\tiny\ref{shift-3}\!\!}
\sym{W}\?\sym{V}\?\Redex{$\trans{a}{B}b$}\?a\?\sym{E}

\xrightarrow{\!\tiny\ref{shift-4}\!\!}
\sym{W}\?\sym{V}\?b\Redex{$\trans{a}{B}\?a$}\?\sym{E}

\xrightarrow{\!\tiny\ref{shift-4}\!\!}
\sym{W}\?\sym{V}\?b\?a\Redex{$\trans{a}{B}\?\sym{E}$}
\\
&\xrightarrow{\!\tiny\ref{shift-5}\!\!}
\Redex{$\sym{W}\?\sym{V}$}\?b\?a\?a\?\sym{E}

\xrightarrow{\!\tiny\ref{shift-6}\!\!}
\sym{R}\?\Redex{$\sym{L}\?b$}\?a\?a\?\sym{E}

\xrightarrow{\!\tiny\ref{shift-7}\!\!}
\sym{R}\trans{b}{C}\Redex{$\sym{L}\?a\?a$}\?\sym{E}

\xrightarrow{\!\tiny\ref{shift-8}\!\!}
\sym{R}\?\Redex{$\trans{b}{C}\sym{D}$}\?a\?b\?a\?\sym{E}
\\
&\xrightarrow{\!\tiny\ref{shift-9}\!\!}
\Redex{$\sym{R}\?\sym{D}$}\?b\?a\?b\?a\?\sym{E}

\xrightarrow{\!\tiny\ref{shift-10}\!\!}
\sym{B}\?b\?a\?b\?a\?\sym{E}
\end{array}$$
\end{example}

Together with Proposition~\ref{prop:N-shifts-are-sufficient-shift}, the following lemma
implies soundness of the transformation $\sym{shift}$.

\begin{lemma}\label{lemm:shift-sound}
 If $u \shift^{< {\mathrm{len}}(R)} v \to_R w$, then $\sym{B}\?u\?\sym{E} \to_{\sym{shift}(R)}^+ \sym{B}\?w\?\sym{E}$.
\end{lemma}
\begin{proof}
Let ${{\mathrm{len}}}(R) = m+1$, $u = a_1 \ldots a_n$ and $u \shift^k a_{k+1}\ldots a_n~a_1\ldots a_k = v$,
and let $v = a'_1 \ldots a_n'$ and $w = a_1' \ldots a_i' a_1'' \ldots a_r'' a_{i+j}' \ldots a_n'$,
{i.e.}~the applied rewrite rule is $a_{i+1}'\ldots a_{i+j-1}' \to a_1'' \ldots a_r''$.
Then the following rewrite sequence using $\sym{shift}(R)$ can be constructed:

\[\begin{array}{lr}
\sym{B}\?a_1\ldots a_n\?\sym{E}  
\xrightarrow{\eqref{shift-1}}    \sym{W}\?\sym{M}^m\?\sym{V}\?a_1\ldots a_n\?\sym{E}
\xrightarrow{\eqref{shift-2}}^*  \sym{W}\?\sym{M}^k\?\sym{V}\?a_1\ldots a_n\?\sym{E}\\
\xrightarrow{\eqref{shift-3}}    \sym{W}\?\sym{M}^{k-1}\?\sym{V}\trans{a_1}{B}\?a_2\ldots a_n\?\sym{E}
\xrightarrow{\eqref{shift-4}}^*  \sym{W}\?\sym{M}^{k-1}\?\sym{V}\?a_2\ldots a_n\trans{a_1}{B}\sym{E}\\
\xrightarrow{\eqref{shift-5}}    \sym{W}\?\sym{M}^{k-1}\?\sym{V}\?a_2\ldots a_n\?a_1\sym{E}
\left(\xrightarrow{\eqref{shift-3}}.\xrightarrow{\eqref{shift-4}}^*.\xrightarrow{\eqref{shift-5}}\right)^*
\sym{W}\?\sym{V}\?a_{k+1}\ldots a_n\?a_1\ldots a_k\?\sym{E}\\
\xrightarrow{\eqref{shift-6}}   \sym{R}\?\sym{L}\?a_{k+1}\ldots a_n\?a_1\ldots a_k\?\sym{E}
                           =    \sym{R}\?\sym{L}\?a'_1\ldots a'_n\?\sym{E}
\xrightarrow{\eqref{shift-7}}^* \sym{R}\trans{a'_1}{C}\ldots\trans{a_i'}{C}\?\sym{L}a_{i+1}'\ldots a'_n\?\sym{E}\\
\xrightarrow{\eqref{shift-8}}   \sym{R}\trans{a'_1}{C}\ldots\trans{a_i'}{C}\?\sym{D}a''_1\ldots a''_r a_{i+j}'\ldots a'_n\?\sym{E}
\xrightarrow{\eqref{shift-9}}^* \sym{R}\?\sym{D}\?a_1'\ldots{a_i'}a''_1\ldots a''_r a_{i+j}'\ldots a'_n\?\sym{E}\\
\xrightarrow{\eqref{shift-10}}^*\sym{B}\?a_1'\ldots{a_i'}a''_1\ldots a''_r a_{i+j}'\ldots a'_n\?\sym{E}\\
\end{array}\vspace{-16 pt}
\]
\end{proof}

\noindent We prove completeness by using type introduction.
Let us write $\widetilde{\Sigma_A} := \Sigma_A \cup \Sigma_B \cup \Sigma_C \cup \{\sym{B},\sym{E},\sym{W},\sym{V},\sym{M},\sym{L},\sym{R},\sym{D}\}$.
Let $\Types := \{\typeAB, \typeC, \typeConst, \typeM, \typeT\}$
be the set of types.
We assign the following types to the symbols of $\widetilde{\Sigma_A}$:
$$
\begin{array}{@{}c@{\quad}c@{\quad}c@{\quad}c@{}}
\begin{array}{lll}
a                       &: \typeAB \to \typeAB &\text{for all $a \in \Sigma_A$}\\
b                       &: \typeAB \to \typeAB &\text{for all $b  \in \Sigma_B$}\\
c                       &: \typeC \to \typeC &\text{for all $c \in \Sigma_C$}
\end{array}
&
\begin{array}{lll}
\sym{E}              &: \typeConst \to \typeAB\\ 
\sym{V}              &: \typeAB \to \typeM\\
\sym{M}              &: \typeM \to \typeM
\end{array}
&
\begin{array}{lll}
\sym{W}              &: \typeM \to \typeT\\
\sym{B}              &: \typeAB \to \typeT\\
\sym{L}              &: \typeAB \to \typeC
\end{array}
&
\begin{array}{lll}
\sym{D}              &: \typeAB \to \typeC\\
\sym{R}              &: \typeC \to \typeT\\
                        &\\
\end{array}
\end{array}
$$
We verify that  $\sym{shift}(R)$ is a typed SRS.
Rules 
\eqref{shift-1}, 
\eqref{shift-6},
and \eqref{shift-10}
rewrite strings of 
type $\typeAB \to \typeT$,
rule \eqref{shift-2}
rewrites strings of type $\typeM \to \typeM$,
rule \eqref{shift-3}
rewrites strings of type $\typeAB \to \typeM$,
rule \eqref{shift-4}
rewrites strings of type $\typeAB \to \typeAB$,
rule \eqref{shift-5}
rewrites strings of type $\typeConst \to \typeAB$,
and
rules 
\eqref{shift-7},
\eqref{shift-8}, and
\eqref{shift-9}
rewrite strings of type $\typeAB \to \typeC$.

\begin{lemma}\label{lem:shift-t3-suff}
Let $w \in \widetilde{\Sigma_A}^*$  be a well-typed string, 
{s.t.}~$w$ admits an infinite reduction 
{w.r.t.}~$\toshift$.
Then there exists a well-typed string $w' \in \widetilde{\Sigma_A}^*$ of type $\typeConst \to \typeT$ which admits
an infinite reduction {w.r.t.}~$\toshift$.
\end{lemma}
\begin{proof}
Let $w$ be a well-typed string of type $\tau_1 \to \tau_2 \not= \typeConst \to \typeT$ where
$\tau_1,\tau_2 \in \Types$ s.t.\ $w$ admits an infinite reduction w.r.t.\ $\toshift$.
We prepend  and append symbols to $w$ constructing a string $w' = uwv$ of type 
$\typeConst \to \typeT$ as follows:

The string $u$ is the empty string if $\tau_2 = \typeT$ and
otherwise for any $\tau_2 \in \{\typeConst,\typeAB,\typeM,\typeC\}$ 
there is a sequence of type $\tau_2 \to \typeT$,
which is used for $u$:
\[
\begin{array}{l@{~}c@{~}l@{\qquad\qquad}l@{~}c@{~}l@{\qquad\qquad}l@{~}c@{~}l@{\qquad\qquad}l@{~}c@{~}l}
\sym{R} &:& \typeC \to \typeT
&
\sym{W} &:& \typeM \to \typeT
&
\sym{B} &:& \typeAB \to \typeT
&
\sym{B}\?\sym{E}&:&\typeConst \to \typeT
\end{array}
\]
The string $v$ is the empty string if $\tau_1 = \typeConst$ and
otherwise for any $\tau_1 \in \{\typeAB,\typeM,\typeC,\typeT\}$ 
there is a sequence of type $\typeConst \to \tau_1$,
which is used for $v$:
\[
\begin{array}{l@{~}c@{~}l@{\qquad\qquad}l@{~}c@{~}l@{\qquad\qquad}l@{~}c@{~}l@{\qquad\qquad}l@{~}c@{~}l}
\sym{E} &:& \typeConst \to \typeAB
&
\sym{V\?E} &:& \typeConst \to \typeM
&
\sym{D\?E} &:& \typeConst \to \typeC
&
\sym{B\?E} &:& \typeConst \to \typeT
\end{array}
\]
Clearly, the infinite reduction for $w$ can be used to construct an infinite reduction for $uwv$.
\end{proof}
By checking all possible cases,
the following lemma can be verified:

\begin{lemma}\label{lem:cases-terms-sort-T3}
Any well-typed string of type $\typeConst \to \typeT$ 
is of one of the following forms:
\begin{itemize}
 \item $\sym{W}\?\sym{M}^i\sym{V}w\sym{E}$ where $i \geq 0$ and $w \in (\Sigma_A \cup \Sigma_B)^*$,
 \item $\sym{B}\?w\?\sym{E}$ where $w \in (\Sigma_A \cup \Sigma_B)^*$,
 \item $\sym{R}\?w_c\?\sym{L}\?w\?\sym{E}$ where $w_c \in \Sigma_C^*$ and $w \in (\Sigma_A \cup \Sigma_B)^*$, or
 \item $\sym{R}\?w_c\?\sym{D}\?w\?\sym{E}$ where $w_c \in \Sigma_C^*$ and $w \in (\Sigma_A \cup \Sigma_B)^*$.
\end{itemize}
\end{lemma}

For a string $w \in (\Sigma_A\cup \Sigma_B)^*$, let $\pi_A(w)$ be the string $w$ where all symbols  $b\in\Sigma_B$
are removed and let $\overline{\pi}_B(w)$ be the reversed string of $w'$ where  $w'$  is $w$ where all symbols $a\in\Sigma_A$ are removed.

\begin{definition}
For a well-typed string $w : \typeConst \to \typeT$,
the mapping $\mapShift(w) \in \Sigma_A^*$ is defined according to the cases of 
Lemma~\ref{lem:cases-terms-sort-T3} as follows:
\[\begin{array}{lcl@{\qquad\qquad}lcl}
  \mapShift(\sym{W}\?\sym{M}^i\?\sym{V}\?w\?\sym{E}) &:=& \pi_A(w)\?\overline{\pi}_B(w)
 &\mapShift(\sym{B}\?w\?\sym{E}) &:=& \pi_A(w)\?\overline{\pi}_B(w)\\
  \mapShift(\sym{R}\?w_c\?\sym{L}\?w\?\sym{E})&:=&\trans{w_c}{A}\?\pi_A(w)\?\overline{\pi}_B(w)
 &\mapShift(\sym{R}\?w_c\?\sym{D}\?w\?\sym{E})&:=&\trans{w_c}{A}\?\pi_A(w)\?\overline{\pi}_B(w)
  \end{array}
 \]
\end{definition}

\begin{lemma}\label{lemma:shift-correspondence}
Let $u$ be a well-typed string of type $\typeConst \to \typeT$.
 If $u  \toshift v$, then $[\mapShift(u)] \cycto_R^* [\mapShift(v)]$.
\end{lemma}
\begin{proof}
We go through the cases of Lemma~\ref{lem:cases-terms-sort-T3} and inspect all applicable rewrite rules:
\begin{itemize}
 \item If $u = \sym{W}\?\sym{M}^i\?\sym{V}\?w\?\sym{E}$,
 then there are the following cases: If rule 
\eqref{shift-2}, \eqref{shift-4}, \eqref{shift-5}, or \eqref{shift-6} is applied,
then $\mapShift(u) = \mapShift(v)$ and
if rule \eqref{shift-3} is applied, then $\mapShift(u) = a\?w_1\?w_2$ (with $w = a\?w'$ and $\pi_A(w) = a\?w_1$ and $\overline{\pi}_B(w) = w_2$) and $\mapShift(v) = w_1\?w_2\?a$ and thus $\mapShift(u) \sim \mapShift(v)$.
\item If $u = \sym{B}\?w\?\sym{E}$, then rules 
\eqref{shift-1}, \eqref{shift-4}, or \eqref{shift-5} may be applied and in all cases  $\mapShift(u) = \mapShift(v)$ holds.
\item For $u= \sym{R}\?w_c\?\sym{L}\?w\?\sym{E}$, there are two cases:
If rule \eqref{shift-4}, \eqref{shift-5}, or \eqref{shift-7} is applied, then  $\mapShift(u) = \mapShift(v)$.
If rule \eqref{shift-8} is applied, then $\mapShift(u) \to_R \mapShift(v)$.
\item If $u = \sym{R}\?w_c\?\sym{D}\?w\?\sym{E}$, then
rules \eqref{shift-4}, \eqref{shift-5}, \eqref{shift-9}, and \eqref{shift-10} may be applied and in all cases
 $\mapShift(u) = \mapShift(v)$ holds.\qedhere
\end{itemize}
\end{proof}

\begin{proposition}\label{prop-shift-complete}
Let $u : \typeConst \to \typeT$ {s.t.}~$u$ 
admits an infinite $\toshift$-reduction.
Then $[u]$ admits an infinite $\cycto_R$-reduction.
\end{proposition}
\begin{proof}
Suppose that $u : \typeConst \to \typeT$
admits an infinite $\toshift$-reduction.
Applying Lemma~\ref{lemma:shift-correspondence} shows
that it is possible to construct an infinite reduction of $\cycto_R^*$-steps starting with $[u]$.
Since applications of rule $\eqref{shift-8}$ are translated into
$\cycto_R$ steps, it remains to show that the given infinite reduction of $u$ 
has infinitely many applications of rule \eqref{shift-8}.
Therefore we show that the rewrite  system consisting of all rules of 
$\sym{shift}(R)$ except for rule $\eqref{shift-8}$ is terminating. Let us denote
this system by $\sym{shift}_{\setminus \{\eqref{shift-8}\}}(R)$, i.e. 
$\sym{shift}_{\setminus \{\eqref{shift-8}\}}(R) := (\sym{shift}(R) \setminus \{\ell \xrightarrow{\eqref{shift-8}} r\})$.
For proving termination of $\sym{shift}_{\setminus \{\eqref{shift-8}\}}(R)$ we first eliminate rule $\eqref{shift-1}$:
Let $\sigma_0$ be the polynomial interpretation defined by 
 $\sigma_0(B) = \sigma_0(D) = \lambda x.x+1$ and $\sigma_0(y) = \lambda x.x$
 for all other symbols $y$. It is easy to verify that $\sigma_0(\ell) > \sigma_0(r)$ for $\ell \xrightarrow{\eqref{shift-1}} r$
 and $\sigma_0(\ell) \geq \sigma_0(r)$ for $\ell \to r \in \sym{shift}_{\setminus \{\eqref{shift-8}\}}(R)$.
Thus it suffices to prove termination of 
$\sym{shift}_{\setminus \{\eqref{shift-1},\eqref{shift-8}\}}(R) := (\sym{shift}(R) \setminus \{\ell \xrightarrow{\eqref{shift-8}} r,\ell \xrightarrow{\eqref{shift-1}} r\})$:
We use the following polynomial interpretation  $\sigma$. 
Let $N = \max{\{1,{\mathrm{len}}(R)\}}$ and
\[
\begin{array}{@{}l@{~}l@{~}l@{\qquad}l@{~}l@{~}l@{\qquad}l@{~}l@{~}l@{}}
\sigma(a) &=& \lambda x.4x+4 \text{ for all } a\in \Sigma_A
&\sigma(\sym{D}) &=& \lambda x.2^N x + 2 ^ {N+1} - 1
&\sigma(\sym{M}) &=& \lambda x.2x+1
\\
\sigma(b) &=& \lambda x.8x+1 \text{ for all } b\in \Sigma_B
&\sigma(\sym{R}) &=& \lambda x.2x+1
&\sigma(\sym{L}) &=& \lambda x.x
\\
\sigma(c) &=& \lambda x.4x \text{ for all } c \in \Sigma_C
&\sigma(\sym{E}) &=& \lambda x.7x+1
&\sigma(\sym{W}) &=& \lambda x.2x+2
\\
\sigma(\sym{B}) &=& \lambda x.2^N x + 2^{N+1}-1
&\sigma(\sym{V}) &=& \lambda x.x
\end{array}
\]
We verify that the inequation $\sigma(\ell) > \sigma(r)$ holds for all rules of $\sym{shift}_{\setminus \{\eqref{shift-1},\eqref{shift-8}\}}(R)$:
For 
rule \eqref{shift-2}, we have $\lambda x.2x+1 > \lambda x.x$,
for rule \eqref{shift-3}, we have $\lambda x.8x + 9 > \lambda x.8x+1$,
for  rule \eqref{shift-4}, we have $\lambda x.32x+33 > \lambda x.32x+8$,
for rule \eqref{shift-5}, we have $\lambda x.56x+9 > \lambda x.28x+8$,
for rule \eqref{shift-6}, we have $\lambda x.2x+2  > \lambda x.2x+1$,
for rule  \eqref{shift-7}, we have $\lambda x.4x+4 > \lambda x.4x$,
for rule \eqref{shift-9}, we have $\lambda x.2^{N+2}x + 3\cdot 2^{N+1} + 2^{N+1}- 4 > \lambda x. 2^{N+2}x + 3\cdot 2^{N+1}-1$,
and
for rule \eqref{shift-10}, we have $\lambda x.2^{N+1} x + 2 ^ {N+2} - 1 > \lambda x.2^N x + 2 ^ {N+1} - 1$.\qedhere
\end{proof}

\begin{theorem}\label{theo:shift-sound-and-complete}
 The transformation $\sym{shift}$ is sound and complete.
\end{theorem}
\begin{proof}
 Soundness follows by Lemma~\ref{lemm:shift-sound} and Proposition~\ref{prop:N-shifts-are-sufficient-shift}.
Completeness follows by type introduction 
(Corollary~\ref{cor:typeintro-term}), 
Proposition~\ref{prop-shift-complete}, and Lemma~\ref{lem:shift-t3-suff}.
\end{proof}

\subsubsection{Relative Termination}
We also provide a variation of the transformation $\sym{shift}$ to encode relative cycle termination by relative string termination.
\begin{definition}
Let $S \subseteq R$ be SRSs over an alphabet $\Sigma$.
The transformation $\sym{shift}_{\mathit{rel}}$ is defined as:
$$
\begin{array}{lcl}
 \sym{shift}_{\mathit{rel}}(S,R)  &:=& 
 (\{\sym{L} \ell \to \sym{D} r \mid (\ell \to r) \in S\},\sym{shift}(R))
\end{array}
$$
\end{definition}

\begin{theorem}\label{theo:shift-sound-and-complete-relative}
The transformation $\sym{shift}_{\mathit{rel}}$ is sound and complete for relative termination.
\end{theorem}
\begin{proof}
Let $\Sigma_A$ be an alphabet, $S \subseteq R$ be SRSs over $\Sigma_A$, and  
$\sym{shift}_{\mathit{rel}}(S,R) = (S',R')$.
For proving soundness of $\sym{shift}_{\mathit{rel}}$, assume that $S$ is not cycle terminating  relative to $R$.
Using the same arguments as in Proposition~\ref{prop:N-shifts-are-sufficient-shift}
there exists an infinite reduction such that 
 $w_i\,(\shift^{< {\mathrm{len}}(R)} . \to_R)^*.\shift^{< {\mathrm{len}}(S)} v_i \to_S w_{i+1}$  
for $i=1,2,\ldots$. By Lemma~\ref{lemm:shift-sound} we have for all $i=1,2,\ldots$:
$\sym{B}w_i\sym{E} \to_{\sym{shift}(R)}^*  \sym{B}v_i\sym{E} \to_{\sym{shift}(S)}^+ \sym{B}w_{i+1}\sym{E}$.
Since ${\to_{\sym{shift}(S)}} \subseteq {\to_{R'}^*.\to_{S'}.\to_{R'}^*}$, this shows that $S'$ is not string terminating relative to $R'$.

For proving completeness, we assume that $S'$ is not string terminating relative to $R'$.
We use typed string rewriting and apply Theorem~\ref{theo:typeintro-rel}
and Lemma~\ref{lem:shift-t3-suff}.
Thus there exists a typed string $w_0 : \typeConst \to \typeT$ 
{s.t.}~for all $i = 0,1,2,\ldots$ the reduction sequence $w_i \to_{S'} w'_i \to_{R'}^* w_{i+1}$  exists.
Typing of $w_i$, Lemma~\ref{lem:cases-terms-sort-T3}, and applicability of $\to_{S'}$ show that 
$w_i$ and $w_i'$ must be of the forms
$w_i = \sym{R}\?u_i\?\sym{L}\?\ell\?v_i\?\sym{E}$
and
$w_i' = \sym{R}\?u_i\?\sym{D}\?r\?v_i\?\sym{E}$
where $(\ell \to r) \in S$, $u_i \in \Sigma_C$, and $v_i \in (\Sigma_A \cup \Sigma_B)^*$.
Since 
$\mapShift(w_i) = \trans{u_i}{A}\ell\pi_A(v_i) \overline{\pi}_B(v_i)$
and $\mapShift(w'_i) 
               = \trans{u_i}{A}r\pi_A(v_i) \overline{\pi}_B(v_i)$,
               we have $\mapShift(w_i) \to_S \mapShift(w'_i)$
               and thus we also have $[\mapShift(w_i)] \cycto_S [\mapShift(w'_i)]$.
Since $R' = \sym{shift}(R)$, Lemma~\ref{lemma:shift-correspondence} shows that 
$[\mapShift(w_i')] \cycto_R^* [\mapShift(w_{i+1})]$.
Thus for all $i=1,2,\ldots$ we have $[\mapShift(w_i)] \cycto_S.\cycto_R^* [\mapShift(w_{i+1})]$
which shows that $S$ is not  cycle terminating relative to $R$.
\end{proof}

\subsection{The Transformation \sym{rotate}}\label{subsec:rotate}
We first present the idea of the transformation $\sym{rotate}$.
The transformation is closely related to the definition  of $\cycto$:
The idea is to first rotate the string and then to apply a prefix rewrite step 
(i.e. both steps can be expressed by the relation $\sim . \prefixrewrite_R$).

\begin{example}\label{example-simple-cont-rotate}
As in Example~\ref{example-simple}, let $R = \{a\?b\?c \to b\?b\?b\?b\}$ and  $[b\?c\?d\?d\?a] \cycto_R [b\?b\?d\?d\?b\?b]$. 

The transformation $\sym{rotate}$ simulates the cycle rewrite step, by first rotating the string
(using $\sim$) 
until $a\?b\?c$ becomes a prefix, and then applies a prefix rewrite step. I.e., we have
$b\?c\?d\?d\?a \sim . \prefixrewrite_R b\?b\?b\?b\?d\?d$, since
$b\?c\?d\?d\?a \sim a\?b\?c\?d\?d  \prefixrewrite_R b\?b\?b\?b\?d\?d$.
It would be possible to use shifts $\shift$ to perform the rotation, i.e.
$b\?c\?d\?d\?a \shift^{*} . \prefixrewrite_R b\?b\?b\?b\?d\?d$, since
$b\?c\?d\?d\?a \shift c\?d\?d\?a\?b \shift d\?d\?a\?b\?c \shift d\?a\?b\?c\?d \shift a\?b\?c\?d\?d \prefixrewrite_R b\?b\?b\?b\?d\?d$. However, our transformation $\sym{rotate}$ 
will implement $\sim$ by  moving symbols  from right to left.
\end{example}

The following lemma obviously holds:
\begin{lemma}\label{lem:Rightarrow-simulates-sim-rotate}
Let $R$ be an SRS over an alphabet $\Sigma$ and $u,v\in\Sigma^*$: 
If $[u]\cycto_R [v]$ then there exist  $u',v'\in\Sigma^*$ {s.t.}~$u \sim u' \prefixrewrite_R v' \sim v$.
\end{lemma}

\begin{proposition}\label{prop:N-shifts-are-sufficient-rotate}
Let $R$ be an SRS. If $\cycto_R$ is non-terminating, then
$\sim . \prefixrewrite_R$ admits an infinite reduction.
\end{proposition}
\begin{proof}
Let $[w_1] \cycto_R [w_2] \cycto_R \cdots$ 
be an infinite cycle rewrite sequence.
By Lemma~\ref{lem:Rightarrow-simulates-sim-rotate} we have
$w_1 \sim u_1 \prefixrewrite_R v_1 \sim w_2 \sim u_2 \prefixrewrite_R v_2 \sim \cdots$,
{i.e.}~(by joining the $\sim$-steps) we get
$v_0 = w_1 \sim u_1 \prefixrewrite_R v_1 \sim u_2 \prefixrewrite_R v_2 \sim \cdots$
and thus $v_0$ admits an infinite reduction of the required form.\qedhere
\end{proof}

For an SRS $R$, the SRS $\sym{rotate}(R)$ will encode the relation
$\sim.\prefixrewrite_R$ where extra symbols are used to separate the steps, 
and copies of the alphabet underlying $R$ are used to ensure completeness of the transformations. 
We first introduce the transformation $\sym{rotate}(R)$ and then explain the ideas of the transformation in detail.

Again for the rest of the section, we fix an SRS $R$ over alphabet $\Sigma_A = \{a_1,\ldots,a_n\}$,
write $\Sigma_B, \Sigma_C, \Sigma_D, \Sigma_E$ for fresh copies of the alphabet $\Sigma_A$,
and use $\trans{w}{Y}$ to switch between the alphabets (see Section~\ref{subsec:shift}).

\begin{definition}[The transformation $\sym{rotate}$]\label{def:rotate}
Let $R$ be an SRS over alphabet $\Sigma_A$.
The SRS $\sym{rotate}(R)$ over the alphabet
$\Sigma_A \cup \Sigma_B \cup \Sigma_C \cup \Sigma_D \cup \Sigma_E \cup \{
\BEGIN,\END,\REWRITE,\GORIGHT,\GUESS,\ROTATE,\CUT, 
\MOVELEFT,
\WAIT, 
\FINISH,
\FINISHTWO\}$ (where 
\BEGIN, 
\END, 
\REWRITE, 
\GORIGHT, 
\GUESS, 
\ROTATE, 
\CUT, 
\MOVELEFT, 
\WAIT, 
\FINISH, and
\FINISHTWO\ 
are fresh for $\Sigma_A \cup\Sigma_B \cup\Sigma_C \cup\Sigma_D \cup\Sigma_E$) is:
\begin{flushleft}
\begin{minipage}{.53\textwidth}
\begin{align}
\BEGIN\?\END &\to \REWRITE\?\END  \label{rotate-0}\tag{rotA}
\\
\BEGIN\?a &\to \ROTATE\?\CUT\trans{a}{D}\?\GUESS &&\hspace*{-1.1ex}\text{for all $a \in \Sigma_A$} \label{rotate-1}\tag{rotB}
\\
\GUESS\?a &\to \trans{a}{D}\?\GUESS&&\hspace*{-1.1ex}\text{for all $a \in \Sigma_A$}\label{rotate-2}\tag{rotC}
\\
\GUESS\?a &\to \MOVELEFT\trans{a}{C}\?\WAIT &&\hspace*{-1.1ex}\text{for all $a \in \Sigma_A$}\label{rotate-3}
\tag{rotD}
\\
\GUESS\?\END &\to \FINISH\?\END\label{rotate-4}\tag{rotE}
\\
d\?\MOVELEFT\?c &\to \MOVELEFT\?c\trans{d}{B} &&\hspace*{-1.1ex}\text{for all $c\in\Sigma_C,d\in\Sigma_D$}\label{rotate-5}\tag{rotF}
\\
\CUT\?\MOVELEFT\?c &\to \trans{c}{E}\?\CUT\?\GORIGHT &&\hspace*{-1.1ex}\text{for all $c \in \Sigma_C$}\label{rotate-6}\tag{rotG}
\\
\GORIGHT\?b &\to \trans{b}{D}\?\GORIGHT &&\hspace*{-1.1ex}\text{for all $b \in \Sigma_B$}\label{rotate-7}\tag{rotH}
\end{align}
\end{minipage}\quad~
\begin{minipage}{.45\textwidth}
\begin{align}
\GORIGHT\?\WAIT\?a &\to \MOVELEFT\trans{a}{C}\?\WAIT &&\hspace*{-1.1ex}\text{for all $a \in \Sigma_A$}\label{rotate-8}\tag{rotI}
\\
\GORIGHT\?\WAIT\?\END &\to \FINISH\?\END\label{rotate-9}\tag{rotJ}
\\
d\?\FINISH &\to \FINISH\trans{d}{A} &&\hspace*{-1.1ex}\text{for all $d \in \Sigma_D$}\label{rotate-10}\tag{rotK}
\\
\CUT\?\FINISH &\to \FINISHTWO\label{rotate-11}\tag{rotL}
\\
e\?\FINISHTWO &\to \FINISHTWO\trans{e}{A} &&\hspace*{-1.1ex}\text{for all $e \in \Sigma_E$}\label{rotate-10-2}\tag{rotM}
\\
\ROTATE\?\FINISHTWO &\to \REWRITE\label{rotate-12}\tag{rotN}
\\
\REWRITE\?\ell &\to \BEGIN\? r &&\hspace*{-1.1ex}\text{for all  $(\ell \to r) \in R$} \label{rotate-13}\tag{rotO}
\end{align}
\end{minipage}
\end{flushleft}
\end{definition}\bigskip

\noindent We describe the intended use of the rewrite rules, where we ignore the copies of the alphabet in our explanations.
The goal is that for any string $w \in \Sigma_A^*$, the string $\BEGIN\?w\?\END$ is 
rewritten to $\BEGIN\?u\?\END$, where $w \sim.\prefixrewrite_R u$.
The prefix rewrite step is performed by the last rule $\eqref{rotate-13}$. All other rules perform the rotation $\sim$
{s.t.}~$\BEGIN\?w\?\END$ is rewritten into $\REWRITE\?v\?\END$ where $w \sim v$.
This is done by moving a suffix of the string in front of the string.

The first rewrite rule \eqref{rotate-0} covers the case that $w$ is empty.
Otherwise, if $w = a_1\ldots a_n$, then first choose a position to cut the string into $w_1\?w_2$ (the goal is then to form the string $w_2\?w_1$). The symbol $\GUESS$ is used for the non-deterministic selection of the position. Rule \eqref{rotate-1} starts the rotate phase and the guessing, rule \eqref{rotate-2} shifts the $\GUESS$-marker and rule \eqref{rotate-3} stops the guessing. Rule \eqref{rotate-4} covers the case that $w_2=\varepsilon$ and no rotation will be performed.
After stopping the guessing, every symbol of $w_2$ is moved in front of $w_1$, resulting in $w_2\?w_1$.
A typical situation is $a_{k+1}\ldots a_m\?a_1 \ldots a_k a_{m+1} \ldots a_n$ and now the symbol $a_{m+1}$ must be moved in between $a_m$ and $a_1$. To keep track of the position of $a_1$, the symbol $\CUT$ (inserted in front of $a_1$)  marks the original beginning, and to keep track of the position of $a_k$, the  symbol $\WAIT$  (inserted after $a_k$) marks this position. The symbol $\MOVELEFT$ guards the movement of $a_{m+1}$ (by rule \eqref{rotate-5}). When arriving at the right place (rule \eqref{rotate-6}), the symbol $\GORIGHT$ is used to walk along the string (rule \eqref{rotate-7}) to find the next symbol which has to be moved (rule \eqref{rotate-8}). If all symbols are moved, rule \eqref{rotate-9} is applied to start the clean-up phase. There the symbols $\FINISH$ and $\FINISHTWO$ are used to remove the markers and to replace the copied symbols of the alphabet with the original ones (rules \eqref{rotate-10} -- \eqref{rotate-12}).

In the following, we denote with $\rotatesrs$ the string rewrite system consisting of the rules
$\eqref{rotate-0}$ -- $\eqref{rotate-12}$ from Definition~\ref{def:rotate} (excluding
the rule \eqref{rotate-13}).

\begin{example}
For the system $R = \{a \? a \? a \to a \? b \? a \? b \? a\}$ and 
$\Sigma_A=\{a,b\}$, the transformed string rewrite system is
${\sym{rotate}}(R) = {{\rotatesrs} \cup 
{\{
\REWRITE \? a \? a \? a 
\to
\BEGIN\? a \? b \? a \? b \? a
\}}}$. 
The cycle rewrite step 
$[a\?b\?b\?a\?a] \cycto_{R_1} [a\?b\?a\?b\?a\?b\?b]$ 
is simulated in the system ${\sym{rotate}}(R)$ by rewriting 
$\BEGIN\?a\?b\?b\?a\?a\?\END$ into 
$\BEGIN\?a\?b\?a\?b\?a\?b\?b\END$ as follows:
\begin{align*}
&  \Redex{$\BEGIN\?a$}\?b\?b\?a\?a\?\END
\xrightarrow{\!\tiny\ref{rotate-1}\!\!}
  \ROTATE\?\CUT\trans{a}{D}\Redex{$\GUESS\?b$}\?b\?a\?a\?\END
\xrightarrow{\!\tiny\ref{rotate-2}\!\!}
  \ROTATE\?\CUT\trans{a}{D}\trans{b}{D}\Redex{$\GUESS\?b$}\?a\?a\?\END
\xrightarrow{\!\tiny\ref{rotate-2}\!\!}
  \ROTATE\?\CUT\trans{a}{D}\trans{b}{D}\trans{b}{D}\Redex{$\GUESS\?a$}\?a\?\END
\displaybreak[3]\\&
\!\quad\xrightarrow{\!\tiny\ref{rotate-3}\!\!}
  \ROTATE\?\CUT\trans{a}{D}\trans{b}{D}\Redex{$\trans{b}{D}\MOVELEFT\trans{a}{C}$}\WAIT\?a\?\END
\xrightarrow{\!\tiny\ref{rotate-5}\!\!}
  \ROTATE\?\CUT\trans{a}{D}\Redex{$\trans{b}{D}\MOVELEFT\,\trans{a}{C}$}\trans{b}{B}\WAIT\?a\?\END
\xrightarrow{\!\tiny\ref{rotate-5}\!\!}
  \ROTATE\?\CUT\?\Redex{$\trans{a}{D}\MOVELEFT\,\trans{a}{C}$}\trans{b}{B}\trans{b}{B}\WAIT\?a\?\END
\displaybreak[3]\\&
\!\quad\xrightarrow{\!\tiny\ref{rotate-5}\!\!}
  \ROTATE\?\Redex{$\CUT\?\MOVELEFT\trans{a}{C}$}\trans{a}{B}\trans{b}{B}\trans{b}{B}\WAIT\?a\?\END
\xrightarrow{\!\tiny\ref{rotate-6}\!\!}
  \ROTATE\trans{a}{E}\CUT\?\Redex{$\GORIGHT\trans{a}{B}$}\trans{b}{B}\trans{b}{B}\WAIT\?a\?\END
\xrightarrow{\!\tiny\ref{rotate-7}\!\!}
  \ROTATE\trans{a}{E}\CUT\trans{a}{D}\Redex{$\GORIGHT\trans{b}{B}$}\trans{b}{B}\WAIT\?a\?\END
\displaybreak[3]\\&
\!\quad\xrightarrow{\!\tiny\ref{rotate-7}\!\!}
  \ROTATE\trans{a}{E}\CUT\trans{a}{D}\trans{b}{D}\Redex{$\GORIGHT\trans{b}{B}$}\WAIT\?a\?\END
\xrightarrow{\!\tiny\ref{rotate-7}\!\!}
  \ROTATE\trans{a}{E}\CUT\trans{a}{D}\trans{b}{D}\trans{b}{D}\Redex{$\GORIGHT\?\WAIT\?a$}\?\END
\displaybreak[3]\\&
\!\quad
\xrightarrow{\!\tiny\ref{rotate-8}\!\!}
  \ROTATE\trans{a}{E}\CUT\trans{a}{D}\trans{b}{D}\Redex{$\trans{b}{D}\MOVELEFT\trans{a}{C}$}\WAIT\?\END
\xrightarrow{\!\tiny\ref{rotate-5}\!\!}
  \ROTATE\trans{a}{E}\CUT\trans{a}{D}\Redex{$\trans{b}{D}\MOVELEFT\trans{a}{C}$}\trans{b}{B}\WAIT\?\END
\displaybreak[3]\\&
\!\quad
 \xrightarrow{\!\tiny\ref{rotate-5}\!\!}
  \ROTATE\trans{a}{E}\CUT\?\Redex{$\trans{a}{D}\MOVELEFT\trans{a}{C}$}\trans{b}{B}\trans{b}{B}\WAIT\?\END
\xrightarrow{\!\tiny\ref{rotate-5}\!\!}
  \ROTATE\trans{a}{E}\Redex{$\CUT\?\MOVELEFT\trans{a}{C}$}\trans{a}{B}\trans{b}{B}\trans{b}{B}\WAIT\?\END
\displaybreak[3]\\&
\!\quad
  \xrightarrow{\!\tiny\ref{rotate-6}\!\!}
  \ROTATE\trans{a}{E}\trans{a}{E}\CUT\?\Redex{$\GORIGHT\trans{a}{B}$}\trans{b}{B}\trans{b}{B}\WAIT\?\END
\xrightarrow{\!\tiny\ref{rotate-7}\!\!}
  \ROTATE\trans{a}{E}\trans{a}{E}\CUT\trans{a}{D}\Redex{$\GORIGHT\trans{b}{B}$}\trans{b}{B}\WAIT\?\END
\displaybreak[3]\\&
\!\quad
 \xrightarrow{\!\tiny\ref{rotate-7}\!\!}
  \ROTATE\trans{a}{E}\trans{a}{E}\CUT\trans{a}{D}\trans{b}{D}\Redex{$\GORIGHT\trans{b}{B}$}\WAIT\?\END
% 
% \\
\xrightarrow{\!\tiny\ref{rotate-7}\!\!}
  \ROTATE\trans{a}{E}\trans{a}{E}\CUT\trans{a}{D}\trans{b}{D}\trans{b}{D}\Redex{$\GORIGHT\?\WAIT\?\END$}
\displaybreak[3]\\&
\!\quad
\xrightarrow{\!\tiny\ref{rotate-9}\!\!}
  \ROTATE\trans{a}{E}\trans{a}{E}\CUT\trans{a}{D}\trans{b}{D}\Redex{$\trans{b}{D}\FINISH$}\END
\xrightarrow{\!\tiny\ref{rotate-10}\!\!}
  \ROTATE\trans{a}{E}\trans{a}{E}\CUT\trans{a}{D}\Redex{$\trans{b}{D}\FINISH$}\?b\?\END
\xrightarrow{\!\tiny\ref{rotate-10}\!\!}
  \ROTATE\trans{a}{E}\trans{a}{E}\CUT\?\Redex{$\trans{a}{D}\FINISH$}\?b\?b\?\END
\displaybreak[3]\\&
\!\quad
\xrightarrow{\!\tiny\ref{rotate-10}\!\!}
  \ROTATE\trans{a}{E}\trans{a}{E}\Redex{$\CUT\?\FINISH$}\?a\?b\?b\?\END
\xrightarrow{\!\tiny\ref{rotate-11}\!\!}
  \ROTATE\trans{a}{E}\Redex{$\trans{a}{E}\FINISHTWO$}\?a\?b\?b\?\END
\xrightarrow{\!\tiny\ref{rotate-11}\!\!}
  \ROTATE\?\Redex{$\trans{a}{E}\FINISHTWO$}\?a\?a\?b\?b\?\END
\xrightarrow{\!\tiny\ref{rotate-11}\!\!}
  \Redex{$\ROTATE\?\FINISHTWO$}\?a\?a\?a\?b\?b\?\END
\displaybreak[3]\\&
\!\quad
\xrightarrow{\!\tiny\ref{rotate-12}\!\!}
  \Redex{$\REWRITE\?a\?a\?a$}\?b\?b\?\END
\xrightarrow{\!\tiny\ref{rotate-13}\!\!}
  \BEGIN\?\?a\?b\?a\?b\?a\?b\?b\?\END
\end{align*}\end{example}

\begin{proposition}\label{prop:rotate-rotates}
If $u \sim  v$ then $\BEGIN\?u\?\END \torot^* \REWRITE\?v\?\END$.
\end{proposition}
\begin{proof}
If $u = \varepsilon$ then $v = \varepsilon$ and $\BEGIN\?\END \torot \REWRITE\?\END$.
If $u = a_1\?\ldots\?a_n$ for some $n \geq 1$ and $v = a_k \ldots a_n a_1 \ldots a_{k-1}$ then there are two cases:
If $k=1$ ({i.e.}~$u = v$) then the following rewrite sequence for $\BEGIN\?u\?\END$ exists:
$$
\begin{array}{l}
\BEGIN\?a_1 \ldots a_n\?\END 
\torot    \ROTATE\?\CUT\trans{a_1}{D}\?\GUESS\?a_2 \ldots a_n\?\END
\torot^* \ROTATE\?\CUT\trans{a_1}{D} \ldots \trans{a_n}{D} \GUESS\?\END\hspace*{3cm}\\
\multicolumn{1}{r}{
\torot   \ROTATE\?\CUT\trans{a_1}{D}\ldots \trans{a_n}{D}\?\FINISH\?\END
\torot^* \REWRITE\?a_1\ldots a_n\?\END}
\end{array}
$$
If $1 < k \leq n$, then the following rewrite sequence for $\BEGIN\?u\?\END$ exists:
$$
\begin{array}{@{}l@{}}
          \BEGIN\?a_1 {\ldots} a_n\?\END 
\torot    \ROTATE\?\CUT\trans{a_1}{D}\GUESS\?a_2 {\ldots} a_n\?\END
\torot^*\!\ROTATE\?\CUT\trans{a_1}{D}{\ldots} \trans{a_{k-1}}{D}\GUESS\?a_{k} {\ldots} a_n\?\END\\
\multicolumn{1}{@{}r@{}}{
\torot  \ROTATE\?\CUT\trans{a_1}{D}{\ldots} \trans{a_{k-1}}{D}\MOVELEFT\trans{a_{k}}{C}\WAIT\?a_{k+1}{\ldots} a_n\?\END
\torot^*\!\ROTATE\trans{a_{k}}{E} {\ldots} \trans{a_n}{E}\CUT\trans{a_1}{D}{\ldots} \trans{a_{k-1}}{D}\GORIGHT\?\WAIT\?\END
}
\\
\multicolumn{1}{@{}r@{}}{
\torot  \ROTATE\trans{a_{k}}{E} {\ldots} \trans{a_n}{E}\CUT\trans{a_1}{D}{\ldots} \trans{a_{k-1}}{D}\FINISH\?\END
\torot^*\!\REWRITE\?a_{k}{\ldots} a_n\?a_1{\ldots} a_{k-1}\?\END %\qedhere
}
\end{array}
$$
\end{proof}

\begin{theorem}\label{theo:rotate-sound}
The transformation $\sym{rotate}$ is sound.
\end{theorem}
\begin{proof}
Proposition~\ref{prop:rotate-rotates} and the construction of $\sym{rotate}(R)$ show
that whenever
$u \sim v \prefixrewrite_R w$ then also $\BEGIN\?u\?\END \torotate^* \BEGIN\?w\?\END$.
Now Proposition~\ref{prop:N-shifts-are-sufficient-rotate} implies soundness.
\end{proof}
\subsubsection{Completeness of Rotate}\label{subsubsec:rotate-complete}
We write $\widehat{\Sigma_A}$ for the extension of alphabet $\Sigma_A$ by the 
fresh symbols
\BEGIN,\END,\REWRITE,\GORIGHT,\GUESS,\ROTATE,\CUT,\MOVELEFT,\WAIT,\FINISH,\FINISHTWO, and by four copies $\Sigma_B, \Sigma_C, \Sigma_D, \Sigma_E$ of the alphabet $\Sigma_A$.

For proving completeness of the transformation $\sym{rotate}$,
we use type introduction. Let $\Types := \{\typeConst,\typeA,\typeB,\typeC,\typeD,\typeE,\typeT\}$
and let the symbols in $\widehat{\Sigma_A}$ be typed as follows:
$$\begin{array}{c@{\qquad}c@{\qquad}c}
\begin{array}{l@{~}l@{~}l}
a                     & : \typeA \to \typeA   & \text{for all $a \in \Sigma_A$}\\
b                     & : \typeB \to \typeB   & \text{for all $b \in \Sigma_B$}\\
c                     & : \typeB \to \typeC   & \text{for all $c \in \Sigma_C$}\\
d                     & : \typeD \to \typeD   & \text{for all $d \in \Sigma_D$}\\
e                     & : \typeE \to \typeE  & \text{for all $e \in \Sigma_E$}\\
\end{array}
&
\begin{array}{l@{~}l@{~}l}
\END          & : \typeConst \to \typeA\\
\WAIT         & : \typeA \to \typeB\\
\MOVELEFT     & : \typeC \to \typeD\\
\GORIGHT      & : \typeB \to \typeD\\
\CUT          & : \typeD \to \typeE\\
\end{array}
&
\begin{array}{l@{~}l@{~}l}
\GUESS,\FINISH       & : \typeA \to \typeD\\
\BEGIN,\REWRITE      & : \typeA \to \typeT\\
\ROTATE       & : \typeE \to \typeT\\
\FINISHTWO      & : \typeA \to \typeE
\end{array}
\end{array}
$$
One can verify that with definition $R$ is indeed a typed SRS:
rule \eqref{rotate-0} rewrites strings of type $\typeConst \to \typeT$,
rules 
\eqref{rotate-1}, 
\eqref{rotate-12}, and
\eqref{rotate-13}   rewrite strings of type $\typeA \to \typeT$,
rules 
\eqref{rotate-2},
\eqref{rotate-3},
\eqref{rotate-8}, and 
\eqref{rotate-10}  rewrite strings of type $\typeA \to \typeD$,
rules
\eqref{rotate-4}  and 
\eqref{rotate-9} rewrite strings of type $\typeConst \to \typeD$,
rules 
\eqref{rotate-5} and \eqref{rotate-7} rewrite strings of type $\typeB \to \typeD$,
rule \eqref{rotate-6} rewrites strings of type $\typeB \to \typeE$, and 
rules \eqref{rotate-11} and
\eqref{rotate-10-2} rewrite strings of type $\typeA \to \typeE$.
Thus, by Corollary \ref{cor:typeintro-term} type introduction can be used, {i.e.}~$\sym{rotate}(R)$ is terminating
if, and only if the typed system terminates (and analogously for relative termination, see Theorem~\ref{theo:typeintro-rel}).

\begin{lemma}\label{lemm:rotate-t6-suff}
Let $w \in \widehat{\Sigma_A}^*$ be a well-typed string, {s.t.}~$w$ admits an infinite reduction {w.r.t.}~$\torotate$.
Then there exists a well-typed string $w' \in \widehat{\Sigma_A}^*$ of type $\typeConst \to \typeT$ which admits
an infinite reduction {w.r.t.}~$\torotate$.
\end{lemma}
\begin{proof}
W.l.o.g.\ we can assume that $w \not=\varepsilon$ (since $\varepsilon$ is irreducible w.r.t.\ $\torotate$).
If $w$ is not of type $\typeConst \to \typeT$, then we can prepend  and append symbols to $w$ constructing a string $u\?w\?v$ of type $\typeConst \to \typeT$,
since for any type $\tau  \in \{\typeA,\typeB,\typeC,\typeD,\typeE,\typeConst,\typeT\}$ there  are the following sequences of 
type $\tau \to \typeT$ where $\tau \not = \typeT$ (used for $u$)  and of type $\typeConst \to \tau$ where $\tau\not=\typeConst$ (used for $v$),
where $c\in\Sigma_C$:
$$
\begin{array}{@{}l@{\,}c@{\,}l@{\quad}l@{\,}c@{\,}l@{\quad}l@{\,}c@{\,}l@{\quad}l@{\,}c@{\,}l@{\quad}l@{\,}c@{\,}l@{\quad}l@{\,}c@{\,}l@{}}
\BEGIN &:& \typeA \to \typeT
&
\ROTATE\?\CUT\?\MOVELEFT\?c &:& \typeB \to \typeT
&
\ROTATE\?\CUT\?\MOVELEFT &:& \typeC \to \typeT
&
\ROTATE\?\CUT &:& \typeD \to \typeT
&
\ROTATE &:& \typeE \to \typeT
&
\BEGIN\?\END &:&\typeConst \to \typeT
\\
%%%
\END &:& \typeConst \to \typeA
&
\WAIT\?\END &:& \typeConst \to \typeB 
&
c\?\WAIT\?\END&:& \typeConst \to \typeC
&
\MOVELEFT\?c\?\WAIT\?\END &:& \typeConst \to \typeD
&
\sym{f}\?\END &:& \typeConst \to \typeE
&
\end{array}
$$
Clearly, the infinite reduction for $w$ can be used to construct an infinite reduction for $uwv$.
\end{proof}

By inspecting the typing the following characterization of terms of type $\typeT$ holds:

\begin{lemma}\label{lem:cases-terms-sort-T6}
Any well-typed string of type $\typeConst \to \typeT$ is of one of the following forms:
\begin{enumerate}
\item 
$\ROTATE\?w_E\?\CUT\?w_D\?\MOVELEFT\?c\?w_B\?\WAIT\?w_A\?\END$
where $w_E\in\Sigma_E^*$, $w_D \in \Sigma_D^*$, $c \in \Sigma_C$, $w_B \in \Sigma_B^*$, and $w_A \in \Sigma_A^*$.
\item 
$\ROTATE\?w_E\?\CUT\?w_D\?\GORIGHT\?w_B\?\WAIT\?w_A\?\END$
where $w_E\in\Sigma_E^*$, $w_D \in \Sigma_D^*$, $w_B \in \Sigma_B^*$, and $w_A \in \Sigma_A^*$.
\item 
$\ROTATE\?w_E\?\CUT\?w_D\?\FINISH\?w_A\?\END$
where $w_E\in\Sigma_E^*$, $w_D \in \Sigma_D^*$, and $w_A \in \Sigma_A^*$.
\item 
$\ROTATE\?w_E\?\CUT\?w_D\?\GUESS\?w_A\?\END$
where $w_E\in\Sigma_E^*$, $w_D \in \Sigma_D^*$, and $w_A \in \Sigma_A^*$.
\item 
$\ROTATE\?w_E\?\FINISHTWO\?w_A\?\END$
where $w_E\in\Sigma_E^*$ and $w_A \in \Sigma_A^*$.
\item 
$\BEGIN\?w_A\?\END$
where $w_A \in \Sigma_A^*$.
\item 
$\REWRITE\?w_A\?\END$
where $w_A \in \Sigma_A^*$.
\end{enumerate}
\end{lemma}
Informally, the parts of strings in the lemma
above can be describes as follows:
for a string $w = w_1w_2$,
$w_E$ is a prefix of $w_2$ that has been 
rotated,    while $w_A$ is a suffix of $w_2$
that has to be rotated, and  $w_B$, $w_D$ 
are parts of $w_1$ and symbol $c$ is
the symbol that is currently processed.

\begin{definition}
For well-typed strings  $w : \typeConst \to \typeT$, the mapping $\mapRotate(w) \in {\mathbb{P}}(\Sigma_A^*)$ (where $\mathbb{P}$ denotes the power set) is defined according to the cases of 
Lemma~\ref{lem:cases-terms-sort-T6} as follows:
\begin{enumerate}
\item 
$\mapRotate(\ROTATE\?w_E\?\CUT\?w_D\?\MOVELEFT\?c\?w_B\?\WAIT\?w_A\?\END)
:=  \{\trans{w_D}{A}\trans{w_B}{A}\trans{w_E}{A}\trans{c}{A}\?w_A\} 
$
\item 
$\mapRotate(\ROTATE\?w_E\?\CUT\?w_D\?\GORIGHT\?w_B\?\WAIT\?w_A\?\END)
:= \{\transA{w_D}\transA{w_B}\transA{w_E}\?w_A\}$

\item 
$\mapRotate(\ROTATE\?w_E\?\CUT\?w_D\?\FINISH\?w_A\?\END)
:= \{\transA{w_D}\?w_A\transA{w_E}\}$

\item
$\mapRotate(\ROTATE\?w_E\?\CUT\?w_D\?\GUESS\?w_A\?\END)
:= \{\transA{w_D}\?w_A'\transA{w_E}\?w_A'' \mid \text{for all } w_A' w_A'' = w_A\}$
\item 
$\mapRotate(\ROTATE\?w_E\?\FINISHTWO\?w_A\?\END) :=\{\transA{w_E}\?w_A\}$
\item 
$\mapRotate(\BEGIN\?w_A\?\END) := \{w_A\}$
\item 
$\mapRotate(\REWRITE\?w_A\?\END) := \{w_A\}$
\end{enumerate}
\end{definition}

\begin{lemma}\label{lem:mapping-rotate}
Let $w,w':\typeConst \to \typeT$ with $w \torot w'$.
Then for any $u' \in \mapRotate(w')$ there exists $u \in \mapRotate(w)$ {s.t.}~$u \sim u'$. 
\end{lemma}
\begin{proof}
We go through the cases for $w$ according to Lemma~\ref{lem:cases-terms-sort-T6}:
\begin{enumerate}
\item If $w = \ROTATE\?w_E\?\CUT\?w_D\?\MOVELEFT\?c\?w_B\?\WAIT\?w_A\?\END$,
then rules \eqref{rotate-5} and \eqref{rotate-6} may be applied.
If rule \eqref{rotate-5} is applied, then with $w_D = w_D' d$
we have $w' = \ROTATE\?w_E\?\CUT\?w_D'\?\MOVELEFT\?c\trans{d}{B}\?w_B\?\WAIT\?w_A\?\END$,
and since $\map(w) = \map(w')$, the claim holds.
If rule \eqref{rotate-6} is applied then $w_D = \varepsilon$,
$w' = \ROTATE\?w_E\?e\?\CUT\?\GORIGHT\?w_B\?\WAIT\?w_A\?\END$
and since $\map(w) = \map(w')$, the claim holds.

\item  If $w= \ROTATE\?w_E\?\CUT\?w_D\?\GORIGHT\?w_B\?\WAIT\?w_A\?\END$, then 
rules \eqref{rotate-7}, \eqref{rotate-8}, and \eqref{rotate-9} may be applied.
If rule \eqref{rotate-7} is applied then with $w_B = b w_B'$
we have $w' = \ROTATE\?w_E\?\CUT\?w_D\trans{b}{D}\GORIGHT\?w_B'\?\WAIT\?w_A\?\END$,
and since $\map(w) = \map(w')$, the claim holds.
If rule \eqref{rotate-8} is applied, then $w_B = \varepsilon$ and with $w_A = a\?w_A'$
we have $w' = \ROTATE\?w_E\?\CUT\?w_D\?\MOVELEFT\trans{a}{C}\?\WAIT\?w_A'\?\END$,
and since $\map(w) = \map(w')$, the claim holds.
If rule \eqref{rotate-9} is applied, then $w_A = w_B = \varepsilon$
and  $w' = \ROTATE\?w_E\?\CUT\?w_D\?\FINISH\END$ and since
$\map(w) = \map(w')$, the claim holds.

\item  If $w = \ROTATE\?w_E\?\CUT\?w_D\?\FINISH\?w_A\?\END$, then
rules \eqref{rotate-10} and \eqref{rotate-11}  may be applied.
If rule \eqref{rotate-10} is applied, then with $w_D = w_D' d$
we have
 $w' = \ROTATE\?w_E\?\CUT\?w_D'\?\FINISH\trans{d}{A}\?w_A\?\END$, and
since $\map(w) = \map(w')$, the claim holds.
If  rule \eqref{rotate-11} is applied, then $w_D = \varepsilon$,
 $w' = \ROTATE\?w_E\?\FINISHTWO\?w_A\?\END$, and
since $\map(w) = \{w_A \transA{w_E}\}$, $\map(w') = \{\transA{w_E} w_A\}$, 
and $w_A \transA{w_E} \sim \transA{w_E} w_A$, the claim holds.

\item 
If $w= \ROTATE\?w_E\?\CUT\?w_D\?\GUESS\?w_A\?\END$, then
rules \eqref{rotate-2}, \eqref{rotate-3}, and \eqref{rotate-4} may be applied.
If rule \eqref{rotate-2} is applied, then with $w_A = a v_A$
we have $w' = \ROTATE\?w_E\?\CUT\?w_D\trans{a}{D}\?\GUESS\?v_A\?\END$,
and since $\map(w) \supseteq \map(w') \not=\emptyset$, the claim holds.
If rule \eqref{rotate-3} is applied, then with $w_A = a\? v_A$
we have $w' =\ROTATE\?w_E\?\CUT\?w_D\?\MOVELEFT\trans{a}{C}\?\WAIT\?v_A\?\END$,
and since $\map(w) \supseteq \map(w') \not=\emptyset$, the claim holds.
If rule \eqref{rotate-4} is applied, then $w_A = \varepsilon$ and 
$w'= \ROTATE\?w_E\?\CUT\?w_D\?\FINISH\?\END$, and since $\map(w) = \map(w')$, the claim holds.
\item If $w=\ROTATE\?w_E\?\FINISHTWO\?w_A\?\END$, then
rules \eqref{rotate-10-2} and \eqref{rotate-12}  may be applied:
If rule \eqref{rotate-10-2} is applied, then with $w_E = w_E' e$
we have $w'=\ROTATE\?w_E'\?\FINISHTWO\trans{e}{A}\?w_A\?\END$,
and since $\map(w) = \map(w')$, the claim holds.
If rule \eqref{rotate-12} is applied then  $w_E = \varepsilon$,
$w' = \REWRITE\?w_A\?\END$, and since $\map(w) = \map(w')$, the claim holds.
\item If $w = \BEGIN\?w_A\?\END$, then
rules \eqref{rotate-0} or \eqref{rotate-1} may be applied:
If rule \eqref{rotate-0} is applied then $\map(w) = \{\varepsilon\} = \map(w')$ and thus the claim holds.
If rule \eqref{rotate-1} is applied, then with $w_A = a w_A'$ we have
$w' = \ROTATE\?\CUT\trans{a}{D}\?\GUESS\?w_A'$, and since $\map(w) = \map(w')$ the claim holds.
\item If $w= \REWRITE\?w_A\?\END$. then no rule is applicable.\qedhere
\end{enumerate}
\end{proof}%\newpage

\begin{lemma}\label{lemm:rotate:terminating}
The SRS $\rotatesrs$ is terminating.
\end{lemma}
\begin{proof}
Let $\sigma$ be the following polynomial interpretation:
$$
\begin{array}{@{}l@{~}l@{~}l@{\qquad}l@{~}l@{~}l@{\qquad}l@{~}l@{~}l@{\qquad}l@{~}l@{~}l@{}}
\sigma(\REWRITE) &=& \lambda x.x
&\sigma(\GORIGHT) &=& \lambda x.2x
&\sigma(\FINISH) &=& \lambda x.2x +2
&\sigma(\FINISHTWO) &=& \lambda x.2x +2
\\
\sigma(\END) &=& \lambda x.x+1
&\sigma(\MOVELEFT) &=& \lambda x.2x+1
&\sigma(\WAIT) &=& \lambda x.3x
&\sigma(\ROTATE) &=& \lambda x.x+9\\
\sigma(\CUT) &=& \lambda x.2x
&\sigma(\GUESS) &=& \lambda x.6x+1
&\sigma(\BEGIN) &=&\multicolumn{4}{@{}l@{}}{\lambda x.12x+18}
\\
\sigma(h) &=& \multicolumn{9}{@{}l@{}}{\lambda x.4x +1 \text{ for all } h \in (\Sigma_A \cup \Sigma_B \cup \Sigma_C \cup \Sigma_D \cup \Sigma_E)}
\end{array}
$$
We verify that for every rule $(\ell \to r) \in \rotatesrs$ the inequation $\sigma(\ell) > \sigma(r)$ holds:
For rule \eqref{rotate-0}, we have  $\lambda x.12x+30 > \lambda x.x+1$,
for rule \eqref{rotate-1}, we have $\lambda x.48x+30 > \lambda x.48x+19$,
for rule \eqref{rotate-2}, we have $\lambda x.24x +7 > \lambda x.24x+5$,
for rule \eqref{rotate-3}, we have $\lambda x.24x+7 > \lambda x.24x+3$,
for rule \eqref{rotate-4}, we have $\lambda x.6x+7 > \lambda x.2x+4$,
for rule \eqref{rotate-5}, we have $\lambda x.32x+13 > \lambda x.32x+11$,
for rule \eqref{rotate-6}, we have $\lambda x.16x+6 > \lambda x.16x+1$,
for rule \eqref{rotate-7}, we have $\lambda x.8x +2 > \lambda x.8x+1$,
for rule \eqref{rotate-8}, we have $\lambda x.24x+6 > \lambda x.24x+3$,
for rule \eqref{rotate-9}, we have $\lambda x.6x+6 > \lambda x.2x+4$,
for rule \eqref{rotate-10}, we have $\lambda x.8x + 9 > \lambda  x. 8x +4$,
for rule \eqref{rotate-11}, we have $\lambda x.4x+4 > \lambda x.2x+2$, 
for rule \eqref{rotate-10-2}, we have $\lambda x.8x + 9 > \lambda  x. 8x +4$, and
for rule \eqref{rotate-12}, we have $\lambda x.2x+11 > \lambda x.x$.
\end{proof}

\begin{proposition}\label{prop:rotate-complete}
If $w : \typeConst \to \typeT$ admits an infinite $\torotate$-reduction, then there exists $u \in \mapRotate(w)$ {s.t.}~$[u]$ admits an infinite $\cycto_R$-reduction.
\end{proposition}
\begin{proof}
Let $w \torotate w_1 \torotate w_2 \torotate \ldots$ be an infinite $\torotate$-reduction.
This sequence must have infinitely many steps using rules \eqref{rotate-13},
since the rewrite system $\rotatesrs$ is terminating 
(Lemma~\ref{lemm:rotate:terminating}),
i.e.
$$w=v_0 
\torot^* v_1 
\to_{\eqref{rotate-13}} v_1'  
   \torot^* v_2 \to_{\eqref{rotate-13}} v_2'  
        \torot^* \ldots$$ 
We consider all sub-sequences $v_i \to_{\eqref{rotate-13}} v_i' \torot^* v_{i+1}$ for $i = 1,2,\ldots$.
Typing of all strings, 
Lemma~\ref{lem:mapping-rotate} and the definition of the rules $\eqref{rotate-13}$ 
imply that the steps $v_i \to_{\eqref{rotate-13}} v_i'$  
must be of the form $\REWRITE\?\ell\?w_A\END \to \BEGIN\?r\?w_A\END$ 
where $(\ell \to r) \in R$. Since
$\mapRotate(v_i) = \{\ell\?w_A\}$ and $\mapRotate(v_i') = \{r\?w_A\}$,
there exist (unique) strings $u_i\in\mapRotate(v_i)$
and $u_i' \in \mapRotate(v_i')$ 
(namely $u_i = \ell\?w_A$ and $u_i'=r\?w_A$)
{s.t.}~$u_i \to_R u_i'$.
Lemma~\ref{lem:mapping-rotate} shows that for any sub-sequence $v_i' \torot^* v_{i+1}$
and any $u_{i+1} \in\mapRotate(v_{i+1})$ there exists $u_i'' \in \mapRotate(v_i')$ with
$u_i'' \sim u_{i+1}$. Since $\mapRotate(v_i') = \{r\?w_A\} = \{u_i'\}$, the equality $u_i'' = u_i'$ must hold.
Thus, we have  $u_i \to_R u_i' \sim u_{i+1}$ where $\mapRotate(v_i) = \{u_i\}$ and $\mapRotate(v_{i+1}) = \{u_{i+1}\}$.
Since this holds for all $i=1,2,\ldots$, and since $\mapRotate(v_j)$ is a singleton for all $v_j$,
we can construct the infinite sequence $u_1 \to_R.\sim u_2 \to_R.\sim u_3 \to_R.\sim \cdots$.
This implies $[u_1]\cycto_R [u_2] \cycto_R [u_3] \ldots$ and thus $\cycto_R$ is non-terminating.
\end{proof}

\begin{theorem}\label{theo:rotate-sound-and-complete}
The transformation $\sym{rotate}$ is sound and complete.
\end{theorem}
\begin{proof}
Soundness is shown in Theorem~\ref{theo:rotate-sound}, completeness follows by 
type introduction (Corollary~\ref{cor:typeintro-term}), Lemma~\ref{lemm:rotate-t6-suff}, and Proposition~\ref{prop:rotate-complete}.
\end{proof}

\subsubsection{Relative Termination}
We provide a variant of the transformation $\sym{rotate}$ for relative termination:
\begin{definition}
Let $S \subseteq R$ be SRSs over an alphabet $\Sigma_A$.
The transformation $\sym{rotate}_{\mathit{rel}}$ is defined as:
$$
\begin{array}{lcl}
 \sym{rotate}_{\mathit{rel}}(S,R) &:=& 
 (\{\REWRITE \ell \to \BEGIN r \mid (\ell \to r) \in S\},\sym{rotate}(R))
\end{array}
$$
\end{definition}

We show soundness and completeness:
\begin{theorem}\label{theo:rotate-sound-and-complete-relative}
The transformation $\sym{rotate}_{\mathit{rel}}$ is sound and complete for relative termination.
\end{theorem}
\begin{proof}
Let $S \subseteq R$ and $\sym{rotate}_{\mathit{rel}}(S,R) = (S',R')$. 
Soundness of the transformation $\sym{rotate}$ implies (see proof of Theorem~\ref{theo:rotate-sound-and-complete}) that 
\begin{equation}
[u] \cycto_{R} [v] \text{ implies } \BEGIN u \END \to_{R'}^* \BEGIN w \END \label{rel-rotate-eq1}
\end{equation}
Since $\sym{rotate}(S) \subseteq R'$, this also shows:
\begin{equation}
[u] \cycto_{S} [v] \text{ implies } \BEGIN u \END \to_{R'}^*.\to_{S'} \BEGIN w \END. \label{rel-rotate-eq2}
\end{equation}
For proving soundness of $\sym{rotate}_{\mathit{rel}}$, let us assume that 
$S$ is not cycle terminating relative to $R$. 
Then there exists an infinite derivation 
$[w_1] \cycto_{R}^*.\cycto_S [w_2] \cycto_{R}^*.\cycto_S \cdots$.
Applying the implications \eqref{rel-rotate-eq1} and \eqref{rel-rotate-eq2} shows that
$\BEGIN w_i \END \to_{R'}^* . \to_{R'}^* . \to_{S'} \BEGIN w_{i+1} \END$ for $i=1,2,\ldots$.
This shows that $\BEGIN w_i \END \to_{R'}^*. \to_{S'} \BEGIN w_{i+1} \END$ for $i=1,2,\ldots$ and 
thus $S'$ is not string terminating relative to $R'$. 

For proving completeness of $\sym{rotate}_{\mathit{rel}}$, we again use the types introduced in Section~\ref{subsubsec:rotate-complete}.
Let us assume that $S'$ is not string terminating relative to $R'$.
By Theorem~\ref{theo:typeintro-rel}  and Lemma~\ref{lemm:rotate-t6-suff} there exists a well-typed string $w_0 : \typeConst \to \typeT$
{s.t.}~for all $i = 0,1,2\ldots$  the derivation $w_i \to_{S'}.\to_{R'}^* w_{i+1}$  exists.
Each such derivation can be written as 
$$\begin{array}{@{}l@{~}l@{}}
   w_i  &\to_{\rotatesrs}^* w_{0,0,i}
        \to_{\sym{\tiny\ref{rotate-13}}(S)} w_{0,1,i}\\
        &\to_{\rotatesrs}^* w_{1,0,i}
        \to_{\sym{\tiny\ref{rotate-13}}(R)} w_{1,1,i}\\
        &\ldots\\
        &\to_{\rotatesrs}^* w_{m_i,0,i}
        \to_{\sym{\tiny\ref{rotate-13}}(R)} w_{m_i,1,i}
        = w_{i+1}
  \end{array}$$
where $m_i \geq 0$. Starting with $w_{0,0,i}$ and appending the sequence 
$w_{m_i,1,i} \to_{\rotatesrs}^* w_{0,0,i+1}$
this can be formulated as follows:
There exists a derivation
$$\begin{array}{ll}
 w_{0,0,i}
        &\to_{\sym{\tiny\ref{rotate-13}}(S)} w_{0,1,i} \\
        &\to_{\rotatesrs}^* w_{1,0,i}
        \to_{\sym{\tiny\ref{rotate-13}}(R)} w_{1,1,i}
        \\
        &\ldots
        \\
        &\to_{\rotatesrs}^* w_{m_i,0,i} \to_{\sym{\tiny\ref{rotate-13}}(R)} w_{m_i,1,i}
        \to_{\rotatesrs}^* w_{0,0,i+1}
\end{array}
$$
for all $i=0,1,\ldots$ where $m_i \geq 0$. 
Typing of the strings and the cases in Lemma \ref{lem:cases-terms-sort-T6}  show
that $\mapRotate(w_{k,l,i})$ is a singleton $\{u_{k,l,i}\}$ for all $i$ and $k=0,\ldots,m_i$, $l=0,1$.
Due to applicability of rule $\eqref{rotate-13}$, we also have
$u_{0,0,i} \to_S u_{0,1,i}$ and $u_{k,0,i} \to_R u_{k,1,i}$ for $k = 1,\ldots,m_i$.
Lemma~\ref{lem:mapping-rotate} implies that  $u_{k,0,i} \sim u_{k-1,1,i}$ for $k=1,\ldots,m_i+1$.
Thus, we have
$$u_{0,0,i} \to_S u_{0,1,i} \sim u_{1,0,i} \to_R u_{1,1,i} \sim . \to_R \ldots \sim.\to_R u_{m_i,1,i} \sim u_{0,0,i+1}$$
which shows $[u_{0,0,i}] \cycto_S [u_{1,0,i}] \cycto_R^* [u_{0,0,i+1}]$.
Since this holds for all $i=0,1,\ldots$, we have shown that $S$ is not  cycle terminating relative to $R$.
\end{proof}

\section{Trace Decreasing Matrix Interpretations}\label{sec:matrix-interpretations}
\label{secmatr}
In this section we present a variant of matrix interpretations suitable for proving cycle termination.
The basics of matrix interpretations for string and term rewriting were presented in  \cite{HW06,HW06-rta,EWZ08,KW08}.
The special case of tropical and arctic matrix interpretations for cycle rewriting was presented in 
\cite{ZKB14}, in the setting of {\em type graphs}.  
Natural matrix interpretations for cycle rewriting
were presented in \cite{sabel-zantema:15}, inspired by an approach proposed by Johannes Waldmann. Here we 
present a self-contained uniform framework covering all these cases
as we show by subsequently instantiating the framework for three semi-rings.
At the end of the section we discuss and illustrate the approach by examples and
also discuss some limitations.

\subsection{A Uniform Framework for Matrix Interpretations}
Fix a commutative {\em semi-ring}, that is, 
a set $X$, a zero element $0 \in X$, a unit element $1 \in X$,
and two operations
$+,\times$ that are commutative and associative, and satisfying 
\[ x+0 = x, \; x\times 1 = x, \; x \times 0 = 0, \; x \times (y + z) = (x \times y) + (x \times z) \]
for all $x,y,z \in X$. Further we assume a well-founded order $>$ on $X$.

Next we fix a dimension $d > 0$ and choose $M$ to be a set of $d \times d$ matrices over $X$.
Multiplication of matrices is defined as usual, now using $+$ and $\times$ from the
semi-ring as basic operations on elements:

\[
(AB)_{ij} = \sum_{k=1}^d A_{ik} \times B_{kj} 
\]
for all $i,j = 1,\ldots,d$, where $\sum$ extends the operation $+$. 
Note that in the following we only require matrix multiplication and no addition.
Thus we 
assume that $M$ is closed under
multiplication and contains the identity matrix.

For a $d \times d$ matrix $A$ over $X$, its {\em trace} $\tr(A)$ is defined to be
$\sum_{i=1}^d A_{ii}$: the sum of its diagonal. 

On $M$ we assume relations $>$, $\geq$ satisfying 
\begin{align}
&A > B    \implies (AC > BC \wedge CA > CB)\label{assume1}\\
&A \geq B \implies (AC \geq BC \wedge CA \geq CB)\label{assume2}\\
&A > B    \implies \tr(A) > \tr(B)\label{assume3}\\
&A \geq B \implies \tr(A) \geq \tr(B)\label{assume4}
\end{align}
for all $A,B,C \in M$. Note the overloading of $>$ and $\geq$: when applied on matrices it refers to the
relations $>$, $\geq$ we assume on $M$, when applied to elements of $X$ it refers to the well-founded
order on $X$, where $x \geq y \desda x > y \vee x=y$. 

Note that on the one hand these requirements imply that
$>$ is irreflexive by $A > B \implies \tr(A) > \tr(B)$, and on the other hand the 
requirements imply that $M$ does not consist of all $d \times d$ matrices over $X$: for $A > B$ and $C$ being the matrix $0$ having
$0$ on all positions, we have $AC = 0 = BC$, violating $A > B \implies AC > BC$.

A {\em matrix interpretation} $\li \cdot \ri $  for a signature $\Sigma$ is defined to be a mapping from 
$\Sigma$ to $M$.
It is extended to $\li \cdot \ri  : \Sigma^* \to M$ by defining inductively
$ \li \varepsilon \ri  = I$ and  $\li u\?a \ri  = \li u \ri \times\li a \ri$
for all $u \in \Sigma^*$, $a \in \Sigma$, where $I$ is the identity matrix.

\begin{theorem}
\label{thmmatr}
Let $S \subseteq R$ be SRSs over $\Sigma$ and 
let $\li \cdot \ri  : \Sigma \to M$ satisfy the above properties and 
\begin{itemize}
\item $\li \ell \ri  > \li r \ri $ for all $(\ell \to r) \in S$, and
\item $\li \ell \ri  \geq \li r \ri $ for all $(\ell \to r) \in R \setminus S$
\end{itemize}
Then $\rto{S}$ is terminating relative to $\rto{R}$.
\end{theorem}
\begin{proof}
If $u \sim v$ then $u = u_1 u_2$ and $v = u_2 u_1$ for some $u_1,u_2$. Write $A = \li u_1 \ri$ and
$B = \li u_2 \ri$. Then 
\[ \tr( \li u \ri) = \tr(AB) =  \sum_{i=1}^d \sum_{k=1}^d A_{i,k} \times B_{k,i} = \tr(BA) = \tr( \li v \ri). \]
If $u \to_{R  \setminus S} v$ then there is a rule $(\ell \to r) \in R \setminus S$ and $x,y \in \Sigma^*$ such that 
$u = x \ell y$ and $v = x r y$, so 
\[\li u \ri = \li x \ri \li \ell \ri \li y \ri \geq \li x \ri \li r \ri \li y \ri = \li v \ri\]
using $A \geq B \implies (AC \geq BC \wedge CA \geq CB)$, hence $\tr(\li u \ri) \geq \tr(\li v \ri)$.

If $u \to_{S} v$ then there is a rule $(\ell \to r) \in S$ and $x,y \in \Sigma^*$ 
such that $u = x \ell y$ and $v = x r y$, so 
\[\li u \ri = \li x \ri \li \ell \ri \li y \ri > \li x \ri \li r \ri \li y \ri = \li v \ri\]
using $A > B \implies (AC > BC \wedge CA > CB)$, hence $\tr(\li u \ri) > \tr(\li v \ri)$.

Assume $\rto{S}$ is not terminating relative to $\rto{R}$,
then there are $u_i,v_i \in \Sigma^*$ for $i \in \nat$ such that
\[ u_1 \to_{R} v_1 \sim  u_2 \to_{R} v_2 \sim  u_3 \to_{R} v_3 \sim  \cdots ,\]
containing infinitely many steps $u_i \to_S v_i$. 
By the above observations we have $\tr(v_i) = \tr(u_{i+1})$ for all $i$,
$\tr(\li u_i \ri) > \tr(\li v_i\ri)$ for infinitely many $i$, and 
$\tr(\li u_i \ri) \geq \tr(\li v_i \ri)$ for all other $i$, yielding an infinite strictly decreasing sequence in $X$,
contradicting well-foundedness of the order $>$ on $X$.
\end{proof}

A typical way to apply Theorem \ref{thmmatr} to prove cycle termination of any SRS $R$ is as
follows: choose an instance of a semi-ring and interpretations $\li a\ri \in M$ for every $a \in
\Sigma$, and for all rules $\ell \to r$ we either have $\li \ell \ri  > \li r \ri$
or $\li \ell \ri  \geq \li r \ri $. If $\li \ell \ri  > \li r \ri$ holds for all rules we are done, 
otherwise the remaining proof obligation is to prove cycle termination of the rules for
which $\li \ell \ri  \geq \li r \ri $.

Although the proof of Theorem \ref{thmmatr} is not very hard, it is quite subtle: while in the 
final argument $\tr(\li \cdot \ri)$ is applied on the strings, it is essential to require  
$\li \ell \ri  > \li r \ri $ and not the
weaker requirement  $\tr(\li \ell) \ri  > \tr(\li r \ri)$. Surprisingly, the proof does not
need further requirements on the relation $>$ and $\geq$ like transitivity or ${>} \subseteq
{\geq}$.

In order to prove relative termination, we give the following extension of Theorem \ref{thmmatr}.

\begin{theorem}
\label{thmmatrrel}
Let $S' \subseteq S$, $R' \subseteq R$, $S' \subseteq R'$, and $S \subseteq R$ be SRSs over $\Sigma$ and 
let $\li \cdot \ri  : \Sigma \to M$ satisfy the above properties and 
\begin{itemize}
\item $\rto{S'}$ is terminating relative to $\rto{R'}$, and
\item $\li \ell \ri  \geq \li r \ri $ for all $(\ell \to r) \in ((R' \setminus S) \cup S')$
\item $\li \ell \ri  > \li r \ri $ for all $(\ell \to r) \in (R \setminus ((R' \setminus S) \cup S'))$
\end{itemize}
Then $\rto{S}$ is terminating relative to $\rto{R}$.
\end{theorem}
\begin{proof}
Assume $[u_1] \rto{R} [u_2] \rto{R} [u_3] \rto{R} \cdots$; we have to prove that it contains 
finitely many $\rto{S}$-steps. As in the proof of Theorem \ref{thmmatr} we have $\tr(\li u_i \ri) >
\tr(\li u_{i+1} \ri)$ for applications of rules $(\ell \to r) \in (R \setminus ((R' \setminus S) \cup S'))$
and $\tr(\li u_i \ri) \geq \tr(\li u_{i+1} \ri)$ for  applications of remaining rules.
Since $>$ is well-founded on $X$, there are only finitely many steps of the former type, so after a finite
initial part the infinite reduction only consists of $((R' \setminus S) \cup S')$-steps. But then applying that 
$\rto{S'}$ is terminating relative to $\rto{R'}$, these remaining $((R' \setminus S) \cup S')$-steps only contain finitely
many $S'$-steps. All other steps are $(R' \setminus S)$-steps and thus do not contain $S$-steps.
Hence in total there are only finitely many $\rto{S}$-steps.
\end{proof}

Indeed Theorem \ref{thmmatr} is an instance of Theorem \ref{thmmatrrel} if $S' = \emptyset$ and $R' = R \setminus S$.
A typical way to apply Theorem \ref{thmmatrrel} to prove cycle termination of any SRS $S$ relative to $R$
is very similar to the typical way to apply Theorem \ref{thmmatr}: remove the rules for which '$>$' is
obtained, both from $S$ and $R$.

In order to apply Theorem \ref{thmmatr} or Theorem \ref{thmmatrrel} we need an instance of a semi-ring $X$, a set $M$ of matrices
over $X$ and relations $>$, $\geq$ on $M$ such that all assumed properties holds. We give three such
instances: the tropical, arctic and natural matrix interpretations, all depending on dimension
$d$. 
For all these three instances the search for applying Theorem \ref{thmmatr} has been implemented in our 
tool {\tt torpacyc} by transforming the requirements to SMT format and calling the external SMT 
solver {\tt Yices} \cite{D14,yices}.  The same has been done in our tool {\tt tdmi}, also covering
Theorem \ref{thmmatrrel}.

A nice point is that we do not need to try separately for
which rules we require $\li \ell \ri  \geq \li r \ri $ and for which rules $\li \ell \ri  >
\li r \ri $, but we just specify in the SMT formula that the latter occurs at least once and the
former for all other rules. 

\subsection{Natural  Matrix Interpretations}

In natural  matrix interpretations we have $X = \nat$, and $0,1,+,\times, >$ have their usual meaning.
We define $M$ to consist of all $d\times d$ matrices $A$ satisfying $A_{11} > 0$. 
On $M$ we define the relations $>$ and $\geq$ by
\[ A > B \desda  A_{11} > B_{11} \wedge \forall i,j : A_{ij} \geq B_{ij},
\quad\qquad A \geq B \desda   \forall i,j : A_{ij} \geq B_{ij}.\]
For $M$ with these relations the required properties 
\eqref{assume1}, \eqref{assume2}, \eqref{assume3}, and
\eqref{assume4}
are all easily checked. Hence this yields a way to prove (relative) cycle termination by Theorem
\ref{thmmatr}.
As an example we consider the single rule $a\?a \to a\?b\?a$ and choose
\[ 
\arraycolsep=1.4pt 
\li a \ri = \left( \begin{array}{cc} 1 & 1 \\ 1 & 0 \end{array} \right), \;
 \li b \ri = \left( \begin{array}{cc} 1 & 0 \\ 0 & 0 \end{array} \right), \]
yielding
\[
\arraycolsep=1.4pt 
\li a a \ri = \left( \begin{array}{cc} 1 & 1 \\ 1 & 0 \end{array} \right)
\left( \begin{array}{cc} 1 & 1 \\ 1 & 0 \end{array} \right) =
\left( \begin{array}{cc} 2 & 1 \\ 1 & 1 \end{array} \right) >
\left( \begin{array}{cc} 1 & 1 \\ 1 & 1 \end{array} \right) =
\left( \begin{array}{cc} 1 & 1 \\ 1 & 0 \end{array} \right) 
\left( \begin{array}{cc} 1 & 0 \\ 0 & 0 \end{array} \right) 
\left( \begin{array}{cc} 1 & 1 \\ 1 & 0 \end{array} \right) =
\li a b a \ri,\]
proving cycle termination by Theorem \ref{thmmatr}.

The original versions of matrix interpretations in \cite{HW06,EWZ08} 
are not suitable
for proving cycle 
termination since they succeed in proving termination of $a\?b \to b\?a$ for which cycle
termination does not hold. Even more, the same holds for the original killer example for the method 
of matrix interpretations was $a\?a \to b\?c,\; b\?b \to a\?c, \; c\?c \to a\?b$, for which cycle
termination does not hold due to  
$[c\?c\?a\?a] \cycto [a\?b\?a\?a] \cycto [a\?b\?b\?c] \cycto [a\?a\?c\?c]$. 

Main differences are our conditions \eqref{assume3} and \eqref{assume4} on the trace of the matrices
and that in our setting the interpretation of symbols is multiplication by a matrix, 
while in \cite{HW06,EWZ08} it combines such a matrix multiplication by adding a vector.

\subsection{Tropical Matrix Interpretations}

In tropical matrix interpretations we choose the semi-ring $X = \nat \cup \{\infty\}$, with
$\min$ being the semi-ring addition and
the normal addition being the semi-ring multiplication, both extended to $X$ by defining
\[ \min(\infty,x) = \min(x,\infty) = x \; \mbox{ and } \; \infty + x = x + \infty = \infty
\]for all $x \in \nat \cup \{ \infty\}$. Now $\infty$ acts as the semi-ring zero and 0 acts as
the semi-ring unit; it is easily checked that all semi-ring requirements hold. This semi-ring 
is called the {\em tropical semi-ring} after its study by the Brazilian mathematician Imre 
Simon \cite{S88}. Over this semi-ring multiplication of matrices 
becomes
\[
(AB)_{i,j} = \min ( \{ A_{i,k} + B_{k,j} \mid k =
1,\ldots,d \}),\]
so being quite different from the usual matrix operations.
On this semi-ring we define the well-founded order $>$ to be the extension on $>$ on
$\nat$ defined by 
\[ x > y \desda (x,y \in \nat \wedge x > y) \vee (x = \infty \wedge y \in \nat).\] 
So in this semi-ring the zero element is not the smallest but the largest element.

We define $M$ to consist of all $d\times d$ matrices $A$ satisfying $A_{11} \neq \infty$. 
On $M$ we define the relation $>$ by
\[ A > B \desda  \forall i,j : (A_{ij} > B_{ij} \vee A_{ij} = B_{ij} = \infty),\]
and the relation $\geq$ by
\[ A \geq B \desda  \forall i,j : (A_{ij} > B_{ij} \vee A_{ij} = B_{ij}).\]
For $M$ with these relations the required properties 
\eqref{assume1}, \eqref{assume2}, \eqref{assume3}, and
\eqref{assume4}
are all easily checked; note that for every $A \in M$ we have $\tr(A) = \min_i A_{ii} \neq
\infty$ since $A_{11} \neq \infty$. For these requirements it is essential that we defined 
$A > B$ by $A_{ij} > B_{ij}$ on all positions and not only at 1,1, since from $a > b$
we cannot conclude $\min(a,c) > \min(b,c)$, but from $a > b \wedge c > d$  
we can conclude $\min(a,c) > \min(b,d)$.

As an example we again consider the single rule $a\?a \to a\?b\?a$ and choose
\[ 
\arraycolsep=1.4pt 
\li a \ri = \left( \begin{array}{cc} 1 & \infty \\ 0 & 1 \end{array} \right), \;
 \li b \ri = \left( \begin{array}{cc} 0 & 0 \\ 1 & 1 \end{array} \right), \]
by using the tropical matrix multiplication
yielding
\[
\arraycolsep=1.4pt 
\li a a \ri = 
\left( \begin{array}{cc} 1 & \infty \\ 0 & 1 \end{array} \right) 
\left( \begin{array}{cc} 1 & \infty \\ 0 & 1 \end{array} \right) =
\left( \begin{array}{cc} 2 & \infty \\ 1 & 2 \end{array} \right) >
\left( \begin{array}{cc} 1 & 2 \\ 0 & 1 \end{array} \right) =
\left( \begin{array}{cc} 1 & \infty \\ 0 & 1 \end{array} \right) 
\left( \begin{array}{cc} 0 & 0 \\ 1 & 1 \end{array} \right)
\left( \begin{array}{cc} 1 & \infty \\ 0 & 1 \end{array} \right) =
\li a b a \ri,\]
proving cycle termination by Theorem \ref{thmmatr}.

\subsection{Arctic Matrix Interpretations}

In arctic matrix interpretations we choose the semi-ring $X = \nat \cup \{-\infty\}$, with
$\max$ being the semi-ring addition and
the normal addition being the semi-ring multiplication, both extended to $X$ by defining
\[ \max(-\infty,x) = \max(x,-\infty) = x \; \mbox{ and } \; -\infty + x = x + -\infty = -\infty
\]
for all $x \in \nat \cup \{- \infty\}$. Now $-\infty$ acts as the semi-ring zero and 0 acts as
the semi-ring unit; it is easily checked that all semi-ring requirements hold. This semi-ring 
is called the {\em arctic semi-ring} as it is a kind of opposite to the tropical semi-ring.
Over this semi-ring 
multiplication of matrices becomes
\[
(AB)_{i,j} = \max ( \{ A_{i,k} + B_{k,j} \mid k =
1,\ldots,d \}),\]
On this semi-ring we define the well-founded order $>$ to be the extension on $>$ on
$\nat$:
\[ x > y \desda (x,y \in \nat \wedge x > y) \vee (x \in \nat \wedge y = -\infty).\] 
So in this semi-ring the zero element is the smallest element, just like in the natural semi-ring.

We define $M$ to consist of all $d\times d$ matrices $A$ satisfying $A_{11} \neq -\infty$. 
On $M$ we define the relation $>$ by
\[ A > B \desda  \forall i,j : (A_{ij} > B_{ij} \vee A_{ij} = B_{ij} = -\infty),\]
and the relation $\geq$ by
\[ A \geq B \desda  \forall i,j : (A_{ij} > B_{ij} \vee A_{ij} = B_{ij}).\]
For $M$ with these relations all required properties hold, again providing a method for proving
(relative) cycle termination by Theorem \ref{thmmatr}.
Cycle termination of our example $a\?a \to a\?b\?a$ can also be proved by arctic matrix
interpretation, now by choosing
\[ 
\arraycolsep=1.4pt 
\li a \ri = \left( \begin{array}{cc} 0 & 1 \\ 0 & 1 \end{array} \right), \;
 \li b \ri = \left( \begin{array}{cc} 0 & -\infty \\ -\infty & -\infty \end{array} \right). \]

\subsection{Examples and Bounds on Reduction Lengths}
As an example of combining various versions of matrix interpretations we consider the system 
from the introduction: 
\[0\?P \to 1\?P,\; 1\?P \to c\?P,\; 0\?c \to 1\?0, \; 1\?c \to c\?0, \; P\?0 \to P\?1\?0\?0\]
for which the following proof is found fully automatically by {\tt torpacyc}.

First the following tropical matrix interpretation is found:
\[ 
\arraycolsep=1.4pt 
\li P \ri = \left( \begin{array}{cc} 0 & \infty \\ 0 & \infty \end{array} \right), \;
 \li 0 \ri = \left( \begin{array}{cc} 2 & 2 \\ 0 & 0 \end{array} \right), \;
 \li 1 \ri = \left( \begin{array}{cc} 2 & 1 \\ \infty & 0 \end{array} \right), \;
 \li c \ri = \left( \begin{array}{cc} 1 & 0 \\ \infty & 0 \end{array}
 \right).\]

\noindent In this interpretation for the first four rules we have $\li \ell \ri \geq \li r \ri$ and 
for the last rule we have $\li \ell \ri > \li r \ri$. So by Theorem \ref{thmmatr}
the last rule may be removed, and {\tt torpacyc} continues with the remaining four rules.

Next  the following natural matrix interpretation is found: 
\[ 
\arraycolsep=1.4pt 
\li P \ri = \left( \begin{array}{cc} 1 & 0 \\ 2 & 0 \end{array} \right), \;
 \li 0 \ri = \left( \begin{array}{cc} 1 & 2 \\ 0 & 2 \end{array} \right), \;
 \li 1 \ri = \left( \begin{array}{cc} 1 & 1 \\ 0 & 2 \end{array} \right), \;
 \li c \ri = \left( \begin{array}{cc} 1 & 0 \\ 0 & 2 \end{array} \right).\]
For these interpretations we obtain $\li 0\?P \ri > \li 1\?P \ri$,
$\li 1\?P \ri > \li c\?P \ri$, $\li 0\?c \ri \geq \li 1\?0 \ri$ and $\li 1\?c \ri \geq \li c\?0 \ri$, hence
by Theorem \ref{thmmatr} it suffices to prove cycle termination of $0\?c \to 1\?0, \; 1\?c \to c\?0$, for which
{\tt torpacyc} finds a simple counting argument. These simple counting arguments can be seen as
the instance of tropical matrix interpretations of dimension $d=1$, not using $\infty$: 
interpret every $a \in
\Sigma$ by a natural number, and then $\li u \ri$ is the sum of the interpretations of
all symbols in $u$.

The method for proving cycle termination induced by Theorem~\ref{thmmatr} has similar limitations as 
the method of matrix interpretations in \cite{HW06-rta} for string termination: Since the entries 
of a product of $n$ matrices are bounded by an exponential function in $n$, the method cannot prove
cycle termination of systems which allow reduction sequences where every rewrite rule is applied
more often than exponentially often.

\begin{example}\label{expl:tower}
The rewrite system $R_1 := \{a\?b \to b\?c\?a, c\?b \to b\?b\?c\}$  allows for string derivations of a length 
which is a tower of exponentials (see \cite{HW06-rta}),
{i.e.}~the string $a^k\?b^k$ has such a long derivation, since 
the derivation $a\?b^n \to_{R_1}^*  b^{2^n-1}\?c^n\?a$ exists and this can be iterated for every $a$ in $a^k$.
Moreover, the number of applications of the first and of the second rule of $R_1$ is a tower of exponentials.
This shows that the matrix interpretations in \cite{HW06-rta} are unable to prove string termination of $R_1$.
The system $\phi(R_1) := \{
R\?E \to L\?E,
 a\?L \to L\?a', 
 b\?L \to L\?b',
 c\?L \to L\?c',
 R\?a' \to a\?R,
 R\?b' \to b\?R,
 R\?c' \to c\?R,
 a\?b\?L \to b\?c\?a\?R,
 c\?b\?L \to b\?b\?c\?R\}$ uses the transformation $\phi$ from \cite{ZKB14} and transforms the string rewrite system $R_1$ into a cycle rewrite system {s.t.}~$R_1$ is string terminating iff $\phi(R_1)$ is cycle terminating.
One can verify that $[a\?b^n\?L\?E] \cycto_{\phi(R_1)}^* [b^{2^n-1}\?c^n\?a\?L\?E]$ which can also be iterated
{s.t.}~$[a^k\?b^k\?L\?E]$ has a cycle rewriting sequence whose length is a tower of $k$ exponentials. 
Inspecting all nine rules of $\phi(R_1)$, the number of applications of any of the rules in this rewrite sequence
is also a tower  of $k$ exponentials and thus it is impossible to prove cycle termination of $\phi(R_1)$ using
Theorem~\ref{thmmatr}. Consequently, our tool \texttt{torpacyc} does not find a termination proof for $\phi(R_1)$.
\end{example}

\begin{remark}
As expected our tool \texttt{torpacyc} does not find a termination proof for $\phi(R_1)$ from Example~\ref{expl:tower}.
On the other hand, with our transformational approach cycle termination can be proved:
\aprove~proves string termination of $\sym{split}(\phi(R_1))$. 
\end{remark}

A further question is whether matrix interpretations are limited to 
cycle rewrite systems with exponential derivation lengths only.
The following example shows that this is not true:

\begin{example}\label{expl:doubly-exp}
The SRS $R_2 := \{a\?b \to b\?a\?a, c\? b \to b\?b\?c\}$ (see \cite{HW06-rta}) has derivations of doubly exponential
length (since $a\?c^k\?b \to_{R_2}^* b^{2^k}\?a^{2^{2^k}}\?c^k$ and any rewrite step adds one symbol),
but its string termination can be proved by relative termination and matrix interpretations by first
removing the rule $c\? b \to b\?b\?c$ and then removing the other rule. This is possible, since the second
rule is  applied only exponentially often. For cycle rewriting the encoding $\phi$ from \cite{ZKB14}
is $\phi(R_2) = \{
R\?E\to L\?E, 
a\?L \to  L\?a',
b\?L \to L\? b', 
c\?L  \to L\?c',
R\?a'  \to a\?R,
R\?b' \to b\?R,
R\?c'  \to c\?R,
a\?b\?L  \to b\?a\?a\?R,
c\?b\?L  \to b\?b\?c\?R\}$ and $\phi(R_2)$ is cycle terminating iff $R_2$ is string terminating.
The system $\phi(R_2)$ also has doubly exponential cycle derivations, {e.g.}~$[a\?c^k\?b\?L\?E] \cycto_{\phi(R_2)}^* [b^{2^k}\?a^{2^{2^k}}\?c^k\?L\?E]$. However, \texttt{torpacyc} proves
cycle termination of $\phi(R_2)$ by first removing the last rule using the matrix interpretation
\[ 
\arraycolsep=1.4pt 
\begin{array}{@{}c@{}}
\li R \ri = \left( \begin{array}{cc}  1 & 2 \\  1 & 0  \end{array}\right),
\li E \ri = \left( \begin{array}{cc}  2 & 0 \\  0 & 0  \end{array}\right),
\li L \ri = \left( \begin{array}{cc}  1 & 2 \\  1 & 0  \end{array}\right),
\li a \ri = \left( \begin{array}{cc}  1 & 0  \\ 0 & 1  \end{array}\right),
\li a' \ri = \left( \begin{array}{cc}  1& 0 \\ 0& 1  \end{array}\right),
\\
\li b \ri = \left( \begin{array}{cc}  1 &2 \\ 0 &1  \end{array}\right),
\li b' \ri = \left( \begin{array}{cc}  1 &0 \\ 1 &1  \end{array}\right),
\li c \ri = \left( \begin{array}{cc}  3 &0 \\ 0 &1  \end{array}\right),
\li c' \ri = \left( \begin{array}{cc}  1 &0 \\ 1 &3  \end{array}\right).
\end{array} \]
Thereafter the remaining rules (which now only have derivations of exponential length)
are eliminated by matrix interpretations and counting arguments.

Note that the methods in \cite{ZKB14} are not able to prove cycle termination of
$\phi(R_2)$, since they can only remove rules which are applied polynomially often in any
derivation. 

Also the transformational approach successfully proves cycle termination of  $\phi(R_2)$:
\TTTT{} proves string termination of $\sym{split}(\phi(R_2))$. 
Interestingly, we did not find a termination proof of $\sym{split}(\phi(R_2))$ using \aprove.
\end{example}

\section{Tools and Experimental Results}
\label{secexp}
In this section we first explain which tools and which benchmark sets 
were used to evaluate the techniques presented in this paper. Thereafter we summarize and analyze
the obtained results.
\subsection{Tools for Proving Cycle Termination}
We implemented several tools for proving cycle termination and cycle non-termination, 
which will be explained in this section. 

The command line tool {\tt cycsrs} is mainly a wrapper which allows to call different
other tools and termination provers. It also allows to call the different tools
in succession on a given termination problem and to distribute the time limit among the tools.
The tool participated in the category on cycle rewriting in the Termination Competition 2015 \cite{termcomp}
and in the Termination and Complexity Competition 2016 \cite{termcomp16}.

For the transformational approaches, {\tt cycsrs} is able to perform the transformations
$\sym{split}$, $\sym{rotate}$, or $\sym{shift}$ for a given input problem and then it calls
the termination provers $\aprove$ or \TTTT~to prove string termination of the transformed problem.
Analogously, {\tt cycsrs} can apply  transformations $\sym{split}_{\mathit{rel}}$, 
$\sym{rotate}_{\mathit{rel}}$, and $\sym{shift}_{\mathit{rel}}$, to enable the proof of relative cycle termination
by showing relative string termination.

For the search for matrix interpretations described in Section~\ref{sec:matrix-interpretations},
 we implemented two tools:
The prover \texttt{torpacyc} searches for matrix interpretations by 
using the SMT-solver {\tt yices} targeting the logic of 
unquantified linear integer arithmetic with uninterpreted sort and function symbols. 

The tool \texttt{tdmi} uses a similar approach, applying  {\tt yices} to find matrix interpretations, but targets the logic of quantifier-free formulas over
the theory of fixed-size bit vectors (similarly as proposed in \cite{CFFGW:12}).  
Moreover, \texttt{tdmi} is able to prove relative termination.
Another difference is that {\tt torpacyc} checks for match bound proofs, see \cite{ZKB14}. Match bound proofs can be seen as proofs by tropical matrix interpretation, but with dimensions of often more than 100, being far beyond the dimension that is feasible for direct search for tropical matrix interpretations, typically being 3.

For proving non-termination, one technique is to prove string non-termination 
of the cycle rewrite problem by using a termination prover like $\aprove$ or 
\TTTT~(which is correct due to Proposition~\ref{prop:properties-rta14}).
Another technique is to apply one of the transformations and then proving string non-termination 
of the transformed problem (which is correct since the transformations are complete).
Additionally, we implemented a tool {\tt cycnt} which performs a brute-force search for 
cycle non-termination.

The automatic transformation and the prover can also be used online via a web interface
available via \url{http://www.ki.informatik.uni-frankfurt.de/research/cycsrs/} where also
the tools and experimental data can be found.

\subsection{Benchmark Sets and Test Environment}
A problem for doing experiments is that no real appropriate test set for cycle rewriting is available. 
We played around with several examples like the ones in this paper, but we also 
wanted to test larger scale experiments. For proving cycle termination, we use 
the  SRS\_Standard-branch of the Termination Problem Data Base \cite{tpdb} (TPDB), which is a benchmark set 
for proving termination of string rewrite systems, but it was also used as problem set in the
cycle rewriting category of the Termination Competition 2015 and slightly extended by some new problems in 2016.
This set contains 1315 termination problems.

For relative cycle termination, we use the SRS\_Relative-branch of the Termination Problem Data Base,
which is a benchmark set containing 205 relative string termination problems.

In \cite{sabel-zantema:15} we also tested some of the techniques on 50000 randomly generated string rewrite systems.
We excluded them in the new tests, since most of the problems where either trivially cycle terminating
or could easily be proved to be cycle non-terminating.
A further difference to the benchmarks presented in \cite{sabel-zantema:15} is the test environment: 
the current benchmarks where all performed on StarExec \cite{starexec} --
a cross community logic solving service --
which made it possible to rapidly run the tests
and thus also to perform tests with higher time-outs.%\newpage

\subsection{Experimental Results on Cycle Termination}

\begin{table}
\begin{center}
\begin{tabular}{ll||c|c|c||c|c|c}
\multirow{2}{1in}{SRS\_Standard}            &&\multicolumn{3}{c||}{Time limit 60 secs.} &\multicolumn{3}{c}{Time limit 300 secs.}
\\
                               &       &YES     &NO    &MAYBE  &YES    &NO     &MAYBE
\\
\multicolumn{8}{l}{\cellcolor{black!20!white}\rule{0ex}{3ex}\bf Transformational Approach}
\\[.5ex]
\multirow{3}{1in}{$\sym{split}$}&\aprove&26     &316    &973    &41     &324    &950
\\
                                &\TTTT  &27     &164    &1124   &46     &180    &1089 
\\
                                &any    &35     &316    &964    &53     &324    &938
\\
 \hline
\multirow{3}{1in}{$\sym{rotate}$}&\aprove&10    &48     &1257   &10     &53     &1252
\\
                                 &\TTTT  &3     &0      &1312   &9      &0      &1306
\\
                                 &any    &10    &48     &1257   &11     &53     &1251
\\
 \hline
\multirow{3}{1in}{$\sym{shift}$} &\aprove&10    &79     &1226   &10     &89     &1216
\\
                                 &\TTTT  &8     &0      &1307   &8      &1      &1306
\\
                                 &any    &10    &79     &1226   &10     &89     &1216
\\
\multicolumn{8}{l}{\cellcolor{black!20!white}\rule{0ex}{3ex}\bf Trace-Decreasing Matrix Interpretations}
\\[.5ex]
\multicolumn{2}{l||}{\texttt{torpacyc}}          &44     &0      &1271   &49     &0      &  1266 
\\
\multicolumn{2}{l||}{\texttt{tdmi}}              &31     &0      &1284   &46     &0      &  1269
\\
\multicolumn{2}{l||}{any}                        &48     &0      &1267   &63     &0      &  1252
\\

\multicolumn{8}{l}{\cellcolor{black!20!white}\rule{0ex}{3ex}\bf Non-Termination Check}
\\[.5ex]
% \hline
\multicolumn{2}{l||}{\texttt{cycnt}}             &0      &580    &735    &0      &588    &727
\\            
\multicolumn{2}{l||}{\aprove~(SRS)}              & 0      & 99      & 1216    &0      &109     &1206
\\            
\multicolumn{2}{l||}{\TTTT~(SRS)}                & 0      & 33      & 1282    &0      &62     &1253
% \hline
\\
\multicolumn{2}{l||}{any}                        &0     &583      &732   &0     &591      &  724
\\[.5ex]
\multicolumn{8}{l}{\cellcolor{black!20!white}\rule{0ex}{3ex}\bf Combination of Techniques}
\\[.5ex]
% \hline
\multicolumn{2}{l||}{back-end: \aprove} &51     &511    &753    &62     &540    &713
\\ 
\multicolumn{2}{l||}{back-end: \TTTT}   &50     &511    &754    &68     &539    &708
\\
% \hline\hline
\rowcolor{black!90!white}\multicolumn{2}{l||}{\color{white}\rule{0ex}{3ex}\bf any}               
                                        &\color{white}\rule{0ex}{3ex}\bf 60
                                                &\color{white}\rule{0ex}{3ex}\bf 583
                                                      &\color{white}\rule{0ex}{3ex}\bf 682    
                                                                &\color{white}\rule{0ex}{3ex}\bf 83     
                                                                        &\color{white}\rule{0ex}{3ex}\bf 591
                                                                                &\color{white}\rule{0ex}{3ex}\bf 651
\\[.5ex]
\end{tabular}
\caption{Experimental results for proving cycle (non)-termination on the 1315 problems of the SRS\_Standard branch of the TPDB\label{table:results-cycletermination}}
\end{center}
\end{table}
Table~\ref{table:results-cycletermination} summarizes the obtained results for running several
techniques to prove cycle termination on the SRS\_Standard-branch of the TPDB.
The first column lists the different techniques, thereafter there are two main columns: the first one
lists the results where a time limit of 60 secs.\ was used, and the second main column lists the
result for a time limit of 300 secs. 
There are three kinds of results: YES means that cycle termination has been proved,
NO means that cycle non-termination has been proved (and thus cycle termination was disproved),
and MAYBE means that no result was obtained.

The rows of Table~\ref{table:results-cycletermination} consist of five main parts.
\begin{itemize}
\item The first part consists of the results for the transformational approach, where each of
the three transformations $\sym{split}$, $\sym{rotate}$, and $\sym{shift}$ was applied to
the input problem and thereafter either $\aprove$ or \TTTT~was applied to the transformed
system. We also list the numbers for combining the results of the two solvers 
(in the rows named ``any''), which sometimes shows that the termination techniques perform
differently on the same problems. 

\item The second part contains the summarized results of applying the two tools \texttt{torpacyc} and \texttt{tdmi}
to the problems. Both tools try to prove cycle termination by searching for matrix interpretations.
Note that both tools are unable to show cycle non-termination. Again the row named with ``any'' combines
the results of both techniques.

\item The third part shows the result of applying the brute-force search for cycle non-termination, using our tool \texttt{cycnt},
and the results of applying $\aprove$ and \TTTT~to prove string non-termination of the input problem (which implies cycle termination).
We also list the numbers for combining the three techniques to prove cycle non-termination (listed in the row named ``any'').

\item The fourth part consists of the results for a combination of the termination techniques. 
Here we use $\aprove$ or \TTTT, respectively, as back-end prover to show string termination and non-termination.
In detail, the command line tool {\tt cycsrs} is used to first run {\tt torpacyc} (with 25\%  of the time limit), 
then {\tt tdmi} (14 \%), 
then the back-end prover (9 \%) and {\tt cycnt} (10 \%) to prove cycle non-termination, and finally 
to apply the transformation $\sym{split}$ together with a subsequent call
of the back-end prover to show (non-)termination of the transformed system (42 \%).

\item The last row of the table combines all results of the previous rows (where the YES- and NO-results
are summed up per problem).
\end{itemize}

\noindent An overall observation is that our techniques were able to obtain a 
result for about the half of the problems, while the other half of the problems
seem to be too hard to be proved by the techniques.
This is not really surprising since the test set contains already `hard' instances
for proving {\em string termination}: In the Termination Competition 2015 the winning 
prover $\aprove$ solved 832 out of the 1325 string termination problems (with a time limit of 300 secs.).

Considering cycle termination, we were now able to solve 643 problems
with a time limit of 60 secs.\ (by combining all of our techniques), while in \cite{sabel-zantema:15}
only 399 problems were solved in the same time limit and on the same problem set 
(but on a different environment, not on StarExec).
This increase in the number of solved problems mainly comes from adding the tool {\tt cycnt} to search for
cycle non-termination.
The results for checking non-termination also show that in the benchmark set many problems seem to cycle non-terminating,
while they are string terminating.
We expected this, since a substantial part of the problems may contain a renaming of the rule $ab \to ba$. 

A further observation is that increasing the time limit allows to solve more problems,
where the increase on proving cycle non-termination is rather small, while for cycle termination
the increase is noticeable.

Comparing the three transformations, the transformation $\sym{split}$ leads to much better
results than the other two transformations, which holds for termination and for non-termination proofs.

For proving cycle termination, problem specific methods (i.e. \texttt{torpacyc} and \texttt{tdmi}) seem to 
perform a little bit better than the transformational method using $\sym{split}$. 

Comparing the back-end prover, there is surprisingly no 
clear winner: $\aprove$ seems to perform better for short time limits and for proving non-termination, while \TTTT~seems to perform better
for longer time limits.%\newpage

\subsection{Relative Cycle Termination}

\begin{table}
%% DONE!

\begin{center}
\begin{tabular}{ll||c|c|c||c|c|c}
\multirow{2}{1in}{SRS\_Relative}            &&\multicolumn{3}{c||}{Time limit 60 secs.} &\multicolumn{3}{c}{Time limit 300 secs.}
\\
                                            &       &YES     &NO    &MAYBE  &YES    &NO     &MAYBE
\\
\multicolumn{8}{l}{\cellcolor{black!20!white}\rule{0ex}{3ex}\bf Transformational Approach}
\\[.5ex]
% % \multicolumn{8}{l}{\cellcolor{black!20!white}\bf~~~}
% \\
% \hline
\multirow{3}{1in}{$\sym{split}$}&\aprove& 8      & 7      & 190    &14     &7      &184
\\
                                &\TTTT  & 1      & 8      & 196    &0      &8      &197
\\                                
                                &any    & 8      & 8      & 189    &14     &8      &183
                                
\\
 \hline

\multirow{3}{1in}{$\sym{rotate}$}&\aprove& 1     & 0      & 204    &1      &0      &204
\\
                                 &\TTTT  & 0     & 0      & 205    &0      &0      &205
\\
                                 &any    & 1     & 0      & 204    &1      &0      &204
\\
 \hline
\multirow{3}{1in}{$\sym{shift}$} &\aprove& 1     & 0      & 204    &2      &0      &203
\\
                                 &\TTTT  & 0     & 0      & 205    &0      &1      &204
\\
                                 &any    & 1     & 0      & 204    &2      &1      &202
\\
%\hline
\multicolumn{8}{l}{\cellcolor{black!20!white}\rule{0ex}{3ex}\bf Trace-Decreasing Matrix Interpretations}
\\[.5ex]
\multicolumn{2}{l||}{\texttt{tdmi}}              & 11         & 0      & 194    &21     &0      &184
\\
% \hline
\multicolumn{8}{l}{\cellcolor{black!20!white}\rule{0ex}{3ex}\bf Non-Termination Check}
\\[.5ex]
% \hline
\multicolumn{2}{l||}{\texttt{cycnt}}             & 0      & 13     & 192    &0      &13     &192
\\            
\multicolumn{2}{l||}{\aprove~(SRS)}              & 0      & 1      & 204    &0      &1     &204
\\            
\multicolumn{2}{l||}{\TTTT~(SRS)}                & 0      & 1      & 204    &0      &1     &204
\\            
\multicolumn{2}{l||}{any}                        & 0     & 13      & 192   &0     &13     &  192
\\
% \hline
\multicolumn{8}{l}{\cellcolor{black!20!white}\rule{0ex}{3ex}\bf Combination of Techniques}
\\[.5ex]
% \hline
\multicolumn{2}{l||}{back-end: \aprove}  & 10     & 13     & 182    &17     &13     &175
\\ 
\multicolumn{2}{l||}{back-end: \TTTT}   & 8      & 13     & 184    &16     &13     &176
\\
% \hline\hline
\rowcolor{black!90!white}\multicolumn{2}{l||}{\color{white}\rule{0ex}{3ex}\bf any}               
                                        &\color{white}\rule{0ex}{3ex}\bf  12
                                                &\color{white}\rule{0ex}{3ex}\bf  13
                                                      &\color{white}\rule{0ex}{3ex}\bf  180
                                                                &\color{white}\rule{0ex}{3ex}\bf  21
                                                                        &\color{white}\rule{0ex}{3ex}\bf  13
                                                                                &\color{white}\rule{0ex}{3ex}\bf  171
\\[.5ex]
\end{tabular}
\caption{Experimental results for proving relative cycle (non)-termination on the 205 problems of the SRS\_Relative branch of the TPDB\label{table:results-relative}}
\end{center}
\end{table}
Table~\ref{table:results-relative} summarizes the results on applying different techniques
to the 205 problems in the SRS\_Relative-branch of the TPDB.
Again, the tests were run with a time limit of 60 secs.\ and with a time limit of 300 secs.
Since {\tt torpacyc} is not able to prove relative termination, there
are no tests using {\tt torpacyc}. 

One observation of the results is that again the transformation $\mathit{split}_{\mathit{rel}}$ 
leads to better results than the other transformations. A further observation is that
the problem specific methods (the matrix interpretations to prove relative cycle termination
and {\tt cycnt} to disprove cycle termination) perform slightly better than the transformational approach.

Finally, the number of solved problems is quite small compared to the number of problems.
One reason may be that the benchmark set contains only few small rewrite systems, and 
many large problems. Large problems are disadvantageous for the transformational approach, since
the transformations (especially the transformation $\sym{split}$) increases the size of the problem.

\subsection{Proving Cycle Termination by Relative String Termination}
To prove cycle termination of a string rewrite system $S$, 
we can use all the techniques for proving relative cycle termination 
by using all rules of $S$ as strict rules (and thus there are no weak rules).
This usually does not lead to an improvement by applying automated termination tools, 
since the problem of showing $S$ being terminating relative to $S$ is equal to showing 
that $S$ is terminating.
However, for the transformational approaches the setting is different. Let $\psi$ be one of the sound and complete transformations 
and $\psi_{\mathit{rel}}$ its variant for relative cycle termination.
Since for an SRS $S$, the transformation $\psi_{\mathit{rel}}(S,S)$ results in  $(S',R')$
where this inclusion is strict, {i.e.} $S' \subset R'$, it makes sense to analyse whether transforming 
$S$ into $(S',R')$ 
and subsequently proving that $S'$ is string terminating relative to $R'$
enables further (non-)termination proofs compared to trying to prove string termination of $\psi(S)$.

For each problem in SRS\_Standard-branch of the TPDB we applied 
each combination of a transformation and the two termination provers
$\aprove$ and \TTTT.
Table~\ref{table-rel-nonrel} summarizes the results of this analysis,
where we used a time limit of 300 secs.
The numbers of proved, disproved, and open problems are listed: 
first for the usual approach to prove string termination of $\psi(S)$, 
secondly for proving relative string termination of $\psi_{\mathit{rel}}(S,S)$, 
and as a third row the combined results  are shown (in the rows labeled with ``any'').

\begin{table}
\begin{center}
\begin{tabular}{l||c|c|c||c|c|c}
SRS\_Standard &\multicolumn{3}{c||}{$\aprove$}  &\multicolumn{3}{c}{\TTTT}
\\
Time limit 300 secs. &YES&   NO& MAYBE  &  YES&   NO& MAYBE  
\\
\hline\hline
{$\sym{split}$}                     &41&324&950&46&180&1089
\\
{$\sym{split}_{\mathit{rel}}$}      &26&179&1110&8&213&1094
\\
any                                 &41&324&950&46&216&1053
\\
\hline
{$\sym{rotate}$}                     &10&53&1252&9&0&1306
\\
{$\sym{rotate}_{\mathit{rel}}$}      &10&1&1304&2&1&1312
\\
any                                  &10&53&1252&9&1&1305
\\
\hline
{$\sym{shift}$}                      &10&89&1216&8&1&1306
\\
{$\sym{shift}_{\mathit{rel}}$}      &10&6&1299&4&13&1298
\\
any                                 &10&89&1216&8&13&1294
\end{tabular}

\caption{Results for proving cycle termination of the 1315 problems in the SRS\_Standard branch of the TPDB by transformation into string termination problems and relative string termination problems \label{table-rel-nonrel}}

\end{center}
\end{table}

The results show that indeed also the technique using relative string termination as target 
of the transformation works for several problems. 
However, they lead to new cycle (non-)termination proofs in very rare cases.

\section{Conclusions}
\label{secconcl}
We presented techniques to prove termination and relative termination for cycle rewriting. The main approach is
to apply a sound and complete transformation from cycle into string rewriting.
We presented and analyzed three such transformations, both for termination and relative termination. 
Apart from that we provided a framework covering several variants of matrix interpretations serving for
proving (relative) cycle termination.
Our implementations and the corresponding experimental results show that both techniques are useful in the sense that 
they apply for several examples for which the earlier techniques failed.

Together with the sound and complete transformation $\phi$ in the reverse direction from \cite{ZKB14}, the existence of a sound and complete transformation like 
$\sym{split}$ implies that the problems of cycle termination and string termination of SRSs are equivalent in a strong sense. For instance, 
it implies that they are in the same level of the arithmetic hierarchy, which is $\Pi^0_2$-complete along the lines of \cite{EGSZ11}. 
Alternatively, $\Pi^0_2$-completeness of cycle termination can be concluded from the sound and complete transformation $\phi$ combined with the observation that cycle termination is in $\Pi^0_2$.

For future research, there is a need for more sophisticated techniques to prove cycle non-termination,
since we conjecture that {e.g.} several open problems in the benchmark sets are cycle non-terminating,
but our tools are not able to find a corresponding proof. A promising non-termination technique may follow
the ideas of \cite{endrullis-zantema:15}, some first steps in this direction have been worked out in \cite{zantema-fedotov:16}.

For a fair comparison of tools for (non-)termination of cycle rewriting, a set of benchmarks is needed with more focus on cycle rewriting, 
rather than the current set in which nearly all systems are copied from benchmarks for string rewriting.

\section*{Acknowledgments}
We thank Johannes Waldmann for fruitful remarks, in particular for his suggestions leading to Section \ref{secmatr} on trace decreasing matrix interpretations.
We also thank the anonymous reviewers of RTA 2015 for their valuable comments on the topic,
and the anonymous reviewers of this journal version for their very careful reading and valuable suggestions and comments.

\bibliographystyle{plain}% the recommended bibstyle

\begin{thebibliography}{10}

\bibitem{aprove}
Homepage of {AProVE}, 2016.
\newblock \url{http://aprove.informatik.rwth-aachen.de}.

\bibitem{BKNZ15}
Harrie Jan~Sander Bruggink, Barbara K{\"{o}}nig, Dennis Nolte, and Hans
  Zantema.
\newblock Proving termination of graph transformation systems using weighted
  type graphs over semirings.
\newblock In Francesco Parisi{-}Presicce and Bernhard Westfechtel, editors,
  {\em Graph Transformation - 8th International Conference, {ICGT} 2015, Held
  as Part of {STAF} 2015, L'Aquila, Italy, July 21-23, 2015. Proceedings},
  volume 9151 of {\em Lecture Notes in Computer Science}, pages 52--68.
  Springer, 2015.

\bibitem{BKZ14}
Harrie Jan~Sander Bruggink, Barbara K{\"{o}}nig, and Hans Zantema.
\newblock Termination analysis for graph transformation systems.
\newblock In Josep Diaz, Ivan Lanese, and Davide Sangiorgi, editors, {\em Proc.
  8th {IFIP} International Conference on Theoretical Computer Science}, volume
  8705 of {\em Lecture Notes in Comput. Sci.}, pages 179--194. Springer, 2014.

\bibitem{CFFGW:12}
Michael Codish, Yoav Fekete, Carsten Fuhs, J{\"{u}}rgen Giesl, and Johannes
  Waldmann.
\newblock Exotic semi-ring constraints.
\newblock In Pascal Fontaine and Amit Goel, editors, {\em 10th International
  Workshop on Satisfiability Modulo Theories (SMT 2012)}, volume~20 of {\em
  EPiC Series}, pages 88--97. EasyChair, 2012.

\bibitem{D14}
Bruno Dutertre.
\newblock Yices 2.2.
\newblock In Armin Biere and Roderick Bloem, editors, {\em Computer-Aided
  Verification (CAV'2014)}, volume 8559 of {\em Lecture Notes in Comput. Sci.},
  pages 737--744. Springer, 2014.

\bibitem{EGSZ11}
J{\"{o}}rg Endrullis, Herman Geuvers, Jakob~Grue Simonsen, and Hans Zantema.
\newblock Levels of undecidability in rewriting.
\newblock {\em Inf. Comput.}, 209(2):227--245, 2011.

\bibitem{EWZ08}
J{\"{o}}rg Endrullis, Johannes Waldmann, and Hans Zantema.
\newblock Matrix interpretations for proving termination of term rewriting.
\newblock {\em J. Autom. Reasoning}, 40(2-3):195--220, 2008.

\bibitem{endrullis-zantema:15}
J{\"{o}}rg Endrullis and Hans Zantema.
\newblock Proving non-termination by finite automata.
\newblock In Maribel Fern{\'{a}}ndez, editor, {\em 26th International
  Conference on Rewriting Techniques and Applications (RTA 2015)}, volume~36 of
  {\em LIPIcs}, pages 160--176. Schloss Dagstuhl - Leibniz-Zentrum fuer
  Informatik, 2015.

\bibitem{geser1990}
Alfons Geser.
\newblock {\em Relative Termination}.
\newblock Dissertation, Universit\"at Passau, Germany, 1990.

\bibitem{FBEFFOPSKSST:14}
J{\"{u}}rgen Giesl, Marc Brockschmidt, Fabian Emmes, Florian Frohn, Carsten
  Fuhs, Carsten Otto, Martin Pl{\"{u}}cker, Peter Schneider{-}Kamp, Thomas
  Str{\"{o}}der, Stephanie Swiderski, and Ren{\'{e}} Thiemann.
\newblock Proving termination of programs automatically with {AProVE}.
\newblock In St{\'{e}}phane Demri, Deepak Kapur, and Christoph Weidenbach,
  editors, {\em Proc. 7th International Joint Conference on Automated Reasoning
  (IJCAR'14)}, volume 8562 of {\em Lecture Notes in Comput. Sci.}, pages
  184--191. Springer, 2014.

\bibitem{GMRTW:15}
J{\"{u}}rgen Giesl, Fr{\'{e}}d{\'{e}}ric Mesnard, Albert Rubio, Ren{\'{e}}
  Thiemann, and Johannes Waldmann.
\newblock Termination competition ({{termCOMP 2015}}).
\newblock In Amy~P. Felty and Aart Middeldorp, editors, {\em 25th International
  Conference on Automated Deduction (CADE 2015)}, volume 9195 of {\em Lecture
  Notes in Comput. Sci.}, pages 105--108. Springer, 2015.

\bibitem{GM04}
J\"urgen Giesl and Aart Middeldorp.
\newblock Transformation techniques for context-sensitive rewrite systems.
\newblock {\em J. Funct. Program.}, 14(4):379--427, 2004.

\bibitem{GZ03}
J{\"{u}}rgen Giesl and Hans Zantema.
\newblock Liveness in rewriting.
\newblock In Robert Nieuwenhuis, editor, {\em Proc. 14th Conference on
  Rewriting Techniques and Applications (RTA)}, volume 2706 of {\em Lecture
  Notes in Comput. Sci.}, pages 321--336. Springer, 2003.

\bibitem{HW06}
Dieter Hofbauer and Johannes Waldmann.
\newblock Termination of \{\emph{a}\emph{a}-{\textgreater}\emph{b}\emph{c},
  \emph{b}\emph{b}-{\textgreater}\emph{a}\emph{c},
  \emph{c}\emph{c}-{\textgreater}\emph{a}\emph{b}\}.
\newblock {\em Inf. Process. Lett.}, 98(4):156--158, 2006.

\bibitem{HW06-rta}
Dieter Hofbauer and Johannes Waldmann.
\newblock Termination of string rewriting with matrix interpretations.
\newblock In Frank Pfenning, editor, {\em Proc. 17th Conference on Rewriting
  Techniques and Applications (RTA)}, volume 4098 of {\em Lecture Notes in
  Comput. Sci.}, pages 328--342. Springer, 2006.

\bibitem{KW08}
Adam Koprowski and Johannes Waldmann.
\newblock Arctic termination ...below zero.
\newblock In Andrei Voronkov, editor, {\em Proc. 19th Conference on Rewriting
  Techniques and Applications (RTA)}, volume 5117 of {\em Lecture Notes in
  Comput. Sci.}, pages 202--216. Springer, 2008.

\bibitem{KSZM09}
Martin Korp, Christian Sternagel, Harald Zankl, and Aart Middeldorp.
\newblock Tyrolean termination tool 2.
\newblock In Ralf Treinen, editor, {\em Proc. 20th Conference on Rewriting
  Techniques and Applications (RTA)}, volume 5595 of {\em Lecture Notes in
  Comput. Sci.}, pages 295--304. Springer, 2009.

\bibitem{sabel-zantema:15}
David Sabel and Hans Zantema.
\newblock {Transforming Cycle Rewriting into String Rewriting}.
\newblock In Maribel Fern{\'a}ndez, editor, {\em 26th International Conference
  on Rewriting Techniques and Applications (RTA 2015)}, volume~36 of {\em
  Leibniz International Proceedings in Informatics (LIPIcs)}, pages 285--300,
  Dagstuhl, Germany, 2015. Schloss Dagstuhl--Leibniz-Zentrum fuer Informatik.

\bibitem{S88}
Imre Simon.
\newblock Recognizable sets with multiplicities in the tropical semiring.
\newblock In {\em Mathematical Foundations of Computer Science}, volume 324 of
  {\em LNCS}, pages 107--120. Springer, 1988.

\bibitem{starexec}
Starexec, 2016.
\newblock \url{http://www.starexec.org}.

\bibitem{termcomp}
{Termination Competition 2015}, 2015.
\newblock
  \url{http://termination-portal.org/wiki/Termination_Competition_2015}.

\bibitem{termcomp16}
{Termination and Complexity Competition 2016}, 2016.
\newblock
  \url{http://termination-portal.org/wiki/Termination_and_Complexity_Competition_2016}.

\bibitem{tpdb}
The termination problem data base, 2015.
\newblock \url{http://termination-portal.org/wiki/TPDB}.

\bibitem{ttt2}
Homepage of {\TTTT}, 2016.
\newblock \url{http://cl-informatik.uibk.ac.at/software/ttt2/}.

\bibitem{yices}
Homepage of {Yices}, 2016.
\newblock \url{http://yices.csl.sri.com/}.

\bibitem{zantema:94}
Hans Zantema.
\newblock Termination of term rewriting: Interpretation and type elimination.
\newblock {\em J. Symb. Comput.}, 17(1):23--50, 1994.

\bibitem{zantema-fedotov:16}
Hans Zantema and Alexander Fedotov.
\newblock Non-termination of string and cycle rewriting by automata.
\newblock In Aart Middeldorp and Ren{\'{e}} Thiemann, editors, {\em Proceedings
  of the 15th International Workshop on Termination}, pages 13:1--13:5, 2016.
\newblock \url{http://cl-informatik.uibk.ac.at/workspace/events/wst2016.pdf}.

\bibitem{ZKB14}
Hans Zantema, Barbara K{\"o}nig, and Harrie Jan~Sander Bruggink.
\newblock Termination of cycle rewriting.
\newblock In Gilles Dowek, editor, {\em Proc. Joint 25th Conference on
  Rewriting Techniques and Applications and 12th Conference on Typed Lambda
  Calculi and Applications (RTATLCA)}, volume 8560 of {\em Lecture Notes in
  Comput. Sci.}, pages 476--490. Springer, 2014.

\end{thebibliography}

\end{document}